\documentclass[10pt]{article}
\usepackage[hidelinks]{hyperref}
\usepackage[utf8]{inputenc}
\usepackage{titlesec}
\usepackage{bm}
\usepackage{my/setup}
\usepackage{mathtools}
\usepackage{enumitem}
\usepackage{comment}
\usepackage{xcolor}
\usepackage{stmaryrd}
\usepackage{natbib}
\usepackage{float}
\usepackage{soul}
\usepackage{xfrac}
\usepackage{wrapfig}
\usepackage{authblk}
\usepackage[font=small,labelfont=bf]{caption}

\def\*#1{\boldsymbol{#1}}
\def\1#1{\mathcal{#1}}
\def\2#1{\mathscr{#1}}
\def\3#1{\mathbb{#1}}

\setlist[itemize]{noitemsep}
\setlist[enumerate]{noitemsep}

\usepackage[margin=1.25in]{geometry} %TODO: revert to 1in?
\title{Toward a Complete Criterion for Value of Information in Insoluble Decision Problems}
\author[1]{Ryan Carey} 
\author[2]{Sanghack Lee}
\author[1]{Robin J.~Evans}
\affil[1]{University of Oxford}
\affil[2]{Seoul National University}

\date{}

\newcommand*{\Eps}{{\mathcal{E}}}
\newcommand*{\eps}{{\varepsilon}}
\newcommand*{\calS}{{\mathcal{S}}}
\newcommand*{\calG}{{\mathcal{G}}}
\newcommand*{\calH}{{\mathcal{H}}}
\newcommand*{\calM}{{\mathcal{M}}}
\newcommand*{\Pa}{{\text{Pa}}}
\newcommand*{\pa}{{\text{pa}}}

\newcommand*{\Ch}{{\text{Ch}}}
\newcommand*{\Anc}{{\text{Anc}}}
\newcommand*{\Desc}{{\text{Desc}}}

\newcommand*{\Unif}{{\text{U}}}
\newcommand{\unif}{\text{U}}

\newcommand*{\sr}{{\bm r}}
\newcommand*{\sW}{{\bm W}}
\newcommand*{\sw}{{\bm w}}

\newcommand*{\spi}{{\bm \pi}}

\newcommand*{\sV}{{\bm V}}
\newcommand*{\sv}{{\bm v}}

\renewcommand*{\sp}{{\bm p}}

\newcommand*{\sT}{{\bm T}}

\newcommand*{\sX}{{\bm X}}
\newcommand*{\sx}{{\bm x}}
\newcommand*{\sY}{{\bm Y}}

\newcommand*{\sZ}{{\bm Z}}
\newcommand*{\sz}{{\bm z}}
\newcommand*{\sC}{{\bm C}}
\renewcommand*{\sc}{{\bm c}}

\newcommand*{\sU}{{\bm U}}
\newcommand*{\su}{{\bm u}}
\newcommand*{\sN}{{\bm N}}

\newcommand*{\imin}{{i_\text{min}}}
\newcommand*{\imax}{{i_\text{max}}}
\newcommand*{\jmax}{J}
\newcommand*{\idiff}{{\imax-\imin+1}}
\newcommand*{\dsep}{\perp}

\newcommand{\upathto}{\;\hbox{-\hspace{-0.19em} -\hspace{-0.19em} -}\;}

\newcommand{\pathto}{\dashrightarrow}
\newcommand{\pathfrom}{\dashleftarrow}

\newcommand{\bool}{\mathbb{B}}
\newcommand{\EE}{\mathbb{E}}
\newcommand{\dom}[2][]{#1\mathfrak{X}_{#2}}
\DeclareMathOperator{\aseq}{{\overset{\text{a.s.}}{=\joinrel=}}}

\newcommand{\condset}[1]{\lceil (\sX(\calS) \cup C_{\sX(\calS) \setminus #1}) \setminus #1 \rceil}

\DeclareMathOperator{\doo}{do}
\DeclareMathOperator{\cat}{\times}

 \newcommand{\slee}[1]{\textcolor{red}{}}
 \newcommand{\ryan}[1]{\textcolor{brown}{}}
 \newcommand{\robin}[1]{\textcolor{blue}{}}

\newtheorem*{theorem*}{Theorem}
\newtheorem*{lemma*}{Lemma}
\newcommand{\braces}[2][]{#1\{#2 #1\}}

\newcommand{\parens}[2][]{#1(#2 #1)}

\newcommand{\angles}[2][]{#1\langle#2 #1\rangle}

\newcommand{\undir}{-\!\!\!\!\!\!-\,}

\newcommand{\set}[2][]{\braces[#1]{#2}}

\newcommand{\tuple}[2][]{\angles[#1]{#2}}

\newcommand{\dotcup}{\mathrel{\dot{\cup}}}

\newtheorem{innercustomlem}{Lemma}
\newenvironment{customlem}[1]
  {\renewcommand\theinnercustomlem{#1}\innercustomlem}
  {\endinnercustomlem}

\newtheorem{innercustomthmalt}{Theorem}
\newenvironment{customthmalt}[1]
  {\renewcommand\theinnercustomthmalt{#1}\innercustomthmalt}
  {\endinnercustomthmalt}

\newif\ifthesis
\thesisfalse

\begin{document}

\maketitle

\section*{Abstract}
In a decision problem, observations are said to be material
if they must be taken into account to perform optimally.
Decision problems have an underlying (graphical) causal structure, 
which may sometimes be used to evaluate certain observations as immaterial.
%(whereas for other causal structures, materiality cannot be ruled out).
For soluble graphs --- ones where important past observations are remembered ---
there is a complete graphical criterion; one that rules out materiality 
whenever this can be done on the basis of the graphical structure alone.
In this work, we analyse a proposed criterion for insoluble graphs.
In particular, we prove that some of the conditions used to prove immateriality 
are necessary; when they are not satisfied, materiality is possible.
We discuss possible avenues and obstacles to proving necessity of the remaining conditions.

%\setcounter{tocdepth}{2}
%\tableofcontents

%\href{https://jamboard.google.com/d/1qGKj7XCn5pcC5EqOj7hQJoR7p5vD_QlF8cCwNFcI7kQ/viewer?ts=6142f910&pli=1&f=0}{\textcolor{blue}{Jamboard}}\quad
%\href{https://docs.google.com/presentation/d/183hW-q47Nh4tPW9978uVsCmSBjZjQn6BzX3mN_Tms7w/edit?usp=sharing}{\textcolor{blue}{Google slides}}

\section{Introduction}
We can view any decision problem as having an underlying causal structure --- 
a graph consisting of chance events, decisions and outcomes, and their causal relationships.
Sometimes, it is possible to evaluate key aspects of a decision problem from its causal structure alone.
For example, in \Cref{fig:linear-no-voi} and \Cref{fig:yes-voi}, we see two such causal structures.
For now, let us focus on the three endogenous vertices:
the observation $Z$, 
the decision (chosen by the decision-maker) $X$,
and the downstream outcome $Y$.
In each graph, $Z$ has an effect on $X$, which affects $Y$, 
but in \Cref{fig:yes-voi}, $Z$ also directly influences $Y$,
whereas in \Cref{fig:linear-no-voi}, it does not.

To fully describe a decision problem, we must specify probability distributions for 
each of the non-decision variables ---
distributions that must be compatible with the graphical structure.
In particular, the distribution for any variable must depend only on its direct causes, 
i.e.\ its parents, a condition known as Markov compatibility.
For example, in the causal structure shown in \Cref{fig:yes-voi}, 
one compatible decision problem is shown in the figure.
The variable $Z$ is a Bernoulli trial (i.e.\ a coin flip),
and the decision-maker is rewarded with $Y=1$ if they state the outcome of $Z$ 
(i.e.\ call the outcome of the coin flip), otherwise the reward is $Y=0$.
A variable is then said to be material if the attainable reward is greater given access to 
an observation than without it.
For example, by observing $Z$, the decision-maker can obtain a reward of $1$, such as with the policy $Y=Z$.
Without observing $Z$, any policy will achieve a reward of $0.5$.
The means that the value of information is $1-0.5=0.5$, and since this quantity is strictly 
positive, $Z$ is material.

For the causal structure shown in \Cref{fig:linear-no-voi}, we can instead make a deduction that applies 
to \emph{any} decision problem compatible with the graph.
In this case, for any such decision problem,
there will exist an optimal decision rule that ignores the value of $Z=z$ entirely.
One way to see this is that once a decision $X=x$ is chosen, 
the observation $Z$ becomes independent of $Y$, 
and so there is no reason for the decision to depend on it.
(This can be proved from the fact that $Z$ is d-separated from $Y$ given $X$.)
So for any decision problem compatible with this graph, $Z$ is immaterial.

There are many reasons that we may want to evaluate whether a causal structure 
allows an observation such as $Z$ to be material.
Firstly, for algorithmic efficiency ---
if an observed variable is immaterial, 
then the optimal policies are contained in a small subset of all available policies, 
that we can search exponentially more quickly.
(For example, in \Cref{fig:linear-no-voi}, there are two choices for $X$, 
but there are four deterministic mappings from $Z$ to $X$.)

Secondly, materiality can have implications regarding the fairness of a decision-making procedure.
Suppose that $Z$ designates the gender of candidates available to a recruiter, 
which are male $Z=1$ or female $Z=0$ with equal probability, 
while $X$ indicates whether that person is $X=1$ or is not $X=0$ recruited, 
and $Y$ indicates whether that person is $Y=1$ or is not $Y=0$ hired.
If $Y$ is correlated with $Z$ given $X$, then the applicant's gender is material for the recruiter,  
and to maximise the hiring probability, they will have to recruit applicants at different 
rates based on their gender.
If the causal structure is that of \Cref{fig:linear-no-voi}, then materiality can be ruled out, 
meaning that unfair behaviour is not necessary for optimal performance, 
whereas the causal structure of \Cref{fig:yes-voi} can incentivise unfairness.
Such an analyses can be used for well-studied concepts like counterfactual fairness \citep{Kusner2017}.
An arbitrary graph where $Z$ is a sensitive variable (such as gender),
counterfactual fairness can arise only when there is a path $Z \to\ldots\to O \to X$, where the observation $O$ is material \citep{everitt2021agent}.

Thirdly, materiality can have implications for AI safety ---
if $Z$ represents a corrective instruction from a human overseer,
and there exists no path $Z \to \ldots \to O \to X$ where $O$ is material, 
then there exist optimal policies that ignore this instruction \citep{everitt2021agent}.%
Materiality is also relevant for evaluations of agents' intent \citep{halpern2018towards,ward2024reasons}, 
and relatedly, their incentives to control parts of the environment \citep{everitt2021agent,farquhar2022path}.
For an agent to intentionally manipulate a variable $Z$ 
to obtain an outcome $Y=y$, there must be a path 
$p:X \to \ldots \to Z \to\ldots\to Y$ 
where for each of its decisions $X'$ lying on $p$, 
the parent $O'$ along $p$ is material for $X'$.
In general, a stronger criterion for ruling out materiality will allow us to 
rule out unfair or unsafe behaviour for a wider range of agent-environment interactions \citep{everitt2021agent}.

\begin{figure}[ht]
    \centering
\hfill
\begin{subfigure}[b]{0.23\textwidth}\centering
\begin{tikzpicture}[rv/.style={circle, draw, thick, minimum size=6mm, inner sep=0.8mm}, node distance=15mm, >=stealth]
  \node (Z) [rv]  {$Z$};
  \node (X) [rv,red,below = 10mm of Z] {$X$};
  \node (Y) [rv, right = 10mm of X]  {$Y$};
  \node (dummy) [opacity=0, rv, right = 10mm of X,label=below:{$\vphantom{y=\llbracket z=x \rrbracket}$}]  {$Y$};

  \draw[->, very thick] (Z) -- (X);
  \draw[->, very thick] (X) -- (Y);
%  \draw[->, very thick] (Z) -- (Y);
\end{tikzpicture}
\caption{$Z$ is immaterial for $X$.} \label{fig:linear-no-voi}
\end{subfigure}\hfill\begin{subfigure}[b]{0.23\textwidth}\centering
\begin{tikzpicture}[rv/.style={circle, draw, thick, minimum size=6mm, inner sep=0mm}, node distance=15mm, >=stealth]
  \node (Z) [rv, label=right:{$Z =\Eps_Z$}]  {$Z$};
  \node (EZ) [rv, above = 5mm of Z, label=right:$\Eps_Z \sim U(\mathbb{B})$, gray]  {$\Eps_Z$};
  \node (X) [rv,red,below = 10mm of Z, label=below:{$X\!=\!Z$}] {$X$};
  \node (Y) [rv, right = 10mm of X,label={[xshift=2mm,yshift=-12.4mm]$f_Y(\!\pa_y\!)\!=\!\llbracket \!z\!=\!x \!\rrbracket$}]  {$Y$};
  
  \draw[->, very thick, gray] (EZ) -- (Z);
  \draw[->, very thick] (Z) -- (X);
  \draw[->, very thick] (X) -- (Y);
  \draw[->, very thick] (Z) -- (Y);
\end{tikzpicture}
\caption{$Z$ is material for $X$.}\label{fig:yes-voi}
\end{subfigure}
\hfill
\begin{subfigure}[b]{0.35\textwidth}\centering
\begin{tikzpicture}[rv/.style={circle, draw, thick, minimum size=6mm, inner sep=0mm}, node distance=15mm, >=stealth]
  \node (Z) [rv, label=right:{$Z =\Eps_Z$}]  {$Z$};
  \node (EZ) [rv, above = 5mm of Z, label=right:$\Eps_Z \sim U(\mathbb{B})$, gray]  {$\Eps_Z$};
  \node (X) [rv,red,below = 10mm of Z, label=below:{$X\!=\!Z$}] {$X$};
  \node (Xp) [rv,red,right = 7mm of X, label=below:{$X'\!=\!X$}] {$X'$};
  \node (Y) [rv, right = 10mm of Xp,label={[xshift=4.5mm,yshift=-12.4mm]$f_Y(\!\pa_y\!)\!=\!\llbracket \!z\!=\!x \!\rrbracket$}]  {$Y$};
  
  \draw[->, very thick] (Z) -- (X);
  \draw[->, very thick, gray] (EZ) -- (Z);
  \draw[->, very thick] (X) -- (Xp);
  \draw[->, very thick] (Xp) -- (Y);
  \draw[->, very thick] (Z) -- (Y);
\end{tikzpicture}
\caption{$Z$ is material for $X$.}\label{fig:yes-voi-no-sr}
\end{subfigure}\hfill\null
\caption[Materiality examples]{Three graphs, with decisions in red, and a real-valued outcome $Y$. We write $U(\bool)$ for a uniform distribution over $\bool$, i.e.\ a Bernoulli distribution with $p=0.5$.} \label{fig:intro-voi}
\end{figure}
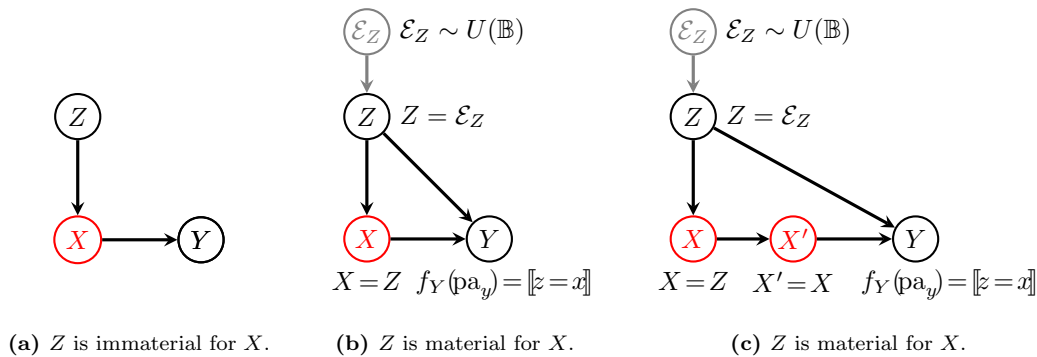

Any procedure for establishing immateriality based on the causal structure
may be called a \emph{graphical criterion}. 
For example, 
if a decision $X$ is not an ancestor of the outcome $Y$, then all of the variables observed at $X$ are immaterial.
An ideal graphical criterion would be proved \emph{complete}, 
in that it can establish immateriality whenever this is possible from the graphical structure alone.
Clearly, this criterion is not complete, because in \Cref{fig:linear-no-voi}, $X$ is an ancestor of the outcome, 
but we still proved $Z$ immaterial.
So far, a graphical criterion from \citet{van2022complete} has been proved complete, 
but only under some significant restrictions.
The causal structure must be \emph{soluble}, meaning that
all of the important information observed from past decisions is remembered at later decision points.
Also, no criteria has been proved complete for identifying immaterial decisions, i.e.\ 
past decisions that can be safely forgotten.

For insoluble graphs, there the criterion of \citet[Thm. 2]{lee2020characterizing}, 
which can identify immaterial decisions and is (strictly) more potent in general.
However, it is not yet known whether this criterion is complete.
In particular, it is not yet clear whether several of its conditions are necessary.
For example, one case where all existing criteria are silent is the simple graph shown in \Cref{fig:yes-voi-no-sr} --- 
we would like to know whether we can rule out $X$ being a material observation for $X'$.
We cannot use \citet{van2022complete} because $X$ is a decision, and because the graph is insoluble.%
\footnote{Formally, this is because $W \not \dsep Y \mid X \cup X'$, and $X' \not \dsep Y \mid X \cup W$, as per the definition of solubility 
that we will review in \Cref{sec:past-work}.}
Furthermore, we cannot establish immateriality using \citet[Thm. 2]{lee2020characterizing}, 
because it violates a property that we term
LB-factorizability, which we will discuss in \Cref{sec:past-lb}.%
\footnote{Specifically, requirement I of LB-factorizability is violated
because $Y$ is d-connected to $\spi_{X'}$ given $X'$.}

By studying \Cref{fig:yes-voi-no-sr} in a bespoke fashion, we find 
that there exists a decision problem with the given causal structure, 
where $X$ is material for $X'$.
As shown in \Cref{fig:yes-voi-no-sr}, $Z$ is a Bernoulli variable, 
and $Y$ is equal to $1$ if $Z=X'$ and to $0$ otherwise. 
If $X$ is observed by $X'$, then a reward of $\EE[Y]=1$ can be achieved by the policy $X'=X=Z$.
If $X$ is not observed, the greatest achievable reward is lower, at $\EE[Y]=0.5$, implying materiality.%
%\footnote{It should come as no surprise that $X$ may be material here, because if we discovered just one graph
%where an observation violated LB-factorizability, but could still be proved immaterial, then 
%the criterion of \citet[Thm. 2]{lee2020characterizing} would be incomplete, and we would need to find a new criterion.}

This raises a question: by generalising this construction,
can we prove that requirement I of LB-factorizability 
is necessary to prove immateriality for a wide class of graphs?
This work will prove that this requirement 
is indeed necessary, meaning that materiality cannot be excluded 
for a wide class of graphs including \Cref{fig:yes-voi-no-sr}.

It remains an open question whether the criterion of \citet[thm. 2]{lee2020characterizing} 
as a whole is complete, in that its other conditions are necessary for establishing immateriality.
In the case that it is complete, our work is a step toward proving this.
On the other hand, we also present some graphs where materiality is difficult to establish, 
that --- if the criterion is not complete --- could bring us closer to a proof of incompleteness.

The structure of the paper is as follows.
In \Cref{sec:setup}, we will recap the formalism used by \citet{lee2020characterizing} for modelling decision problems, 
based on structural causal models.
In \Cref{sec:past-work}, we will review existing procedures for proving that an observation can or cannot be material.
In \Cref{sec:main-result}, we will establish our main result: that requirement I of LB-factorizability is 
necessary to establish immateriality.
In \Cref{sec:further-steps}, we present some analogous results for other requirements of LB-factorizability, 
that could serve as a building block for proving the necessity of those requirements.
We then illustrate the problems that arise in trying to prove necessity of those further requirements, 
and outline some possible directions for further work.
Finally, in \Cref{sec:conclusion}, we conclude.

%\todo{We also know that conditions (1-3) are necessary in cases where each context belongs to only one decision.)}

%\textsuperscript{,}\footnote{Two further related goals that we may hope to fulfil (suggested by Sanghack) are: 
%to identify and study the complexity of an algorithm for applying Thm 2, and 
%to gain a deeper and more intuitive understanding of Thm 2.}
%
% SL: moved to footnotes 

\section{Setup} \label{sec:setup}
Our analysis will follow \citet{lee2020characterizing} by using the structural causal model (SCM) framework \citep[Chapter 7]{pearl2009causality},
although the results also apply equally to Bayesian networks and influence diagrams.

\subsection{Structural causal models} 
A structural causal model (SCM) $\calM$ is  a tuple $\tuple{\sU,\sV,P(\sU),\mathbf{F}}$, where
$\sU$ is a set of variables determined by factors outside the model, called \emph{endogenous}
following a joint distribution $P(\sU)$, and 
$\sV$ is a set of endogenous variables whose values are determined by a collection of functions $\mathbf{F} = \set{f_V}_{V\in\sV}$ such that $V \gets f_V(\Pa(V), \sU_V)$ where $\Pa(V) \subseteq \sV\setminus\set{V}$ is a set of 
endogenous variables and $\sU_V\subseteq \sU$ is a set of exogenous variables.
The observational distribution $P(\sv)$ is defined as 
$\sum_{\su} \prod_{V\in\sV} P(v|\mathbf{pa}_V,\su_V) P(\su)$, 
where $\su_V$ is the assignment $\su$ restricted to variables $\sU_V$.
Furthermore, $\doo(\sX=\sx)$ represents 
the operation
of fixing a set $\sX$ to a constant $\sx$ regardless of their original mechanisms. 
Such intervention induces a submodel $\calM_\sx$, 
which is $\calM$ with $f_X$ replaced by $x$ for $X\in\sX$.
Then, an interventional distribution %$P_\sx(\sv{\setminus}\sx)$ (or also 
$P(\sv{\setminus}\sx|\doo(\sx))$ can be computed as the observational distribution 
in $\calM_\sx$.
%, and can 
%equivalently be written as
%$P(\sv{\setminus}\sx \mid \doo(\sx)) = \sum_{\su} \prod_{V_i\in\sV{\setminus} \sX } P(v_i|\mathbf{pa}_i,\su_i) P(\su)$.
The induced graph of an SCM $\calM$ is a DAG $\calG$ on only the 
endogenous variables $\sV$ where 
(i) $X {\to} Y$ if $X$ is an argument of $f_Y$; and 
(ii) $X {\leftrightarrow} Y$ if $\sU_X$ and $\sU_Y$ are dependent, i.e.\ for any $\su_X,\su_Y$, $P(\su_X,\su_Y) \neq P(\su_X) \times P(\su_Y)$.

We use the notation $\Pa(X)$, $\Ch(X)$, $\Anc(X)$ and $\Desc(X)$ to represent the parents, children, ancestors and descendants of a variable $X$, respectively, and take  \label{pg:tcc:pachancdesc}
ancestors and descendants to include the node $X$ itself.%
\footnote{Note that $\Pa(X)$ is an intentional reuse of the notation used to describe the arguments of $f_X$ in the SCM definition, because the endogenous arguments of $f_X$ 
and the parents of $X$ in the induced graph are the same variables.}

We will use the notation $V_1 \undir V_2$ to designate an edge whose direction may be $V_1 \to V_2$ or $V_1 \gets V_2$.
For a path $V_1 \undir V_2 \undir \cdots \undir V_\ell$, we will use the shorthand $V_1 \upathto V_\ell$, 
and for a directed path $V_1 \to \cdots \to V_\ell$, the shorthand $V_1 \pathto V_\ell$.
For a path $p:A \upathto B \upathto C \upathto D$, 
we will describe the segment $B \upathto C$
using the shorthand $B \overset{p}{\upathto} C$.
We will use the shorthand $\sV_{1:N}$ for
a sequence of variables $V_1,\ldots V_N$ indexed by $1,\ldots,N$, 
$\sv_{1:N}$ for a sequence of assignments, and
$\sp_{1:N}$ for a set of paths $p_1,\ldots p_N$. \label{pp:seq-colon}

There is certain notation that we will use repeatedly when constructing causal models, 
such as tuples, bitstrings, indexing, and Iverson brackets.
We will write a tuple as $z := \langle x, y \rangle$, and this may be indexed as $z[0] = x$.
A bitstring of length $n$, i.e.\ a tuple of $n$ Booleans, may be written as $\bool^n$, 
and a uniform distribution over this space, as $U(\bool^n)$.
We will denote a bitwise XOR operation by $\oplus$ so that, for example, $01 \oplus 11 = 10$.
Bitstrings may also be used for indexing, for example, the $y$\textsuperscript{th} bit of $x$ may be written as as $x[y]$, 
and the leftmost bits are of higher-order so that, for example, $0100[01] = 1$. \label{pp:indexing}
Similarly, for random variables $X,Y$, we will write $X[Y]$ for a variable equal to $x[y]$ when $X=x$ and $Y=y$.
Finally, the Iverson bracket $\llbracket P \rrbracket$ is equal to $1$ if $P$ is true, and $0$ otherwise.

\subsection{Modelling decision problems}
To use an SCM to define a decision problem, 
we need to specify what policies the agent can select from
and what goal the agent is trying to achieve.

We will describe the set of available policies using
a Mixed Policy Scope (\citealp{lee2020characterizing}), which casts certain variables as decisions, and 
others as \emph{context variables} or ``observations'' $\sC_X$, 
that each decision $X$ is allowed to depend on.
Following \citet{lee2020characterizing}, we will consistently illustrate decision variables with red circles, 
as in \Cref{fig:intro-voi}.

\begin{definition}[Mixed Policy Scope (MPS)]\label{def:mixed-policy-structure}
      Given a DAG $\calG$ on vertices $\sV$, a \emph{mixed policy scope} $\calS = \langle X, \sC_X \rangle_{X \in \sX(\calS)}$ 
      consists of a set of decisions $\sX(\calS) \subseteq \sV$
      and a set of context variables $\sC_X \subseteq \sV$ for each decision.
\end{definition}
For a set of decisions $\sX'$, we define their contexts as $\sC_{\sX'} = \bigcup_{X \in \sX'} \sC_X$.

A policy consists of a
probability distribution for each decision $X$, conditional on its contexts $\sC_X$. \label{def:contexts}

\begin{definition}[Mixed Policy]\label{def:mixed-policy}
Given an SCM $\calM$ and scope $\calS=\tuple{X,\sC_X}$,
	%Given $\tuple{\calG, Y, \sX^\star, \sC^\star}$ and an SCM $\calM\sim\calG$ with $\dom{Y}\subseteq \mathbb{R}$,
	a \emph{mixed policy} $\spi$ (or a \emph{policy}, for short)
 contains for each $X$ a decision rule $\pi_{X \mid \sC_X}$,
	%$\set{\pi_{X|\sC_X}}_{\tuple{X,\sC_X}\in \calS}$,
	where $\pi_{X|\sC_X}: \dom{X}\times \dom{\sC_X} \mapsto [0,1]$ is a proper probability mapping.%
 \footnote{
 Following \citet{lee2020characterizing}, we term this a ``mixed policy'' due to its 
 including mixed strategies.
 Note that game theory also has a distinction between ``mixed'' policies, where 
 the decision rules share a source of randomness, and ``behavioural'' policies, where 
 they do not, and in this sense, the ``mixed'' policies of \citet{lee2020characterizing} 
 are actually \emph{behavioural}.
 }
\end{definition}
We will say that such a policy $\spi$ \emph{follows} the scope $\calS$, written $\spi \sim \calS$.
A mixed policy is said to be \emph{deterministic} if every decision is a deterministic function of its contexts.

Once a policy is selected, we would have a new causal structure, described by a \emph{scoped graph}.

\begin{definition}[Scoped graph] \label{def:scoped-graph}
The \emph{scoped graph} $\calG_{\calS}$ is obtained by $\calG$,
by replacing, for each decision $X \in \sX(\calS)$, 
all inbound edges to $X$ with edges $C \to X$ for every $C \in \sC_X$.
We only consider scopes for which $\calG_{\calS}$ is acyclic.
\end{definition}

%We will say that a scope $\calS$i
%$\sX(\calS) = \bigcup_{\tuple{X,\sC_X} \in \calS} X$, and 
%contexts $\sC(\calS)=\bigcup_{X \in \sX(\calS)} \sC_X$.

We will designate one real-valued variable $Y \not \in \sX(\calS) \cup \sC(\calS)$ as the outcome node (also called the ``utility'' variable).
To calculate the expected utility
under a policy $\spi \sim \calS$, 
let $\sC^- = \left(\bigcup_{X \in \sX(\calS)} \sC_X \right)\setminus \sX(\calS)$ be the \textit{non-action} contexts.
Then, the expected utility is: \\ %, \xadd{with $\xx$ simply denoting the value of $\XX(\SSS)$,}
%
%
%\begin{align}\label{eq:expected-reward}
$\mu_{\spi,\calS} = \sum_{y,\sx,\sc^-} y P_\sx(y,\sc^-) \prod_{X\in \sX(\calS)} \pi(x|\sc_x).$
%\end{align}
When the scope is obvious, we will simply write $\mu_{\spi}$.
% The Maximum Expected Utility (MEU) is defined as $\max_{\spi \sim \calS} \mu_{\spi,\calS}$.

This paper is concerned with materiality --- whether removing one context variable from one decision will decrease the expected utility attainable by the best policy.
We define it in terms of the value of information \citep{howard1990influence,everitt2021agent}.

\begin{definition}[Value of Information] \label{def:voi}
Given an SCM $\calM$ and scope $\calS$, the \emph{maximum expected utility} (MEU) is $\mu^*_\calS = \max_{\spi \sim \calS} \mu_{\spi,\calS}$.
The \emph{value of information} (VoI) of context $Z \in \sC_X$ for decision $X \in \sX(\calS)$ is $\mu^*_\calS - \mu^*_{\calS_{Z \not \to X}}$, 
where $\calS_{Z \not \to X}$ is defined as $\tuple{X',\sC_{X'}}_{X' \in \sX(\calS) \setminus \{X\}} \cup \tuple{X,\sC_X \setminus \{Z\}}$.

The context $Z$ is \emph{material} for $X$ in an SCM $\calM$
if $Z$ has strictly positive value of information for $X$, otherwise it is \emph{immaterial}.
\end{definition}

\subsection{Graphical criteria for independence}
Knowing when variables are independent is an important step in 
identifying immaterial contexts, as we will discuss in the next section.
So, we will make repeated use of d-separation, a graphical criterion that establishes the 
independence of variables in a graph.

\begin{definition}[d-separation; \citealp{Verma1988soundness}]\label{def:d-separation}
    A path $p$ is said to be d-separated by a set of nodes $\sZ$ if and only if:
    \begin{enumerate}
    \item $p$ contains a collider $X \to W \gets Y$ such that the middle node $W$ is not in $\sZ$ and no descendants of $W$ are in $\sZ$, or
    \item $p$ contains a chain $X \to W \to Y$ or fork $X \gets W \to Y$ where $W$ is in $\sZ$, or
    \item one or both of the endpoints of $p$ is in $\sZ$.
    \end{enumerate}
    A set $\sZ$ is said to d-separate $\sX$ from $\sY$,
written ${(\sX \dsep_\calG \sY \mid \sZ)}$, if and only if $\sZ$ d-separates every path
from a node in $\sX$ to a node in $\sY$. Sets that are not d-separated are
called d-connected, written $\sX \not \dsep_\calG \sY \mid \sZ$.~\looseness=-1
\end{definition}

When the graph is clear from context, we will write $\dsep$ in place of $\dsep_\calG$.
When 
%disjoint 
sets $\sX,\sW,\sZ$ satsify $\sX \dsep \sW \mid \sZ$
they are conditionally independent:
$P(\sX,\sW \mid \sZ)=P(\sX \mid \sZ)P(\sW \mid \sZ)$ \citep{Verma1988soundness}.

\ifthesis
\newcommand{\LBfactlines}{8}
\else
\newcommand{\LBfactlines}{11}
\fi
\begin{wrapfigure}[\LBfactlines]{r}{0.31\linewidth}
\centering
\begin{tikzpicture}[rv/.style={circle, draw, thick, minimum size=6mm, inner sep=0.8mm}, node distance=15mm, >=stealth]\centering
  \node (Z) [rv,red]  {$Z$};
  \node (Pi) [rv,gray,left = 10mm of Z,inner sep=0.2mm] {$\pi_X$};
  \node (X) [rv,red,below = 10mm of Z] {$X$};
  \node (Y) [rv, right = 10mm of X]  {$Y$};
  
  \draw[->, very thick] (Pi) -- (X);
  \draw[->, very thick] (Z) -- (X);
  \draw[->, very thick] (X) -- (Y);
  \draw[->, very thick] (Z) -- (Y);
\end{tikzpicture}
\caption[Immaterial decision example]{A graph where decisions $Z,X$ jointly determine the 
outcome $Y$. A policy node $\pi_X$ is shown, which decides the decision rule at $X$.} \label{fig:triangle-no-voi}
\end{wrapfigure}
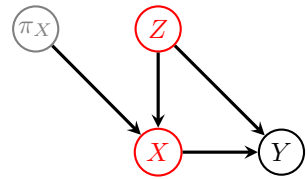

If we know that a deterministic mixed policy is being followed, then we may deduce further conditional independence relations.
This is because conditioning on variables $\sV$ may determine some decision variables, 
which are called ``implied'' \citep{lee2020characterizing}, or ``functionally determined'' \citep{Geiger1990completeness},
making them conditionally independent of other variables in the graph.

\begin{definition}[Implied variables; \citealp{lee2020characterizing}] \label{def:implied-vars}
To obtain the \emph{implied variables} $\lceil \sZ \rceil$ for variables $\sZ$ in $\calG$
given a mixed policy scope $\calS$,
begin with $\lceil \sZ \rceil \gets \sZ$,
then add to $\lceil \sZ \rceil$ every decision $X$ such that $C_X \subseteq \lceil \sZ \rceil$, until convergence.
\end{definition}

For example, in \Cref{fig:triangle-no-voi}, we see that $\lceil X \rceil=\{Z,X\}$,
so $Z$ is d-separated from $Y$ given $\lceil X \rceil$.
This means that under a deterministic mixed policy, $Z$ and $Y$ are statistically independent given $X$.
This has implications for materiality.
In particular, it means that the best deterministic mixed policy $Z=z,X=x$ does not need to observe 
$Z$ at $X$.
Moreover, the performance of the best deterministic mixed policy can never be surpassed by 
a stochastic policy (\citep[Proposition 1]{lee2020characterizing}), so $Z$ is immaterial.

\section{Review of graphical criteria for materiality} \label{sec:past-work}
We will now review some existing techniques for proving whether or not a graph is compatible with some variable $Z$ being material for some decision $X$.
\subsection{Single-decision settings} \label{sec:past-1dec-criteria}
In the single-decision setting,
there is a sound and complete criterion for materiality:
in a scoped graph $\calG(\calS)$, there exists an SCM
where the context $Z \in \sC_X$ is material
if and only if $Z \not \dsep Y \mid \sC_X \cup \{X\} \setminus \{Z\}$
and the outcome $Y$ is a descendant of $X$ \citep{lee2020characterizing,everitt2021agent}.
This statement can be split into proofs for the \emph{only if} and \emph{if} directions, 
both of which are relevant to the current paper.

The argument for the \emph{only if} is that
if $X$ is not an ancestor of the outcome $Y$, then its policy is completely irrelevant to the expected utility, and so all of its contexts are immaterial, 
and if $Z$ is conditionally independent of the outcome $Y$ given the decision and other observations, then it may be safely ignored without changing the outcome.
These arguments are important to us because they remain equally valid as we move to a 
multi-decision setting ---
a context must be an ancestor of $Y$ ,and must provide information about $Y$ over and above 
the other contexts, in order to be material.

The \emph{if} direction is proved by
constructing a decision problem where $Z$ is material.
By assumption, there is a directed path $X \pathto Y$, called 
the \emph{control path},
and a path $Z \upathto Y$, active given $C_X \cup \{X\} \setminus \{Z\}$, 
called the \emph{info path}.

In the SCM that is constructed, the variable $Z$ will contain information about $Y$
(due to a conditional dependency induced by the info path), 
and this will inform $X$ regarding how to influence $Y$
(using influence that is transmited along the control path).

The construction has two cases, which differ based on whether or not the info path 
contains colliders \citep{everitt2021agent,lee2020characterizing}.
For the case where it does not contain colliders, 
the graph and construction are shown in \Cref{fig:yes-voi-directed}.
(Note that when the info path is a directed path, we take this to be a special case where $V=Z$.)
The functions along the info path (dashed line) are chosen to copy $V$ to $\Pa_Y$ and to $Z$, and
$Y$ equals its maximum value of $1$ only if $X$ equals $V$, and $0$ otherwise.
So, $X$ must copy $Z$ to achieve the maximum expected utility.
Without the context $Z$, the maximum expected utility is $0.5$, proving materiality.%
\footnote{To be precise, the formalism of \citet{lee2020characterizing} also allows the active path from $Z$ to include one or more bidirected edges $V \leftrightarrow Y$, but to deal with these cases, 
we begin with the distribution that we would use for a path $V \gets L \to Y$, then marginalise out $L$.}

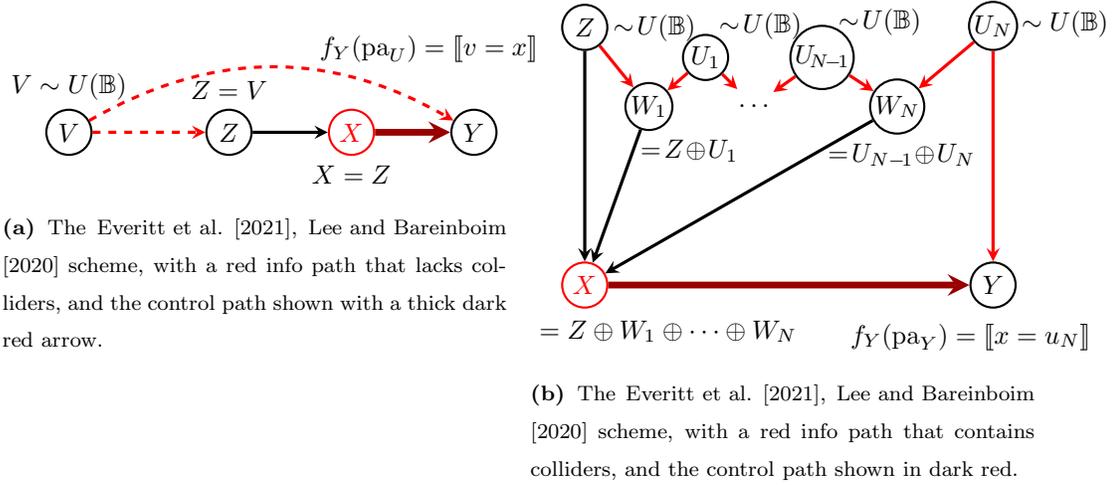
\begin{figure}[t]\centering
\begin{subfigure}[b]{0.44\textwidth}\centering
\begin{tikzpicture}[rv/.style={circle, draw, thick, minimum size=6mm, inner sep=0mm}, node distance=15mm, >=stealth]
  \node (V) [rv, label=$V \sim U(\mathbb{B})$]  {$V$};
  \node (Z) [rv, right = 15mm of V, label={$Z=V$}]  {$Z$};
  \node (X) [rv,red,right = 10mm of Z, label=below:{$X = Z$}] {$X$};
  \node (Y) [rv, right = 10mm of X,label={[xshift=-6mm,yshift=5mm]$f_Y(\pa_U)=\llbracket v=x \rrbracket$}]  {$Y$};
  
  \draw[->, very thick] (Z) -- (X);
  \draw[->, line width=0.9mm, red!60!black] (X) -- (Y);
  \draw[->, red, very thick, dashed] (V) -- (Z);
  \draw[->, red, very thick, dashed, bend left] (V) to (Y);
\end{tikzpicture}
\caption{The \citet{everitt2021agent,lee2020characterizing} scheme, with a red info path that lacks colliders, 
and the control path shown with a thick dark red arrow.}\label{fig:yes-voi-directed}
\end{subfigure}\hspace{2mm}
\begin{subfigure}[t]{0.44\textwidth}\centering
\begin{tikzpicture}[rv/.style={circle, draw, thick, minimum size=6mm, inner sep=0.2mm}, node distance=15mm, >=stealth]
  \node (X0) [rv,red, label={[xshift=11.1mm,yshift=-12mm]$=Z\oplus W_1 \oplus \cdots \oplus W_N$}]              {$X$};
  \node (Z0) [rv, above = 28mm of X0, label=0:{$\!\sim\! U(\bool)$}]  {$Z$};
  \node (Y) [rv, right=48mm of X0, label={[xshift=-3mm,yshift=-13mm,align=left]$f_Y(\pa_Y)=\llbracket x=u_N \rrbracket$}] {$Y$};
  \node (Vl) [rv, above = 28mm of Y, label={[xshift=9.6mm,yshift=-6mm]$\sim U(\bool)$}]  {$U_{N}$};
  \node (W1) [rv, below right = 6mm and 4mm of Z0, label=-80:{$\!\!\!\!\!=\!Z\!\oplus\! U_{1}$}]  {$W_{1}$};
  \node (V1) [rv, above right = 2mm and 3mm of W1, label={[xshift=7.6mm,yshift=-2mm]$\!\sim\! U(\bool)$}]  {$U_1$};
  \node (ldots) [rv, draw=none, below right = 2mm and 2mm of V1]  {$\ldots$};
  \node (Vlm1) [rv, right = 8mm of V1, label={[xshift=8.0mm,yshift=-2.6mm]$\!\sim\! U(\bool)$}]  {$U_{\!N\!-\!1}$};
  \node (Wl) [rv, right = 26mm of W1, label=-90:{$\,\,=\!U_{N-\!1}\!\oplus\! U_N$}]  {$W_{N}$};
  
  \draw[->, very thick] (Z0) -- (X0);
  \draw[->, very thick, red] (Z0) -- (W1);
  \draw[->, very thick] (W1) -- (X0);
  \draw[->, very thick, red] (V1) -- (W1);
  \draw[->, very thick, red] (V1) -- (ldots);
  \draw[->, very thick, red] (Vlm1) -- (ldots);
  \draw[->, very thick, red] (Vlm1) -- (Wl);
  \draw[->, very thick, red] (Vl) -- (Wl);
  \draw[->, very thick, red] (Vl) -- (Y);
  \draw[->, very thick] (Wl) -- (X0);
  \draw[->, line width=0.9mm, red!60!black] (X0) -- (Y);
\end{tikzpicture}
\caption{The \citet{everitt2021agent,lee2020characterizing} scheme, with a red info path that contains colliders, and the control path shown in dark red.} \label{fig:yes-voi-non-directed}
\end{subfigure}\hspace{3mm}
\caption[Past materiality constructions]{Three decision problems where $Z$ is material for $X$. 
For readability, we marginalise out exogenous variables from the SCM, so $z \sim U(\bool)$ can be understood as shorthand for $z=\eps_Z$ where $\eps_Z \sim U(\bool)$, and so on.
} \label{fig:problem-alternative-ways}
\end{figure}

\begin{figure}
\begin{subfigure}[b]{0.48\textwidth}\centering
\begin{tikzpicture}[rv/.style={circle, draw, thick, minimum size=6mm, inner sep=0mm}, node distance=15mm, >=stealth]
  \node (Z) [rv, label={[xshift=5mm,yshift=0mm]$Z \sim U(\mathbb{B})$}]  {$Z$};
  \node (X) [rv,red,below = 10mm of Z, label={[xshift=6mm,yshift=-9mm]$=Z$}] {$X$};
  \node (Zp) [rv,right = 13mm of X, label={[xshift=6mm,yshift=-9mm]$=X$}] {$Z'$};
  \node (Wp) [rv,below right = 7mm and 2mm of Zp] {$W'$};
  \node (Vp) [rv,below = 18mm of X, xshift=8mm] {$U'$};
  \node (Xp) [rv,red,right = 16mm of Zp, label={[xshift=6mm,yshift=-9mm]{$=Z'$}}] {$X'$};
  \node (Y) [rv, right = 17mm of Xp,label={[xshift=-6mm,yshift=1.5mm,align=left]$f_Y(\pa_Y)={\color{red}\llbracket z=x' \rrbracket}$}] {$Y$};

  \draw[->, very thick] (Z) -- (X);
  \draw[->, very thick] (Z) -- (Xp);
  \draw[->, very thick] (X) -- (Zp);
  \draw[->, very thick] (Zp) -- (Xp);
  \draw[->, very thick] (Zp) -- (Wp);
  \draw[->, very thick] (Vp) -- (Wp);
  \draw[->, very thick, color=black] (Wp) -- (Xp);
  \draw[->, very thick] (Vp) to [bend right=28] (Y);
  \draw[->, very thick] (Xp) -- (Y);
  \draw[->, very thick,red] (Z) -- (Y);
\end{tikzpicture}
\caption{The \citet{everitt2021agent} scheme is applied using just the red info path; $Z$ is immaterial for $X$.} \label{fig:soluble-everitt}
\end{subfigure}\hspace{4mm}
\begin{subfigure}[b]{0.48\textwidth}\centering
\begin{tikzpicture}[rv/.style={circle, draw, thick, minimum size=6mm, inner sep=0mm}, node distance=15mm, >=stealth]
  \node (Z) [rv, label={[xshift=5mm,yshift=0mm]$Z \sim U(\mathbb{B})$}]  {$Z$};
  \node (X) [rv,red,below = 10mm of Z, label={[xshift=6mm,yshift=-9mm]$=Z$}] {$X$};
  \node (Zp) [rv,right = 13mm of X, label={[xshift=6mm,yshift=-9mm]$=X$}] {$Z'$};
  \node (Wp) [rv,below right = 7mm and 2mm of Zp, label=right:{$=\langle Z',U'[Z'] \rangle$}] {$W'$};
  \node (Vp) [rv,below = 18mm of X, xshift=8mm, label={[xshift=0mm,yshift=1mm]$U' \sim \mathbb{B}^2$}] {$U'$};
  \node (Xp) [rv,red,right = 16mm of Zp, label={[xshift=6mm,yshift=-9mm]{$=W'$}}] {$X'$};
  \node (Y) [rv, right = 17mm of Xp,label={[xshift=-11mm,yshift=3.5mm,align=left]$f_Y(\pa_Y)={\color{red}\llbracket z=x'[0] \rrbracket}$ \\ $+ {\color{blue}\llbracket u'[x'[0]]=x'[1] \rrbracket}$}] {$Y$};

  \draw[->, very thick] (Z) -- (X);
  \draw[->, very thick] (Z) -- (Xp);
  \draw[->, very thick] (X) -- (Zp);
  \draw[->, very thick] (Zp) -- (Xp);
  \draw[->, very thick,blue] (Zp) -- (Wp);
  \draw[->, very thick,blue] (Vp) -- (Wp);
  \draw[->, very thick, color=black] (Wp) -- (Xp);
  \draw[->, very thick, color=blue] (Vp) to [bend right=28] (Y);
  \draw[->, very thick] (Xp) -- (Y);
  \draw[->, very thick,red] (Z) -- (Y);
\end{tikzpicture}
\caption{The \citet{van2022complete} scheme is applied, using the red and blue info paths; $Z$ is material for $X$.} \label{fig:van-merwijk-scheme}
\end{subfigure}
\caption[Materiality in soluble graphs]{Two decision problems on a soluble graph.} \label{fig:soluble-past}
\end{figure}
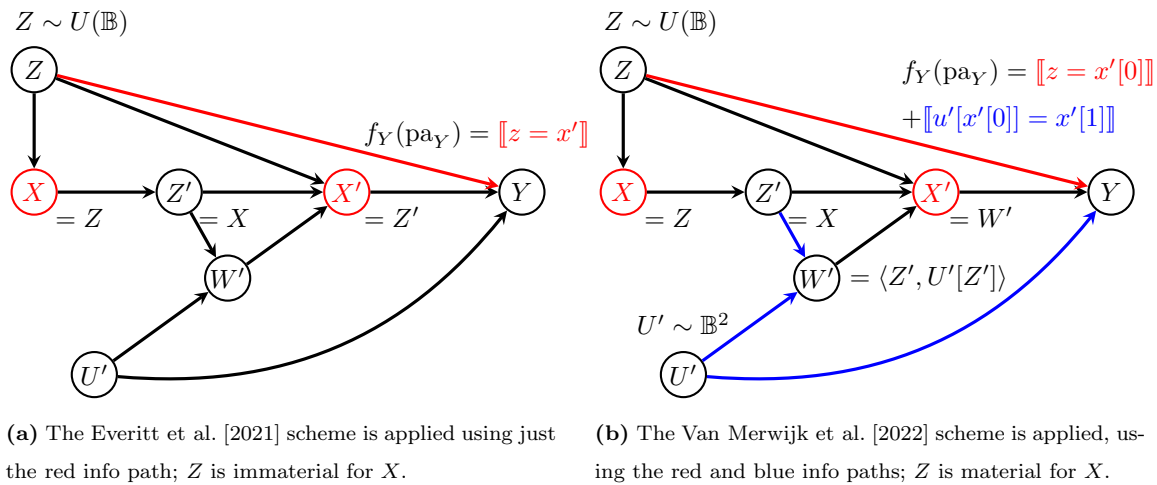

For the case where the info path does contains a collider, 
the graph and construction from \citet{everitt2021agent,lee2020characterizing}
are shown in \Cref{fig:yes-voi-non-directed}.
Each fork $U_i$ in the info path, along with $Z$, generates a random bit, 
while each collider $W_i$ is assigned the XOR ($U_{i-1}\oplus U_i$) of its two parents.
By observing $z$ and the values $\sw_{1:N}$, the agent has just enough information to recover $u_N$.
In particular, hte policy that sets $x$ equal to the XOR of $z$ and $\sw_{1:N}$, 
obtains $x=u_N$ and achieves the MEU, $\EE[Y]=1$.
Without the context $Z$, the MEU becomes $0.5$, so $Z$ is material.

\subsection{Soluble multi-decision settings} \label{sec:past-soluble}
This approach has been generalised to deal with multi-decision graphs 
that are \emph{soluble} (also known as graphs that respect ``sufficient recall'').

To recap, a graph is said to be soluble if there is an ordering 
$\prec \hspace{1mm} = \langle X_1,\ldots,X_N \rangle$
over decisions 
such that at for every $X_i$, 
for every previous decision or context
$V \in \{X_{j} \cup C_{X_j} \mid j \prec i \}$, 
we have 
$V \not \in \Anc(Y)$
or 
$V \dsep Y \mid \{X_i\} \cup C_{X_i}$.
That is, past decisions and contexts do not contain any information
that is relevant for a later decision, 
and unknown at the time that this later decision is made.
%(The case where $V$ is observed ($V \in C_{X_i}$) is a special case that satisfies 
%solubility, which implies that perfect recall satisfies solubility).
For example, in \Cref{fig:soluble-everitt}, using the ordering $X \prec X'$,
the nodes $Z,X$ are d-separated from $Y$ by $X'$ and its contexts $\{Z,Z',W\}$, 
which implies solubility.

For soluble graphs, there exists a complete criterion, for discerning whether a non-decision context $Z$ is material for a decision $X$.
If $X$ lacks a \emph{control path} (a directed path to $Y$), 
or $Z$ lacks an info path (a path to $Y$, active given $\sC \setminus \{Z\}$), 
then $Z$ is immaterial.
Conversely, if in a graph, every $X$ decision has a control path, and each context $Z$ has an info path, 
then every context is material in some decision problem with that causal structure \citep[Theorem 7]{van2022complete}.%
\footnote{
In full generality, the result allows an info path to terminate at another context, rather than at $Y$.
This detail is not pertinent to the methods used to derive our main result in \Cref{sec:main-result}, although we do 
consider this scenario in \Cref{sec:further-steps}.}
For example, in the graph of \Cref{fig:soluble-everitt},
every decision is an ancestor of $Y$, and every context has an info path, 
(the info paths include $Z \to Y$, $Z' \to W' \gets U' \to Y$, and $W' \gets U' \to Y$),
so, all contexts may be material in at least one decision problem with this causal structure.

It will be important for us to understand what obstacles can arise in proving materiality in multi-decision graphs, 
such as was required in proving \citep[Theorem 7]{van2022complete}.
For example, suppose that we seek to construct a decision problem where $Z$ is material for the graph in in \Cref{fig:soluble-past}
Suppose that we apply the single-decision construction of \citet{everitt2021agent} to this graph.
First, we would identify the info path $Z \to Y$ and the control path $X \to Z' \to X' \to Y$.
The info path has no colliders, so we will construct a decision problem using the scheme 
from \Cref{fig:yes-voi-directed}, and the result is shown in \Cref{fig:soluble-everitt}.
The idea of this construction is that $X$ should have to copy $Z$ in order 
for the value $z$ transmitted by the info path to match the value $x'$ transmitted by the control path.
We see, however, that whatever action $x$ is selected,
the decision $X'$ can assume the value $z$, thereby achieving the MEU.
The MEU is then achievable whether $Z$ is a context of $X$ or not, so $Z$ is immaterial in this construction.

In order to render $Z$ material, we must adapt the construction from \Cref{fig:soluble-everitt} 
by incentivising $X'$ to pass along the value of $Z'$.
To this end, we will use the second info path $Z' \to W' \gets U' \to Y$, shown in \Cref{fig:van-merwijk-scheme}.
We add a term $y_2:=\llbracket u'[x'[0]]=x'[1]\rrbracket$ to the reward, which equals $1$ if $X'$ 
presents one bit from $U'$, along with its index.
We then set $W'=U'[Z']$, so that $X'$ knows only the $Z'$\textsuperscript{th} bit of $U'$, 
and since the index $z'$ is one bit, we let $U'$ be two bits in length, i.e.\ $U' \sim U(\bool^2)$.
Finally, rather than requiring $z=x'$ as in \Cref{fig:soluble-everitt}, we now include 
the term $y_1:=\llbracket z=x'[0] \rrbracket$, because $Z'$ will be the zeroth term of $X'$.
In the resulting model, the utility is clearly $Y=2$ in the non-intervened model, 
and to achieve this utility, the MEU, we must have $Y_1=Y_2=1$ with probability $1$.
To maximise $y_2$, the decision $X'$ must reproduce the only known digit from $U'$, i.e.\ $x'=\langle z',u'[z'] \rangle$.
To maximise $y_1$, we must have $Z=X'[0]$ almost surely, and since $X'[0]=X$, this requires $X=Z$ with probability $1$.
This can only be done if $Z$ is a context of $X$, meaning that $Z$ is material for $X$.
There is a general principle here --- if a control path for $X$, such as $X \to Z' \to X' \to Y$, contains 
decisions other than $X$, then we need to incentivise the downstream decision to copy information 
along the control path, and this will be done by choosing values for variables lying on the info path for $X'$ 
(the one shown in blue in \Cref{fig:van-merwijk-scheme}); we will revisit this matter in our main result.

\subsection{Multi-decision settings in full generality} \label{sec:past-lb}

Once the solubility assumption is relaxed, there are some criteria 
for identifying immaterial variables, but it is not yet known to what extent these 
criteria are necessary, in that materiality is possible whenever they are not satisfied.

The simplest criteria for immateriality are those that carry over from the single-decision case:
\begin{itemize}
\item If a decision $X$ is a non-ancestor of $Y$, then its contexts are immaterial,
\item If $C \dsep Y \mid \sC_X \setminus \{C\}$, then the context $C$ is immaterial.
\end{itemize}

But suppose that we have a graph where neither of these criteria is satisfied.
Then on some occasions, we can still establish immateriality, using the more 
sophisticated criterion of \citet[Theorem 2]{lee2020characterizing}.
The assumptions of this criterion are split across:
\citet[Lemma 1]{lee2020characterizing}
and \citet[Theorem 2]{lee2020characterizing} itself.
\citet[Lemma 1]{lee2020characterizing} establishes that 
if some target variables $\sZ$, target actions $\sX'$, and latent variables $\sU'$ 
satisfy certain separation conditions, then they may be factorized in a 
favourable way.
\citet[Theorem 2]{lee2020characterizing} then proves that 
under some further assumptions, the contexts $\sZ$ are immaterial to the decisions $\sX'$.
in this paper, our focus is exclusively on the assumptions of 
\citet[Lemma 1]{lee2020characterizing}, 
and we term them ``LB-factorizability'', after the authors' initials.
\citet[Theorem 2]{lee2020characterizing} does not feature in our analysis, 
but for completeness sake, it is reproduced in \Cref{app:recap}.

\begin{definition} \label{def:lb-factorizable}
For a scoped graph $\calG_\calS$, 
we will say that
target actions $\sX'$,
endogenous variables $\sZ$ disjoint with $\sX'$, 
contexts $\sC' := \sC_{\sX'} \setminus (\sX' \cup \sZ)$ and
exogenous variables $\sU'$
are \emph{LB-factorizable} 
if there exists an ordering $\prec$ over $\sV' := \sC' \cup \sX' \cup \sZ$ such that:

\begin{enumerate}[itemsep=5pt,label=\Roman*.]
    \item $(Y \dsep \spi_{\sX'} \mid \lceil (\sX' \cup \sC') \rceil )$,
    
    \item $(C \dsep \spi_{{\sX'}_{\prec C}},{\sZ}_{\prec C},\sU' \mid \lceil (\sX' \cup \sC')_{\prec C} \rceil)$, for every $C \in \sC'$ and
    
    \item $\sV'_{\prec X}$ is disjoint with $\Desc(X)$ and subsumes $\Pa(X)$ for every $X \in \sX'$,
\end{enumerate}
where $\spi_{\sX'}$ consists of a new parent $\pi_X$ added to each variable $X \in \sX'$, 
and $\sW_{\prec V}$,
for $\sW \subseteq \sV'$,
denotes the subset of $\sW$ that is strictly prior to $V$ in the ordering $\prec$.
\end{definition}

For example, consider the graph \Cref{fig:triangle-no-voi}.
In this case, $Y \in \Desc(X)$ and $Z \not \dsep Y \mid X$, so the single-decision criteria
cannot establish that $Z$ is immaterial for $X$.
However, by choosing $\sZ=\{Z\},\sX'=\{X\}$, and the ordering $\prec = \langle Z,X \rangle$, 
we have that:
\begin{enumerate}[itemsep=5pt,label=\Roman*.]
\item the outcome $Y$ is d-separated from $\spi_X$ by $\lceil X \rceil$, (since $Z$ is a decision that lacks parents, we actually have $\lceil X \rceil = \{Z,X\}$),
\item the contexts $\sC'$ are an empty set, so (II) is trivially true, and 
\item $\sV'_{\prec X}=\sZ$, and $\sZ$ is disjoint with $\Desc(X)$ and $\sZ \supseteq \Pa(X)$
\end{enumerate}
so $\sZ$ and $\sX'$ are LB-factorizable.
As shown in \Cref{app:recap}, the assumptions of \citet[Theorem 2]{lee2020characterizing} are also satisfied, 
enabling us to deduce that $Z$ is immaterial for $X$, matching the ad hoc analysis
of this graph in \Cref{sec:setup}.

\section{Main result} \label{sec:main-result}
\subsection{Theorem statement and proof overview}

The goal of this paper is to prove that condition (I) of LB-factorizability is necessary
to establish immateriality.
More precisely, we prove that if condition (I) is unsatisfiable for 
all observations in the graph, then the  graph is incompatible with materiality.
It might initially seem unnecessarily stringent to assume that this holds for \emph{all} observations, 
rather than the context $Z_0$ for which we are trying to prove materiality.
Recall from \Cref{fig:van-merwijk-scheme}, however, that proofs of materiality 
are recursive ---
to prove that $Z$ material for $X$, 
we incentivised $X$ to copy $Z$, 
and to do this, we had to incentivise $X'$ has to pass on the value of $Z'$.
To do this, we needed to assume that other contexts and decisions
(such as $Z'$ and $X'$)
have their own info paths and control paths, 
not just $Z$ and $X$.
So, in our theorem below, assumption (C) requires that (I) holds for all contexts.
Assumptions (A) and (B) are also necessary for a graph to be compatible with materiality, 
because their negation implies immateriality, as per the single-decision criteria discussed 
in \Cref{sec:past-1dec-criteria}.
%So, our main result is as follows.

\begin{theorem} \label{thm:main}
If, in a scoped graph $\calG_\calS$,
% If $\calG(\calS)$ is a scoped graph wherein every edge $Z \to X$ to a decision
% $X \in \sX(\calS)$ has:
for every $X \in \sX(\calS)$
\begin{enumerate}[label=\Alph*.]
\item $X \in \Anc_{\calG_\calS}(Y)$,
\item $\forall C \in \sC_X:(C \not \dsep_{\calG_\calS} Y \mid (\{X\} \cup \sC_X \setminus \{ C \}))$, and
\item for every decision $X$ and context $Z \in \sC_X$ in $\calG_\calS$, 
$(\pi_X \not \dsep_{\calG_\calS} Y \mid \lceil (\sX(\calS) \cup \sC_{\sX(\calS) \setminus \{Z\}} ) \setminus \{Z\} \rceil)$, where $\pi_X$ is a new parent of $X$,
%for every set $\sX',\sC'$ such that $X \in \sX' \subseteq \sX(\calS) \setminus Z,\sC' \subseteq C_{\sX'} \setminus (\sX' \cup Z)$, 
%we have $(Y \not \dsep_{\calG_\calS} \Pi_X) \mid \lceil \sX' \cup \sC' \rceil$.
\end{enumerate}
 then for every $X_0 \in \sX(\calS)$ and $Z_0 \in \sC_{X_0}$, there exists an SCM where $Z_0$ is material for $X_0$.
  %$\calS$ is nonredundant under optimality.
\end{theorem}

%\ryan{TODO: ensure that we illustrate SCMs, using $x=...$ not $x \in ...$}
%\paragraph{Past NRO results} 
%The strongest NRO result so far is \citet[Thm. 7]{van2022complete}, which essentially
%states that in a graph that is soluble, every $X$ has conditions (1-2), then 
%for each non-decision context $Z$, there exists a graph such that $Z$ has positive VoI.
%This result does not hold in graphs that are insoluble, and where $Z$ is a decision,
%as can be seen in \Cref{fig:no-voi}, where there can exist no model that gives $Z$ positive VoI.
%Our result will relax the assumption of solubility and $Z$ being a non-decision;
%we will show that for any context $Z$ that satisfies (1-3), there exists an SCM such that $Z$ has positive VoI.
%Other NRO results by \citet{everitt2021agent} and \citet{lee2020characterizing} only consider 
%graphs that contain a single decision; these are a special case because they are immediately soluble.

We will prove this result in three stages, across the next three sections.
\begin{itemize}
    \item In \Cref{sec:paths-exist}, we prove that for any scoped graph satisfying the assumptions of \Cref{thm:main}, for any context $Z_0 \in \sC_{X_0}$, there exist certain paths, 
    which we will call the \emph{materiality paths}.
\item In \Cref{sec:nro-model}, we use the materiality paths to define an SCM for this scoped graph, which we will call the \emph{materiality SCM}.  %(by associating functions and probability distributions to variables on those paths, and trivial functions to other variables) 
\item In \Cref{sec:model-implies-nro}, we will prove that in the materiality SCM, $Z_0$ is material for $X_0$.
\end{itemize}

\subsection{The materiality paths} \label{sec:paths-exist}
To prove materiality, we will begin by selecting info paths and a control path, 
similar to what was described in \Cref{sec:past-soluble} and illustrated in \Cref{fig:van-merwijk-scheme}.
One difference, however, is that these paths must allow for the case where 
we are proving the value of remembering a past decision.
We will first describe how to accommodate this case in \Cref{sec:materiality-path-obstacles} 
then define a set of paths for our proof in \Cref{sec:materiality-paths-general}.

\subsubsection{Paths for the value of remembering a decision} \label{sec:materiality-path-obstacles}
One distinction between our setting and that of \citet{van2022complete} 
is that we may need to establish the value of remembering a past decision,
for example, the value of remembering $Z_0$ in \Cref{fig:remember-decision}.
In this graph, the procedures of \citet{everitt2021agent} and \citet{van2022complete} are 
silent about whether we should 
choose the info path $Z_0 \to Y$, and construct the graph \Cref{fig:remember-decision-1}, 
or choose the info path $Z_0 \gets U \to Y$, 
and construct the model depicted in \Cref{fig:remember-decision-2}.
In the first case, we have $Y=1$ if $x_0=z_0$, i.e. the decision $X_0$ is required to match 
the value of a past decision \Cref{fig:remember-decision-1}.
Then, the MEU of $1$ can be achieved with a deterministic policy such as $Z_0=1,X_0=1$, 
and $Z_0$ is immaterial for $X_0$.
To understand this in terms of the paths involved,
The problem is that the info path $Z_0 \to Y$ doesn't include any parents of $Z_0$,
so $Z_0$ is \emph{implied} by values outside the info path, 
and $Z_0 \to Y$ is rendered inactive given $\lceil U \rceil$.
This means that observing $Z_0$ can no-longer provide useful information about how to maximise $Y$.
In the second case, $Y=1$ if $x_0=u$, i.e.\ the decision $X_0$ must match the value of a 
random Bernoulli variable $U$ \Cref{fig:remember-decision-2}.
$U$ is directly observed only by $Z_0$, and so in optimal policy, 
$X_0$ must observe the decision $z_0$, as is the case in the optimal policy 
$z_0=u, x_0=z_0$, and so $Z_0$ is material for $X_0$.
The info path $Z_0 \gets U \to Y$ does include a parent $U$ of $Z_0$, 
and so $Z_0$ is no-longer \emph{implied} by values outside the info path, 
and the path $Z_0 \gets U \to Y$ remains active given $\lceil \emptyset \rceil$.
Thus $Z_0$ may still provide useful information about $Y$.

\begin{figure}[t]\centering
\begin{subfigure}[b]{0.42\textwidth}\centering
\begin{tikzpicture}[rv/.style={circle, draw, thick, minimum size=6mm, inner sep=0.2mm}, node distance=15mm, >=stealth]
  \node (V) [rv, label={[xshift=8.1mm,yshift=-2mm]}]  {$U$};
  \node (Z) [rv,red, right=10mm of V, label={[xshift=8.1mm,yshift=-2mm]$\sim U(\bool)$}]  {$Z_0$};
  \node (X) [rv,red, right=10mm of Z, label={[xshift=6.6mm,yshift=-2mm]$=Z_0$}]              {$X_0$};
  \node (Y) [rv, right=10mm of X, label={[xshift=-10mm,yshift=4.5mm]$f_Y(\pa_Y)=\llbracket z_0=x_0\rrbracket$}] {$Y$};
  
  \draw[->, very thick] (V) -- (Z);
  \draw[->, very thick] (Z) -- (X);
  \draw[->, very thick] (X) -- (Y);
  %\path(Z) edge[bend right] node [left] {} (Y);
  \draw [->,red,very thick] (Z) to [bend right=20] (Y);
  \draw [->,very thick] (V) to [bend right=20] (Y.-140);
\end{tikzpicture}\vspace{-2mm}
\caption{$Z_0$ is immaterial for $X_0$.} \label{fig:remember-decision-1}
\end{subfigure}
\hspace{6mm}
\begin{subfigure}[b]{0.42\textwidth}\centering
\begin{tikzpicture}[rv/.style={circle, draw, thick, minimum size=6mm, inner sep=0.2mm}, node distance=15mm, >=stealth]
  \node (V) [rv, label={[xshift=8.1mm,yshift=-2mm]$\sim U(\bool)$}]  {$U$};
  \node (Z) [rv,red, right=10mm of V, label={[xshift=6.6mm,yshift=-2mm]$=u$}]  {$Z_0$};
  \node (X) [rv,red, right=10mm of Z, label={[xshift=6.6mm,yshift=-2mm]$=z_0$}]              {$X_0$};
  \node (Y) [rv, right=10mm of X, label={[xshift=-10mm,yshift=4.5mm]$f_Y(\pa_Y)=\llbracket u=x_0\rrbracket$}] {$Y$};
  
  \draw[->,red, very thick] (V) -- (Z);
  \draw[->, very thick] (Z) -- (X);
  \draw[->, very thick] (X) -- (Y);
  %\path(Z) edge[bend right] node [left] {} (Y);
  \draw [->,very thick] (Z) to [bend right=20] (Y);
  \draw [->,red,very thick] (V) to [bend right=20] (Y.-140);
\end{tikzpicture}\vspace{-2mm}
\caption{$Z_0$ is material for $X_0$} \label{fig:remember-decision-2}
\end{subfigure}
\caption{Two SCMs, with models constructed using different (red) info paths.} \label{fig:remember-decision}
\end{figure}
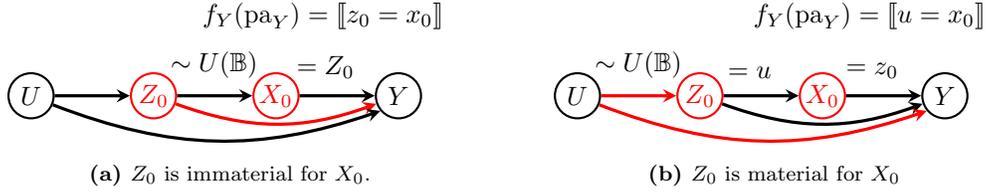

For our proof, we need a general procedure for finding an info path that contains a non-decision parent for every decision.
Condition (C) of \Cref{thm:main} is useful, because it implies the presence of a path from $Z$ to $Y$ that is active given $\condset{Z}$. 
Any fork or chain variables in this path will not be decisions, otherwise they would be contained in $\lceil \sX(\calS) \setminus Z \rceil$, 
which would make them blocked given $\condset{Z}$.
This deals with the possibility of decisions anywhere except for the endpoint $Z$.
But how can we ensure that the info path contains a non-decision parent for $Z$, if it is a decision?
We can use condition (C) again, because it implies that every context that is a decision
must have a non-decision parent.

\begin{lemma} \label{le:chance-parent-of-Z-revised}
If a scoped graph $\calG(\calS)$ satisfies the condition(C) of \Cref{thm:main}, then for every context
$Z \in \sC_X$ 
where $Z,X \in \sX(\calS)$ are decisions, 
there exists a non-decision $N \in \sC_Z \setminus \condset{\{Z\}}$.
\end{lemma}

Intuitively, this is because condition (C) states that there is an active path from $Z$ to $Y$, 
given a superset of $\lceil \sX(\calS) \setminus \{Z \}$. 
If all of the parents of $Z$ 
are decisions, then we would have $Z \in \lceil \sX(\calS) \setminus \{Z \}$, and every path would be blocked, and condition (C) could not be true.

\begin{proof}[Proof of \Cref{le:chance-parent-of-Z-revised}]
Assume that there is no such non-decision $N$, i.e.\ $\sC_Z \subseteq \condset{\{Z\}}$, 
and that $\pi_X \not \dsep Y \mid \condset{\{Z\}}$, (by condition (C) of \Cref{thm:main}), and we will prove a contradiction.
From $\sC_Z \subseteq \condset{\{Z\}}$, 
we deduce that $Z \in \condset{\{Z\}}$ (by the definition of $\lceil \sW \rceil$),
and then there can be no active path from $\pi_X$ to $Y$ given $\condset{\{Z\}} \supseteq \sC_Z \cup \{Z\}$, contradicting condition (C) of \Cref{thm:main}, and proving the result.
\end{proof}

This tells us that for any decision $Z$ there is an edge $Z \gets N$.
Moreover, by condition(C) of the main result, we know that there is an info path from $N$ to $Y$.
By concatenating the edge and the path, we obtain a path from $Z$ to $Y$, 
which we will prove is active given $\condset{\{Z\}}$. 
This is precisely the kind of info path that we are looking for:
activeness given $\condset{Z}$ means that forks and chains will not be decisions,
and we know that the endpoint $Z$ has a non-decision parent $N$.

\begin{restatable}{lemma}{pathstoz} \label{le:decision-backdoor-revised}
If a scoped graph $\calG(\calS)$ satisfies assumptions (B-C) of \Cref{thm:main}, then for every edge 
$Z \to X$ between decisions $Z,X \in \sX(\calS)$,
there exists a path $Z \gets N \upathto Y$,
active given $\condset{\{Z\}}$, %\label{le:decision-backdoor-revised}
(so $N \not \in \condset{\{Z\}}$).
\end{restatable}

Some care is needed in proving that the segment $N \upathto Y$ is active given $\condset{\{Z\}}$, 
rather than just $\condset{\{N\}}$, and the detail is presented in \Cref{le:decision-backdoor-revised}.

%For the particular case of $Z$ being a decision, it is useful 
%to define an ancestor $A$ that is a chance node, and that can serve as the source 
%of the control path, that goes through $Z_0 \to X_0$ and along to $Y$, 
%as shown in \Cref{fig:nro-paths-only}.

\subsubsection{Defining the materiality paths} \label{sec:materiality-paths-general}
We will now describe how to select finitely many info paths, along with a control path, as shown in \Cref{fig:nro-paths-only}.
The assumptions of \Cref{thm:main} allow there to be any finite number of contexts and 
decisions, so we will designate the target decision and context (whose materiality we are trying to establish) 
as $X_0:=X$ and context $Z_0:=Z$.
We know from condition (A) that $X_0$ is an ancestor of $Y$, so we have a directed path $X_0 \pathto Y$.
We also know that $Z_0$ has a chance node ancestor, because it either is a chance node, 
or it has a chance node parent, from \Cref{le:decision-backdoor-revised}.
So we will call that chance node ancestor, $A$, and define a \emph{control path}
of the form $A \pathto Z_0 \to X_0 \pathto Y$, shown in black in \Cref{fig:nro-paths-only}, 
where $A \pathto Z_0$ has length of either $0$ or $1$.

Other paths are then chosen to match this control path.
We will index the decisions on the control path as $X_\imin,\ldots,X_\imax$, 
and their respective contexts are $Z_\imin,\ldots,Z_\imax$.
where $\imin$ is either $0$ (if $Z_0$ is a chance node), or $-1$ (if $Z_0=X_{-1}$).
In general, we allow for the possibility that $Z_i=X_{i-1}$ for any of the decisions.
We define an info path $m_i$ for each context $Z_i$, which must satisfy the desirable properties established in \Cref{le:chance-parent-of-Z-revised}.
To help with our later proofs, it is also useful to define an intersection node $T_i$, 
at which the info path departs from the control path, and a truncated info path $m'_i$, 
which consists of the segment of $m_i$ that is not in the control path.
Recall from \Cref{fig:yes-voi-non-directed} and \Cref{fig:van-merwijk-scheme}
that information from collider variables can play an important role in incentivising 
a decision to copy information from its context.
So, for each collider $W_{i,j}$ in each info path $m_i$ we define an
\emph{auxiliary path} $r_{i,j}:W_{i,j} \pathto Y$.

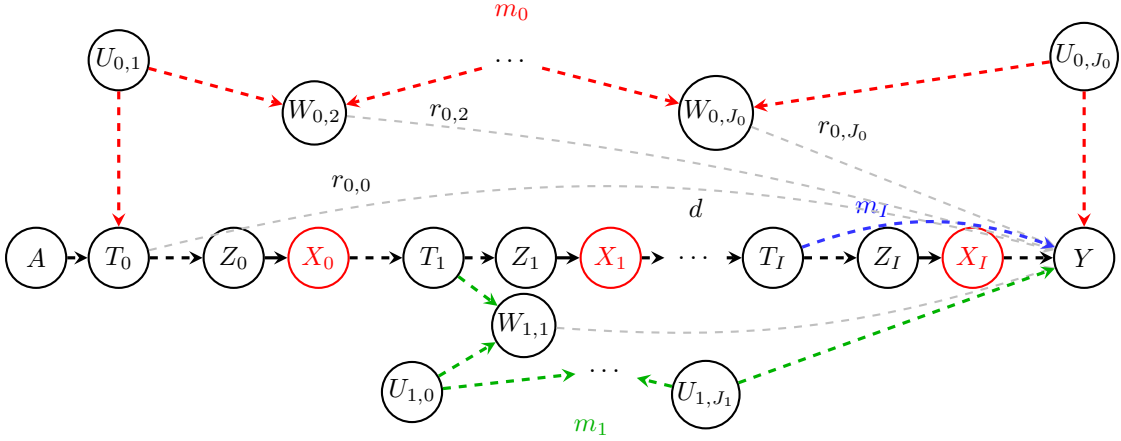
\begin{figure}[t]
\centering
\begin{tikzpicture}[rv/.style={circle, draw, thick, minimum size=8mm, inner sep=0.2mm}, node distance=15mm, >=stealth]
  \node (T0) [rv,label={[xshift=31mm,yshift=3mm]$r_{0,0}$}]  {$T_0$};
  \node (Z0) [rv, right = 7mm of T0]  {$Z_0$};
  \node (X0) [rv,red, right = 3mm of Z0]    {$X_0$};
  \node (N) [rv, left = 2.7mm of T0]  {$A$};
  \node (V0) [rv, above = 18mm of T0]  {$U_{0,1}$};
  \node (Y) [rv, right=93.5mm of X0] {$Y$};
  \node (NJ) [rv, above = 18mm of Y]  {$U_{0,\jmax_0}$};
  \node (W1) [rv, below right = 1mm and 20mm of V0,label={[xshift=18mm,yshift=-7mm]$r_{0,2}$}]  {$W_{0,2}$}; %
  \node (Ndots) [rv, right = 44mm of V0,draw=none, label={[xshift=0mm,yshift=0mm]\color{red}$m_0$}]  {\ldots};
  \node (WJ) [rv, right = 44mm of W1,label={[xshift=17mm,yshift=-10mm]$r_{0,\jmax_0}$}]  {$W_{0,\jmax_0}$};
  
\node (T1) [rv, right=7mm of X0] {$T_1$};
  \node (Z1) [rv, right=4mm of T1] {$Z_1$};
  \node (X1) [rv, red, right=3mm of Z1] {$X_1$};
  \node (Xdots) [rv, right=3mm of X1,draw=none,label={[xshift=0mm,yshift=0mm]\color{black}$d$}] {\ldots};
  \node (TI) [rv, right=2mm of Xdots] {$T_I$};
  \node (ZI) [rv, right=7mm of TI,label={[xshift=-2mm,yshift=0mm]\color{blue}$m_I$}] {$Z_I$};
  \node (XI) [rv, red, right=3mm of ZI] {$X_I$};
  
  \node (V10) [rv, below left = 12mm and -3mm of T1, label={[xshift=12mm,yshift=-11mm]}] {$U_{1,0}$};
  \node (W11) [rv, below right = 3mm and 6mm of T1, label={[xshift=12mm,yshift=-11mm]}] {$W_{1,1}$};
  \node (W1dots) [rv, above right = -3mm and 20mm of V10,draw=none,label={[xshift=-2mm,yshift=-14mm]\color{green!70!black}$m_1$}] {\ldots};
  \node (V1J) [rv, below right = -3mm and 7mm of W1dots, label={[xshift=19mm,yshift=-11mm]}] {$U_{1,\jmax_1}$};

  \draw[->, very thick, color=black,dashed] (N) -- (T0);
  \draw[->, very thick, color=red,dashed] (V0) -- (T0);
  \draw[->, very thick, color=black,dashed] (T0) -- (Z0);
  \draw[->, very thick, color=black] (Z0) -- (X0);
  \draw[->, very thick, color=red,dashed] (V0) -- (W1);
  \draw[->, very thick, color=red,dashed] (Ndots) -- (W1);
  \draw[->, very thick, color=red,dashed] (Ndots) -- (WJ);
  \draw[->, very thick, color=red,dashed] (NJ) -- (WJ);
  \draw[->, very thick, color=red,dashed] (NJ) -- (Y);
  
%\draw[->, thick, color=gray,dashed] (W1) -- (Y);
%\draw[->, thick, color=gray,dashed] (W11) -- (Y);
\draw[->, thick, color=black,dashed, gray!50!white] (W1) to [bend left=5] (Y);
\draw[->, thick, color=black,dashed, gray!50!white] (T0) to [bend left=14] (Y);
\draw[->, thick, color=gray!50!white,dashed] (WJ) -- (Y);
\draw[->, thick, color=black,dashed, gray!50!white] (W11) to [bend right=11] (Y);
  
  \draw[->, very thick, color=black,dashed] (X0) -- (T1);
  \draw[->, very thick, color=black,dashed] (T1) -- (Z1);
  \draw[->, very thick, color=black] (Z1) -- (X1);
  \draw[->, very thick, color=black,dashed] (X1) -- (Xdots);
  \draw[->, very thick, color=black,dashed] (Xdots) -- (TI);
  \draw[->, very thick, color=black,dashed] (TI) -- (ZI);
  \draw[->, very thick, color=black] (ZI) -- (XI);
  \draw[->, very thick, color=black,dashed] (XI) -- (Y);
  \draw[->, very thick, color=green!70!black,dashed] (T1) -- (W11);
  \draw[->, very thick, color=green!70!black,dashed] (V10) -- (W11);
  \draw[->, very thick, color=green!70!black,dashed] (V10) -- (W1dots);
  \draw[->, very thick, color=green!70!black,dashed] (V1J) -- (W1dots);  
  \draw[->, very thick, color=green!70!black,dashed] (V1J) -- (Y);  

\draw[->, very thick, color=blue!80!white,dashed] (TI) to [bend left=20] (Y);
  
\end{tikzpicture}
\caption[The materiality paths]{
The set of paths proven to exist by \Cref{le:nro-paths}
are red, green and blue.
In each case, the point of departure of the active path from the (black) directed path is designated by $T_i$.
In full generality, each path may begin either as $Z_i \pathfrom T_i \gets \cdot$ (as in red), 
or as $Z_i \pathfrom T_i \to \cdot$ (green, blue).
} \label{fig:nro-paths-only}
\end{figure}

Collectively, we refer to the control, info and auxiliary paths as the \emph{materiality paths}.

\begin{restatable}{lemma}{materialitypaths} \label{le:nro-paths}
Let $\calG(\calS)$ be a scoped graph that 
 contains a context $Z_0 \in \sC_{X_0}$ and
satisfies the assumptions of for \Cref{thm:main}. 
Then, it contains the following:

\begin{itemize}
\item A \textbf{control path}: a directed path $d: A \pathto Z_0 \to X_0 \pathto Y$, where $A$ is a non-decision, possibly equal to $Z_0$,
and $d$ contains no parents of $X_0$ other than $Z_0$.
\item 
We can write $d$ as
$A\! \pathto Z_\imin\! \to X_\imin\! \pathto \cdots Z_0 \to X_0 \pathto Z_\imax\! \to X_\imax\! \pathto Y, \imin \leq i \leq \imax$, where
each $Z_i$ is the parent of $X_i$ along $d$ (where $A\pathto Z_{\imin}$ and $X_{i-1}\pathto Z_i$ are allowed to have length $0$).
%The number of decisions upstream of $X_0$ is $-\imin$ and the number of decisions downstream is $\imax$.
Then, for each $i$, define the \textbf{info path}: $m'_i:Z_i \upathto Y$, active given $\condset{Z_i}$, that if $Z_i$ is a decision, begins as $Z_i \gets N$ (so $N \in \sC_{Z_i} \setminus \condset{Z_i}$.)%, and $N$ is neither a decision, nor a context of 
%any variable but $Z_i$.
%Let $T_i$ be the node that is nearest to $Y$ along $m'_i$, such that every node in $Z_i\overset{m'_i}{\upathto} T_i$ is in $S \overset{d}{\pathto} Z_i$, 
\item Let $T_i$ be the node nearest $Y$ in $m'_i:Z_i \upathto Y$ (and possibly equal to $Z_i$) such that 
the segment $Z_i\overset{m'_i}{\upathto }T_i$ of $m'_i$ is identical to 
the segment $Z_i \overset{d}{\pathfrom} T_i$ of $d$.
Then, let the \textbf{truncated info path} $m_i$ be the segment $T_i \overset{m'_i}{\upathto} Y$. 
\item Write $m_i$ as $m_i:T_i \!\pathto\! W_{i,1} \!\pathfrom\! U_{i,1} \!\pathto\! W_{i,2} \!\pathfrom\! U_{i,2} \cdots U_{i,\jmax_i} \!\pathto\! Y$, 
where $\jmax_i$ is the number of \emph{forks} in $m_i$. 
(We allow the possibilities 
that $T_i\!=\!W_{i,1}$ so that $m_i$ begins as $T_i \pathfrom U_{i,1}$, 
or that $J_i=0$ so that $m_i$ is $T_i \pathto Y$.)
Then, for each $i$ and $1 \leq j \leq \jmax_i$, let the
\textbf{auxiliary path} be any 
%shortest
directed path
$r_{i,j}: W_{i,j} \pathto Y$ from $W_{i,j}$ to $Y$. %\ryan{if we need $r_{i,j}$ to be a shortest path, then uncomment that part.}
%(breaking ties arbitrarily).
\end{itemize}
\end{restatable}

The proof was described before the lemma statement, and is detailed in \Cref{app:materiality-paths}.

\subsection{The materiality SCM} \label{sec:nro-model}
We will now show how the materiality paths can be used to define an SCM where $Z_0$ is material for $X_0$.
As with the seleciton of paths, the construction of models will have to
differ a little from the constructions of \Cref{sec:past-1dec-criteria,sec:past-soluble}, 
in order to better deal with insolubility.
So we will first describe how we deal with insoluble graphs, in \Cref{sec:materiality-scm-specific} , 
then define a general model in \Cref{sec:materiality-scm-general}.

\subsubsection{Models for insoluble graphs} \label{sec:materiality-scm-specific}
Certain graphs that are allowed by \Cref{thm:main} violate solubility, 
and the constructions from \citet{everitt2021agent} and \citet{van2022complete}
will need to be altered in order to establish materiality in these graphs.

The assumption of solubility meant that upstream decisions could not contain latent, actionable information ---
in particular, this implied if an info path $m_i$ contains a context $C$ for a decision 
$X' \in \sX(\calS) \setminus \{X_i\}$, then $V$ would have to be context of $X_i$, 
otherwise the past decision $V$ would contain latent information that is of import to $X_i$
\citep[Lemma 28]{van2022complete}.
For example, in \Cref{fig:finite-domain-1} the red info path contains the variable $W_1$, 
which is a context for $X'$ but not for $X_0$, and solubility is violated because $W_1 \dsep Y \mid \{Z_0,X_0,X_1\}$
but it satisfies all the three conditions of \Cref{thm:main}.

We can nonetheless apply the construction from \citep{van2022complete} to this graph, 
by treating $X'$ as through it was a non-decision. 
This yields the decision problem shown in \Cref{fig:finite-domain-1},
which is example of the construction from \Cref{fig:multi-collider}), 
except that there is a decision $X'$ that observes $Z_0$ and $W_1$.
In this model, the outcome $Y$ is equal to $1$ if $x_0$ is equal to $u_1$.
The intended logic of this construction is that since $W_1 = Z_0 \oplus U_q$, the MEU can be achieved 
with the non-intervened policy $X_0=Z_0 \oplus W_1$, which would require $X_0$ to depend on $Z_0$.
In this model, however, there exists an alternative policy where $X'=U_1$ and $X_0=X'$,
which achieves the MEU of $1$, without having $X_0$ directly depend on $Z_0$, and proving that $Z_0$ is immaterial for $X_0$.
Essentially, the single bit of $X'$ sufficed to transmit the value of $U_1$, meaning that $Z_0$ contained no more useful information.
So long as the decision problem allows $X'$ can do this there can be no need for $X_0$ to observe $Z_0$.
So in order to exhibit materiality, we need the domain of $X'$ to be smaller than that of $U_1$.

As such, we can devise a modified scheme, shown in \Cref{fig:finite-domain-2}.
In this scheme, \emph{two} random bits are generated at $U_1$.
The outcome is $Y=1$ if
$X_1$ supplies one bit from $U_1$ along with its index.
A random bit is sampled at $Z_0$, 
and $W_1$ presents the $Z_0$\textsuperscript{th} bit from $U_1$, 
while $X_1$ has a domain of just one bit.
Then, similar to our previous discussion of \Cref{fig:van-merwijk-scheme}, 
the only bit from $U_1$ that $X_0$ can reliably know is the $Z_0$\textsuperscript{th} bit.
Hence the only way achieve the MEU is for $X'$ to inform $X_0$ about the value of $W_1$, 
and for $X_0$ to equal $X_0=\langle Z_0,X' \rangle$.
Importantly, this can only be done if $X_0$ observes $Z_0$; it is material for $X_0$.

\begin{figure}[t]\centering
\begin{subfigure}[b]{0.25\textwidth}\centering
\begin{tikzpicture}[rv/.style={circle, draw, thick, minimum size=6mm, inner sep=0.8mm}, node distance=15mm, >=stealth]
  \node (X0) [rv,red, label={[xshift=12.1mm,yshift=-4mm]$=Z_0 \oplus W_1$}]              {$X_0$};
  \node (Z0) [rv, above = 22mm of X0, label=0:{$Z_0\sim \unif(\bool)$}]  {$Z_0$};
  \node (Y) [rv, right=21.5mm of X0, label={[xshift=-14mm,yshift=-13mm,align=left]$f_Y(\pa_Y)=\llbracket u_{1}=x_0\rrbracket$}] {$Y$};
  \node (V1) [rv, above = 22mm of Y, label={[xshift=-4mm,yshift=-0.5mm]$U_1\sim \unif(\bool)$}]  {$U_{1}$};
  \node (W1) [rv, below right = 3mm and 3.5mm of Z0, label=0:{$=\!Z_0\!\oplus\! U_{1}$}]  {$W_{1}$};
  \node (Xp) [rv,red, below = 3mm of W1, label=0:{$=W_1$}]  {$X'$};
  
  \draw[->, very thick] (Z0) -- (X0);
  \draw[->, very thick, red] (Z0) -- (W1);
  \draw[->, very thick] (Z0) -- (Xp);
  \draw[->, very thick] (W1) -- (Xp);
  \draw[->, very thick] (Xp) -- (X0);
  \draw[->, very thick, red] (V1) -- (W1);
  \draw[->, very thick, red] (V1) -- (Y);
  \draw[->, very thick] (X0) -- (Y);
\end{tikzpicture}
\caption{$Z_0$ is immaterial for $X_0$} \label{fig:finite-domain-1}
\end{subfigure}
\hspace{1mm}
\begin{subfigure}[b]{0.27\textwidth}\centering
\begin{tikzpicture}[rv/.style={circle, draw, thick, minimum size=6mm, inner sep=0.8mm}, node distance=15mm, >=stealth]
  \node (X0) [rv,red, label={[xshift=12.1mm,yshift=-4mm]$=\langle Z_0,X'\rangle$}]              {$X_0$};
  \node (Z0) [rv, above = 22mm of X0, label=0:{$Z_0\sim \unif(\bool)$}]  {$Z_0$};
  \node (Y) [rv, right=21mm of X0, label={[xshift=-12mm,yshift=-13mm,align=left]$f_Y(\pa_Y\!)\!=\!\llbracket u_{1}[x[0]]\!=\!x[1]\rrbracket$}] {$Y$};
  \node (V1) [rv, above = 22mm of Y, label={[xshift=-5mm,yshift=-0.5mm]$U_1\sim \unif(\bool^2)$}]  {$U_{1}$};
  \node (W1) [rv, below right = 3mm and 3.5mm of Z0, label=0:{$=\!U_{1}[Z_0]$}]  {$W_{1}$};
  \node (Xp) [rv,red, below = 3mm of W1, label=0:{$=W_1$}]  {$X'$};
  
  \draw[->, very thick] (Z0) -- (X0);
  \draw[->, very thick, red] (Z0) -- (W1);
  \draw[->, very thick] (Z0) -- (Xp);
  \draw[->, very thick] (W1) -- (Xp);
  \draw[->, very thick] (Xp) -- (X0);
  \draw[->, very thick, red] (V1) -- (W1);
  \draw[->, very thick, red] (V1) -- (Y);
  \draw[->, very thick] (X0) -- (Y);
\end{tikzpicture}
\caption{$Z_0$ is material for $X_0$} \label{fig:finite-domain-2}
\end{subfigure}
\hspace{2mm}
\begin{subfigure}[b]{0.43\textwidth}\centering
\begin{tikzpicture}[rv/.style={circle, draw, thick, minimum size=6mm, inner sep=0.8mm}, node distance=15mm, >=stealth]
  \node (X0) [rv,red,label={[xshift=22mm,yshift=-4.5mm]$X_0\!=\!\langle Z_0,\!\sW_{1:J} \rangle$}]              {$X_0$};
  \node (Z0) [rv, above = 5mm of X0, label={[xshift=7.5mm,yshift=-5.5mm]$=\!U_{0}$}]  {$Z_0$};
  \node (V0) [rv, above = 22mm of X0, label={[xshift=13mm,yshift=-6mm]$U_0\sim \unif(\bool^k)$}]  {$U_{0}$};
  \node (Y) [rv, right=46mm of X0,label={[xshift=-19mm,yshift=-13mm]$f_Y\!(\pa_Y)\!=\!\llbracket x_0 \text{ com.w. }\! u_\jmax \rrbracket$}] {$Y$};
  \node (VJ) [rv, above = 21.5mm of Y, label={[overlay,xshift=-9mm,yshift=-0.5mm]$U_j\sim \unif(\bool^{\exp^{N}_2(k)})$}]  {$U_J$};
  \node (W1) [rv, below right = 7mm and 11mm of V0, label={[xshift=11mm,yshift=-7mm]$=\!U_{1}[U_{0}]$}]  {$W_1$}; %
  \node (Ndots) [rv, right = 24mm of V0,draw=none]  {\ldots};
  \node (WJ) [rv, right = 14mm of W1, label={[xshift=0mm,yshift=-16mm]$W_j=\!U_J[U_{J-1}]$}]  {$W_J$};
  
  \draw[->, very thick,dashed, red] (V0) -- (Z0);
  \draw[->, very thick] (Z0) -- (X0);
  \draw[->, very thick,dashed, red] (V0) -- (W1);
  \draw[->, very thick,dashed, red] (Ndots) -- (W1);
  \draw[->, very thick,dashed, red] (Ndots) -- (WJ);
  \draw[->, very thick,dashed, red] (VJ) -- (WJ);
  \draw[->, very thick,dashed, red] (VJ) -- (Y);
  \draw[->, very thick,dashed] (W1) -- (X0);
  \draw[->, very thick,dashed] (WJ) -- (X0);
  \draw[->, very thick,dashed] (X0) -- (Y);
\end{tikzpicture}
\caption{$Z_0$ is material for $X_0$} \label{fig:multi-collider}
\end{subfigure}
\caption[Constructing the model for the info path]{Two SCMs (a-b), and a description of a family of SCMs, where each dashed line represents a path.
The repeated exponent $\exp^{n}_2(k)$ is defined as $k$ if $n=0$, and $2^{\exp^{n-1}_2(k)}$ otherwise.
}
\label{fig:finite-domain}
\end{figure}
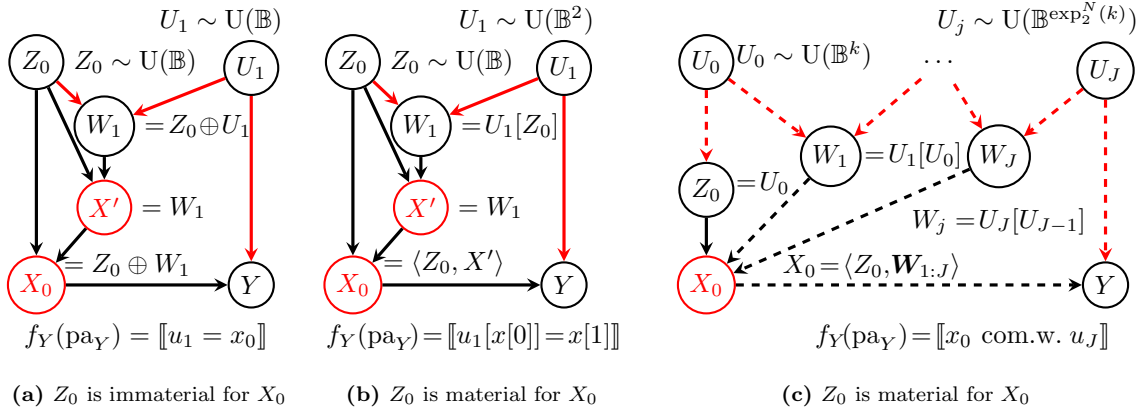

In \Cref{fig:finite-domain-2}, if $x_1$ produces the $z_0$\textsuperscript{th}
bit from $u_1$, i.e.\ $x_1 = \langle z_0,u_1[z_0] \rangle$, 
we will call it \emph{consistent} with $\langle z_0,u_1 \rangle$.
If it produces \emph{any} bit from $u_1$, then we will call it \emph{compatible} with $\langle z_0,u_1 \rangle$.
For instance, either $\langle 0,0 \rangle$ or $\langle 1,1 \rangle$ is compatible with $z_0=0$ and $u_1=01$, 
but only the former is consistent with $z_0=0$ and $u_1=00$.

We can generalise these concepts to a case of multiple fork variables, rather than just $Z_0$ and $U_1$.
For example, \Cref{fig:multi-collider}, we have $J+1$ fork variables $\sU_{0:J}$, 
which sample bitstrings of increasing length.
Then, $Z_0=W_u$, and each collider $W_i$ has $W_i=U_j[U_{j-1}]$.
The outcome $Y$ will still check whether $X_0$ is compatible with $U_J$, but it will do so using a more general definition, 
as follows.

\begin{definition}[Consistency and compatibility] \label{def:compatible}
Let $\sw = \langle w_0,w_1,\ldots,w_J \rangle$ where $w_0 \in \bool^k$ and $w_n \in \bool$ for $n\geq 1$. 
Then, $\sw$ is 
\emph{consistent with}
$\su = \langle u_0,\ldots,u_J, u_i \in \bool^{\exp^i_2(k)} \rangle$ (i.e.\ $\sw \sim \su$) if
$w_0=u_0$ and $w_n = u_n[u_{n-1}]$ for $n \geq 1$.
%Otherwise, we say that $u_0,\ldots,u_J$ is incompatible with $\sw$.
Moreover, 
$\sw$
is \emph{compatible with}
$u_J \in \bool^{\exp^J_2(k)}$ 
(i.e.\ $\sw \sim u_J$) if there exists any $u_0,\dots,u:{J-1}$ 
such that $\sw$ is consistent with $u_0,\dots,u_J$.
\end{definition}
%Note that if $\sw$ is consistent with any $u_0,\dots,u_J$, then it is compatible with $u_J$.

In \Cref{fig:finite-domain-2}, if, with positive probability, the assignment of $X_0$ is inconsistent with $\langle z_0,u_1 \rangle$,
then the decision-maker is also penalised with strictly positive probability.
For instance, if the assignments $z_0=0$ and $u_1=01$ lead to the assignment $x = \langle 1,1 \rangle$, 
then this policy will achieve utility of $y=0$ given the assignments $y_0=0$ and $u_1=00$, 
since they cause the values $z_0=0$ and $w_1=0$, which will cause the assignment $x = \langle 1, 1 \rangle$, 
which is not consistent with $z_0=0$ and $u_1=\langle 0, 0 \rangle$.
We find that the same is true in the more general mode of \Cref{fig:multi-collider}.
If with strictly positive probability, the assignment of $X_0$ is inconsistent with $\su_{0:J}$, 
then there will exist an alternative assignment $\sU_{0:J}=\su'_{0:J}$, 
that produces the same assignments to the observations of $X_0$, 
but where $X_0$ is not compatible with $\su'_J$.

\begin{restatable}{lemma}{encryptiontwo} \label{le:encryption2}
Let 
$\sw = \langle w_0,\ldots,w_J\rangle$
and $\bm{\bar{w}} = \langle \bar{w}_0,\ldots,\bar{w}_J\rangle$
be sequences
with $w_0,\bar{w}_0 \in \bool^k$, 
$w_j,\bar{w}_j \in \bool$ for $j \geq 1$, and
let $J'\leq J$ be the smallest integer such that $w_{J'} \neq \bar{w}_{J'}$.
%Let $n$ be the smallest integer such that $b_n \neq w_n$.
Let $u_0,\ldots,u_{J'}$ be a sequence where $u_j[u_{j-1}]=w_j$ for $1 \leq j< J'$.
%that is compatible with $\sw$.
Then, there exists some
$u_{J'+1},\ldots,u_J$ %where $u_i \in \bool^{\exp^i_2(k)}$, 
such that $\sw$ is consistent with $u_0,\ldots,u_J$,  
but $\bm{\bar{w}}$ is incompatible with $u_J$.
\end{restatable}

The proof is deferred to \Cref{app:collider-path-reqs}.

This result implies that an optimal policy in \Cref{fig:multi-collider}, 
$x_0$ must be consistent with $\su_{0:J}$ with probability $1$.
After all, the non-intervened policy clearly achieves the MEU of $1$, being that it is consistent 
with $\su_{0:J}$, and consistency implies compatibility.
On the other hand, if $x_0$ is inconsistent with $\su_{0:J}$ with strictly positive probability, 
then there will exist an alternative assignment $\su'_{0:J}$ that produces the same assignment $x_0$, 
and since the variables $\sU_{0:J}$ have full support, this will lead to $y=0$ will strictly positive probability, 
and decrease the expected utility.
If a policy cannot copy $Z_0$ without observing it, then this will make $X_0$ inconsistent with $\su$ with strictly positive probability, 
so this policy will not be optimal.
One may notice that by setting $U_0$ to contain $k$ bits rather than just one, this will make it very difficult for $Z_0$ to 
copy the value of $Z_0$ without observing it, if a sufficiently large $k$ is chosen.
We will develop a fully formal argument for materiality in \Cref{sec:model-implies-nro}.

\subsubsection{A decision problem for any graph containing the materiality paths} \label{sec:materiality-scm-general}
We will now generalise the constructions from \Cref{fig:yes-voi-directed} (for a truncated info path is a directed path)
and \Cref{fig:multi-collider} (for a truncated info path that is not a directed path)
to an arbitrary graph containing the materiality paths described in \Cref{le:nro-paths}.

To begin with, let us note that the materiality paths may overlap.
So our general approach will be to define a random variable $V^p$ for each variable in a path $p$.
To derive the overall materiality SCM, we will simply define $V$ by a cartesian product over each $V_p$.
For the outcome variable $Y$, we will instead take a sum over each $Y^p$.
For any set of paths $\sp$, we define $V^{\sp} = \times_{p \in \sp} V^p$.
%For any subset $\bm{p}' \subseteq \bm{p}$ we define $V^{\bm{p}'} = \cat_{p \in {\bm{p}'}} V^p$. \ryan{is this used?}

Let us now discuss the control path.
The initial node $A$ will sample a bitstring that is passed along the control path, and through each intersection node $T_i$ in particular.
To describe this, we will rely on shorthand.

\begin{definition}[Parents along paths] \label{pg:Vp}
When a vertex $V$ has a unique parent $\bar{V}$ along $p$, 
$\Pa({V^p}) = \bar{V}^p$, 
and for a set of paths $\bm{p}'$, let $\Pa(V^{\bm{p}'}) = \cat_{p \in \bm{p}'} \Pa(V^p)$.
For a collider $V$ in a truncated info path $m_i:T_i \upathto Y$, 
let the parent nearer $T_i$ along $m_i$ be $\Pa_L(V)$, and the parent nearer $Y$ be $\Pa_R(V)$.
\end{definition}

For example, a non-outcome child $V$ of $A$ along the control path will be assigned $V^d = \Pa(V^d)$.

Each info path must pass on information from upstream paths that traverse the intersection node.
We therefore use the notation $\sp_i$ to refer to the set of control and auxiliary paths that enter the intersection node $T_i$.
We also devise an extended notion of parents $\Pa^*$ to include this information.
Relatedly, we will define a notion of parents for the auxiliary path, which includes information from the collider $W_{i,j}$ of the info path, 
and a notion of parents for the paths $\sp_i$, that includes the exogenous parent $\Eps_A$ of $A$.
\begin{definition}[Extended parent relations]
For a truncated info path $m_i$, let: 
\ifthesis
\hspace{20mm} $\hspace{20mm}\Pa^*(V^{m_i}) = \begin{cases}T_i^{\sp_i} & \text{ if } \Pa(V^{m_i})=T_i^{m_i} \\ \Pa(V^{m_i}) & \text{ otherwise}\end{cases}$ and \\
\hspace{20mm} $\hspace{20mm}\Pa^*_l(V)=\begin{cases}T_i^{\sp_i} & \text{ if } \Pa_L(V^{m_i})=T_i^{m_i} \\ \Pa_L(V_l^{m_i}) & \text{ otherwise}\end{cases}.$
\else
$$\Pa^*(V^{m_i}) = \begin{cases}T_i^{\sp_i} & \text{ if } \Pa(V^{m_i})=T_i^{m_i} \\ \Pa(V^{m_i}) & \text{ otherwise}\end{cases} \text{, and } 
\Pa^*_l(V)=\begin{cases}T_i^{\sp_i} & \text{ if } \Pa_L(V^{m_i})=T_i^{m_i} \\ \Pa_L(V_l^{m_i}) & \text{ otherwise}\end{cases}.$$
\fi

For an auxiliary path $r_{i,j}$, let $\Pa^*(V^{r_{i,j}}) = \begin{cases}W_{i,j}^{m_i} & \text{ if } \Pa(V^{r_{i,j}})=W_{i,j}^{m_i} \\ \Pa(V^{r_{i,j}}) & \text{ otherwise}\end{cases}.$

Finally, let:
$\Pa^*(V^{\sp_i}) = \begin{cases}\Eps_A \times \Pa(V^{\sp_i}) & \text{ if } V \text{ is } A \\ \Pa(V^{\sp_i}) & \text{ otherwise}\end{cases}.$
\end{definition}

In other respects, the materiality SCM will behave in a similar manner to previous examples.
For instance, when $m_i$ is directed, the outcome $Y^{m_i}$ will evaluate whether the values $\Pa(Y^{\sp_i})$ (which mostly come from $X_i$) are equal to $\Pa(Y^{m_i})$, 
which come from the info path.
When $m_i$ is not directed, the outcome $Y^{m_i}$ will evaluate whether the values from $\Pa(Y^{\sp_i,r_{i,0:\jmax}})$ are compatible with those from $U_{i,\jmax}$.
So let us now define the materiality SCM as follows.

\begin{figure}[t]
\centering
\begin{tikzpicture}[rv/.style={circle, draw, thick, minimum size=6mm, inner sep=0.8mm}, node distance=15mm, >=stealth]
  \node (T0) [rv, label={[xshift=0mm,yshift=-13mm]$=U_{0,1}[A]$}]  {$T_0$};
  \node (Z0) [rv, right = 6.5mm of T0, label={[xshift=0mm,yshift=-13mm]$=T_0$}]  {$Z_0$};
  \node (X0) [rv,red, right = 3.5mm of Z0, label={[xshift=0mm,yshift=-13mm]$=Z_0$}]    {$X_0$};
  \node (N) [rv, left = 4.3mm of T0, label={[xshift=4mm,yshift=0mm]$\sim \unif(\bool^{k})$}]  {$A$};
  \node (V0) [rv, above = 25mm of T0, label={[xshift=0mm,yshift=0mm]$\sim \unif(\bool^{2^k})$}]  {$U_{0,1}$};
  \node (Y) [rv, right=92.5mm of X0, label={[xshift=-59mm,yshift=-41mm,align=left]$f_Y(\pa_Y)\!=\!{\color{red}\llbracket  \pa(Y^{\sp_0,\sr_{0:\jmax^0}})\text{ com.w. } u_{0,\jmax_0} \rrbracket}$ \, $+{\color{green!70!black}\llbracket \pa(Y^{\sp_1,\sr_{0:\jmax^1}}) \text{ com.w. } \!u_{1,\jmax_1} \rrbracket}$ \, $+{\color{blue}\llbracket x_I=t_I \rrbracket}$}] {$Y$};
  \node (NJ) [rv, above = 25mm of Y, label={[xshift=-7mm,yshift=0mm]$\sim \unif(\bool^{\exp_2^{\jmax_0}(k)})$}]  {$U_{0,\jmax_0}$};
  \node (W1) [rv, below right = 3mm and 20mm of V0, label={[xshift=0mm,yshift=0mm]$=U_{0,2}[U_{0,1}]$}]  {$W_{0,2}$}; %
  \node (Ndots) [rv, right = 44mm of V0,draw=none]  {\ldots};
  \node (WJ) [rv, right = 44mm of W1, label={[xshift=0mm,yshift=0mm]$=U_{0,\jmax_0}[U_{0,\jmax_0-1}]$}]  {$W_{0,\jmax_0}$};
  
\node (T1) [rv, right=6.5mm of X0, label={[xshift=0mm,yshift=-13mm]$=X_0$}] {$T_1$};
  \node (Z1) [rv, right=6.5mm of T1, label={[xshift=0mm,yshift=0mm]$=T_1$}] {$Z_1$};
  \node (X1) [rv, red, right=3.5mm of Z1, label={[xshift=0mm,yshift=0mm]$=Z_1$}] {$X_1$};
  \node (Xdots) [rv, right=2.5mm of X1,draw=none] {\ldots};
  \node (TI) [rv, right=2mm of Xdots, label={[xshift=0mm,yshift=0mm]$=\pa(T^d_I)$}] {$T_I$};
  \node (ZI) [rv, right=5.5mm of TI, label={[xshift=0mm,yshift=0mm]$=T_I$}] {$Z_I$};
  \node (XI) [rv, red, right=3.5mm of ZI, label={[xshift=0mm,yshift=0mm]$=Z_I$}] {$X_I$};
  
  \node (V10) [rv, below left = 12mm and 0mm of T1, label={[xshift=12mm,yshift=-11mm]$\sim \unif(\bool^{2^k})$}] {$U_{1,0}$};
  \node (W11) [rv, below right = 4mm and 7mm of T1, label={[xshift=12mm,yshift=-11mm]$=U_{1,0}[T_1]$}] {$W_{1,1}$};
  \node (W1dots) [rv, above right = -5mm and 18mm of V10,draw=none] {\ldots};
  \node (V1J) [rv, below right = 2mm and 7mm of W1dots, label={[xshift=19mm,yshift=-9mm]$\sim \unif(\bool^{\exp_2^{\jmax_1+1}(k)})$}] {$U_{1,\jmax_1}$};

  \draw[->, very thick, color=black,dashed] (N) -- (T0);
  \draw[->, very thick, color=red,dashed] (V0) -- (T0);
  \draw[->, very thick, dashed] (T0) -- (Z0);
  \draw[->, very thick, color=black] (Z0) -- (X0);
  \draw[->, very thick, color=red,dashed] (V0) -- (W1);
  \draw[->, very thick, color=red,dashed] (Ndots) -- (W1);
  \draw[->, very thick, color=red,dashed] (Ndots) -- (WJ);
  \draw[->, very thick, color=red,dashed] (NJ) -- (WJ);
  \draw[->, very thick, color=red,dashed] (NJ) -- (Y);
  
%\draw[->, thick, color=gray,dashed] (W1) -- (Y);
%\draw[->, thick, color=gray,dashed] (W11) -- (Y);
\draw[->, thick, color=black,dashed, gray!50!white] (W1) to [bend left=9] (Y);
\draw[->, thick, color=black,dashed, gray!50!white] (T0) to [bend left=21] (Y);
\draw[->, thick, color=gray!50!white,dashed] (WJ) -- (Y);
\draw[->, thick, color=black,dashed, gray!50!white] (W11) to [bend right=11] (Y);
  
  \draw[->, very thick, color=black,dashed] (X0) -- (T1);
  \draw[->, very thick, dashed] (T1) -- (Z1);
  \draw[->, very thick, color=black] (Z1) -- (X1);
  \draw[->, very thick, color=black,dashed] (X1) -- (Xdots);
  \draw[->, very thick, color=black,dashed] (Xdots) -- (TI);
  \draw[->, very thick, dashed] (TI) -- (ZI);
  \draw[->, very thick, color=black] (ZI) -- (XI);
  \draw[->, very thick, color=black,dashed] (XI) -- (Y);
  \draw[->, very thick, color=green!70!black,dashed] (T1) -- (W11);
  \draw[->, very thick, color=green!70!black,dashed] (V10) -- (W11);
  \draw[->, very thick, color=green!70!black,dashed] (V10) -- (W1dots);
  \draw[->, very thick, color=green!70!black,dashed] (V1J) -- (W1dots);  
  \draw[->, very thick, color=green!70!black,dashed] (V1J) -- (Y);  

\draw[->, very thick, color=blue!80!white,dashed] (TI) to [bend left=24] (Y);
  
\end{tikzpicture}
\caption[The materiality SCM]{
The materiality SCM: a general SCM where $Z_0$ is material for $X_0$.
} \label{fig:nro-model}
\end{figure}

\begin{definition}[Materiality SCM] \label{def:nro-model}
Given a graph containing the materiality paths, we may define the following random variables.

In the control path, $d:A \pathto Y$, let:
    \begin{itemize}
    \item the source be $A^d = \Eps^{A^d}$ where $\Eps^{A^d} \sim \unif(\bool^k)$
    where $k$ is the smallest positive integer such that $2^k > (k + c)bc$, % and $k > bc$,
    where $b$ is the maximum number of variables that are contexts of one decision, $b:=\max_{X \in \sX(\calS)} \lvert C_X \rvert$, 
    and $c$ is the maximum number of materiality paths passing through any vertex in the graph;
    \item every non-endpoint $V$ have $V^d = \Pa(V^d)$.
    \end{itemize}
In each truncated info path that is directed, $m_i:T_i \pathto Y$, let:
\begin{itemize}
    \item the intersection node $T^{m_i}$ have trivial domain;
    \item each chain node be $V^{m_i} = \Pa^*(V^{m_i})$
    \item the outcome have the function $f_{Y^{m_i}}(\pa_Y) = \llbracket \pa(Y^{\sp_i}) = \pa^*(Y^{m_i}) \rrbracket$.
\end{itemize}
In each truncated info path that is not directed, $T_i \upathto \gets W_{i,1} \to \cdots \gets W_{i,\jmax} \pathto Y$, let:
\begin{itemize}
\item each fork be $W_{i,j}^{m_i} \!\!=\! \Eps^{W_{i,j}^{m_i}},\Eps^{W_{i,j}^{m_i}} \!\!\sim \!\Unif(\bool^{\exp_2^{j}(k+\lvert \sp_i \rvert - 1)})$ where $\lvert \sp_i\! \rvert$ is the number of paths in $\sp_i$;
%(Note that the repeated exponent $\exp^{n}_2(k)$ is defined as $k$ if $n=0$, and $2^{\exp^{n-1}_2(k)}$ otherwise.)
\item each chain node be $V^d = \Pa^*(V^d)$;
\item each collider be $V^{m_i} = \Pa_R(V^{m_i})[\Pa^*_L(V^{m_i})]$;
\item each intersection node be $T^{m_i}_i  = \Pa(V^{m_i})[\Pa^*(T_i^{\sp_i})]$ if the info path begins as $T_i \to \cdot$, otherwise it has empty domain;
\item the outcome have the function $f_{Y^{m_i}}(\pa_Y) = \llbracket \pa(Y^{\sp_i,\sr_{i,1:\jmax_i}})$ is compatible with $\pa^*(Y) \rrbracket.$ 
\end{itemize}
    
In each auxiliary path $r_{i,j}:W_{i,j} \to V_2 \pathto Y$, let:
    \begin{itemize}
    \item each chain node have $V^{r_{i,j}} = \Pa^*(v^{r_{i,j}})$.
    \item each source $W_{i,j}$ have trivial domain
    \end{itemize}
    %\item Let every other node and path, other than the outcome $Y$ be a trivial random variable, i.e\ $N^p=\emptyset$. 
%\ryan{Change $N$ to $V$ for arbitrary vertex, and then change fork nodes $V$ to something else.}

Then, let the \emph{materiality SCM}
have outcome variable $Y = \sum_{\imin \leq i \leq \imax} Y^{m_i}$, 
and non-outcome variables $V=\cat_{p \in \{d,m_i,r_{i,1 :\jmax_i} \mid \imin \leq i \leq \imax\}} V^p$.
\end{definition}

Note that this defines an SCM because each variable is a deterministic function of only its endogenous parents and exogenous variables.

We have define the materiality SCM so that decisions behave just as non-decisions, which always 
do what is required to ensure that $Y^{m_i} = 1$.

\begin{restatable}{lemma}{nroyesedge} \label{le:nro-yes-edge}
In the non-intervened model, the materiality SCM has
$Y=\idiff$, surely.
\end{restatable}

The proof follows from the model definition, and is supplied in \Cref{app:nro-yes-edge}.

We also know that each utility term $Y^{m_i}$ is upper bounded at one,
so in order to obtain the MEU, each $Y^i$ must equal $1$, almost surely.

\begin{lemma} \label{le:cannot-lapse}
    If a policy $\spi$ for the materiality SCM,
    has $P^\spi(Y^{m_i} \!<\! 1)\!>\!0$ for any $i$, the MEU is not achieved.
\end{lemma}
\begin{proof}
We know that $\EE^\spi[Y] = \sum_{\imin \leq i \leq \imax} Y^{m_i}$ (\Cref{def:nro-model}),
so for all $i$, $Y^{m_i} \leq 1$ always.
So, if $P^\spi(Y^{m_i}<1)>0$ for any $i$, then 
$\EE^\spi[Y]<\idiff$, which underperforms the policy that is followed in the non-intervened model (\Cref{le:nro-yes-edge}).
\end{proof}

\subsection{Proving materiality in the materiality SCM} \label{sec:model-implies-nro}
We will now prove that in the materiality SCM, if $Z_0$ is removed from the contexts of $X_0$, 
then the performance for at least one of the utility variables $Y^{m_i}$ is compromised,
and so the MEU is not achieved.
The proof divides into two cases, based on whether the child of $X_0$ along the control path 
is a non-decision (\Cref{sec:non-decision-next})
or a decision (\Cref{sec:decision-next}).

%Specifically, if the intersection node $T_1$ is a chance node, 
%as shown in \Cref{fig:T1-chance}, 
%then $Z_0$ must be a context of $X_0$ to achieve $\EE[Y^{m_i}]=1$
%or it does not exist ($\imax=0$), then we will show that the model cannot satisfy the constraints required to achieve $\EE[Y_0]=1$. 
%If $T_1$ is a decision, as shown in \Cref{fig:T1-decision}, then we will show that the model cannot satisfy the constraints required to achieve $\EE[Y_1]<1$.

\begin{figure}[t]\centering
\begin{subfigure}[b]{0.49\textwidth}\centering
\begin{tikzpicture}[rv/.style={circle, draw, thick, minimum size=9.6mm, inner sep=0.3mm}, node distance=15mm, >=stealth,inner sep=0mm]

  \node (X0) [rv,red,inner sep=0mm, label={[xshift=11.1mm,yshift=-12mm]}]              {$X_0$};
  \node (T0) [rv, left=4.3mm of X0,inner sep=0mm] {$T_0$};
  \node (T1) [rv,right=4.3mm of X0, label={[xshift=11.1mm,yshift=-12mm]}]              {$T_1$};
  \node (Pap1) [rv,right=4.3mm of T1, inner sep=0mm]              {$\Pa^{\sp_1}\!(\!Y\!)$};
  \node (Z0) [rv, above = 18mm of T0, label=0:{}]  {$U_{0,1}$};
  \node (Y) [rv, right=36mm of X0, label={[xshift=-3mm,yshift=-13mm,align=left]}] {$Y$};
  \node (Vl) [rv, above = 18.8mm of Y, label={[xshift=9.6mm,yshift=-6mm]}]  {$U_{1,\ell}$};
  \node (ldots) [rv, draw=none, below right = 5mm and 18.5mm of Z0]  {$\ldots$};
  
  \draw[->, very thick,dashed,red] (Z0) -- (T0);
  \draw[->, very thick,dashed] (T0) -- (X0);
  \draw[->, very thick, dashed, red] (Z0) -- (ldots);
  \draw[->, very thick, dashed, red] (Vl) -- (ldots);
  \draw[->, very thick, dashed, red] (Vl) -- (Y);
  \draw[->, very thick,dashed] (X0) -- (T1);
  \draw[->, very thick, dashed] (T1) -- (Pap1);
  \draw[->, very thick] (Pap1) -- (Y);
\end{tikzpicture}
\caption{%The case, analysed in \Cref{le:impossible-decision}, where 
The intersection node $T_1$ is a chance node.\vspace{8mm}} \label{fig:T1-chance}
\end{subfigure}
\begin{subfigure}[b]{0.49\textwidth}\centering
\begin{tikzpicture}[rv/.style={circle, draw, thick, minimum size=9.6mm, inner sep=0.3mm}, node distance=15mm, >=stealth,inner sep=0mm]

  \node (X0) [rv,red,inner sep=0mm, label={[xshift=11.1mm,yshift=-12mm]}]              {$X_0(\!T_1\!)$};
  \node (CX0) [rv, above = 4.3mm of X0, inner sep=0mm]  {$C^{m_1}_{X_0}$};
  \node (X1) [rv,right=4.8mm of X0,red, label={[xshift=11.1mm,yshift=-12mm]}]              {$X_1$};
  \node (Pap1) [rv,right=4.3mm of X1, inner sep=0mm]              {$\Pa^{\sp_1}\!(\!Y\!)$};
  \node (Z0) [rv, above = 18mm of X0, label=0:{}]  {$U_{1,1}$};
  \node (Y) [rv, right=36mm of X0, label={[xshift=-3mm,yshift=-13mm,align=left]}] {$Y$};
  \node (Cnon) [rv, left=4.3mm of X0,inner sep=0mm] {$C^{\neg m_1}_{X_0}$};
  \node (Vl) [rv, above = 18.8mm of Y, label={[xshift=9.6mm,yshift=-6mm]}]  {$U_{1,\ell}$};
  \node (ldots) [rv, draw=none, below right = 5mm and 16.5mm of Z0]  {$\ldots$};
  
  \draw[->, very thick,dashed] (Z0) -- (CX0);
  \draw[->, very thick] (CX0) -- (X0);
  \draw[->, very thick] (Cnon) -- (X0);
  \draw[->, very thick, dashed, red] (Z0) -- (ldots);
  \draw[->, very thick, dashed, red] (Vl) -- (ldots);
  \draw[->, very thick, dashed, red] (Vl) -- (Y);
  \draw[->, very thick] (X0) -- (X1);
  \draw[->, very thick, dashed] (X1) -- (Pap1);
  \draw[->, very thick] (Pap1) -- (Y);
\end{tikzpicture}
\caption{%The case, analysed in \Cref{le:impossible-decision}, where 
The intersection node $T_1$ is a decision. The contexts of $X_0$ are divided into
$C^{m_1}_{X_0}$ (the parent along the info path), and $C^{\neg m_1}_{X_0}$ (the other parents).} \label{fig:T1-decision}
\end{subfigure}
\caption[The intersection node is, or is not, a decision in the proof of materiality]{The cases where the intersection node $T_1$ is a chance node, or a decision} \label{fig:T1-type}
\end{figure}
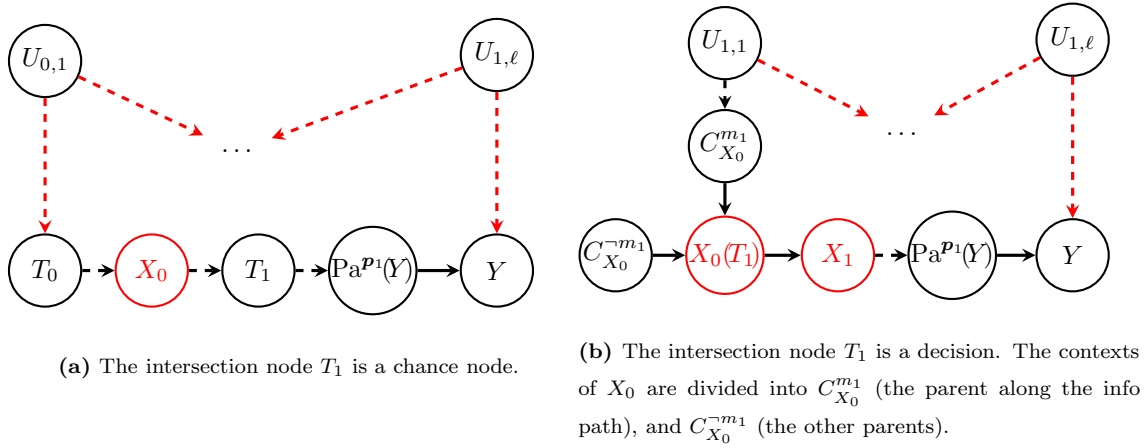

\ryan{Is it really necessary to use the bundle $\sp_i$, rather than just $T^d_i$?}

\subsubsection{Case 1: child of $X_0$ along $d$ is a non-decision.} \label{sec:non-decision-next} 
If the child of $X_0$ along the control path is a non-decision and $Z_0$ is not a context of $X_0$, 
we will prove that $\EE[Y^{m_0}]<1$.
In this case, either $X_0$ is the last decision in the control path, 
or otherwise there must exist an intersection node $T_1$, as shown in \Cref{fig:T1-chance}.
If the former is true, then it is immediate that the value $x_0$ is transmitted to $Y$ 
along the control path, based on the model definition.
As such, $Y_0$ can directly evaluate the decision $X_0$.
For the latter case, we want an assurance that downstream decisions will pass along the value of $X$, 
as was the case in \Cref{fig:van-merwijk-scheme}.
Such an assurance is provided by the following lemma, which shows that 
whenever an intersection node $T_i$ is a chance node --- as is $T_1$ ---
the value $t_i$ is transmitted to $Y$ by every optimal policy.

%Knowing that the value of $X_{i-1}$ is copied from $T_i$ to $Y$ will ultimately enable us to analyse $X_{i-1}$ by its influence on $T_i$.
%This is useful because $\pa(T_i)$ contains the assignment of the decision $x_{i-1}$, 
%while $\pa(Y^{\sp_i})$ contains some of the inputs to $Y^{i-1}$, 
%so this means that $Y^{i-1}$ evaluates the decision $X_{i-1}$ using the value $\pa(Y^{\sp_i})$ (along with $\pa(Y^{m_i})$).
%This is analogous to the fact that in \Cref{fig:decision-on-control-path-yes-voi}
%$X_1=Z_1$, where $z_1$ contains $x_0$.
%allowed us to evaluate whether $x_0=z_0$ using the term $\llbracket z_0=x_1[0] \rrbracket$.

\begin{lemma}[Chance intersection node requirement] \label{le:info-transmission-3}
If in the materiality SCM, where $T_i$ is a chance node, a policy $\spi$ has
$P^\spi(\Pa(T_i^{\sp_i}) = \Pa(Y^{\sp_i}))<1$,
then $P^\spi(Y^{m_i}<1)>0$.
\end{lemma}

First, we prove the case where $m_i$ is a directed path.
In this case, $m_i$ copies the value $t^{\sp_i}$ to $Y$, 
which $Y^{m_i}$ checks against the value $\pa(y^{\sp_i})$ received via the control path.
Maximising $Y^{m_i}$ then requires them to be equal.

\begin{proof}[Proof of \Cref{le:info-transmission-3} when $m_i$ is a directed path]
%If $m_i$ is a directed path, then $T_i$ is a chance node (\Cref{le:nro-paths}).
We have $f_{Y^{m_i}}(\pa_{Y^{m_i}}) = \llbracket \pa(Y^{m_i})  \!=\! \pa(Y^{\sp_i})) \rrbracket$ (\Cref{def:nro-model}).
Also, $\Pa(Y^{m_i}) = T_i^{\sp_i}=\Pa(T_i^{\sp_i})$ surely, 
where the first equality follows from \Cref{def:nro-model}, while 
the second follows from \Cref{def:nro-model} and $T_i$ being a chance node.
So, if $P^\spi(\Pa(Y^{\sp_i}) \!=\! \Pa(T^{\sp_i}))\!<\!1$ then $P^\spi(Y^{m_i}\!=\!1)<1$.
\end{proof}

We now prove the case where $m_i$ is a directed path.
In this case, if the assignment $\pa(Y^{\sp_i})$ 
transmitted along the control path differs from the value $\pa(T^{\sp_i}_i)$
that came in to the intersection node $T_i$,
then just as we established for \Cref{fig:multi-collider}, there will exist an 
assignment $\su_{i,1:\jmax_i}$ to the fork nodes in $m_i$ that 
gives an unchanged assignment to colliders $\sv_{i,1:\jmax_i}$, 
but where $\pa(Y^{\sp_i})$ is incompatible with $u_{\jmax_i}$.

\begin{proof}[Proof of \Cref{le:info-transmission-3} when $m_i$ is not a directed path]
Let us index the forks and colliders of $m_i$ as $T_i \upathto V_{i,1} \pathfrom U_{i,1} \pathto W_{i,1} \pathfrom \cdots W_{i,\jmax_i} \pathfrom U_{i,\jmax_i} \pathto Y$.
Choose any assignments $\pa(T_i^{\sp_i}) \neq \pa(Y^{\sp_i})$ that occur with strictly positive probability. 
Then, there must also exist assignments
$\Pa(Y^{\sp_i,\sr_{i,1:\jmax_i}})=\pa(Y^{\sp_i,\sr_{i,1:\jmax_i}})$,
$\sU_{i,1:\jmax_i}=\su_{1:\jmax_i}$, 
and $\sW_{i,1:\jmax_i}=\sw_{1:\jmax_i}$
such that
$$P^\spi(\pa(T_i^{\sp_i}),\pa(Y^{\sp_i,r_{i,1}}),t_i^{\sp_i},\su_{1:\jmax_i},\sw_{1:\jmax_i})>0.$$
By \Cref{le:encryption2}, there also exists an assignment
$\sU_{i,1:\jmax_i}=\su'_{1:\jmax_i}$ 
such that
$\pa(T_i^{\sp_i}), \sw_{1:\jmax_i}$ is consistent with $u'_{1:\jmax_i}$,
%$w_1=u'_1[t_i^{\sp_i}],\ldots,w_{\jmax_i}=u'_{\jmax_i}[u'_{\jmax_i-1}]$
and
$\pa(Y^\sp_i),\pa(Y^{\sr_{i,1:\jmax_i}})$
is incompatible with 
$u'_{\jmax_i}$.
Now, consider the intervention $\doo(\sU_{i,1:\jmax_i}=\su'_{1:\jmax_i})$.
Since $T_i$ is a chance node, every collider in $m_i$ is a non-decision, and is assigned the (unique) value consistent with $\pa(T_i^{\sp_i}), \su'_{1:\jmax_i}$.
Furthermore, $\pa(T_i^{\sp_i}), \sw_{1:\jmax_i}$ is consistent with $\pa(T_i^{\sp_i}), \su'_{1:\jmax_i}$,
so the intervention does not affect the assignments to these colliders.
Moreover, from \Cref{def:nro-model},
no variable outside of $m_i$ is affected by assignments within $m_i$, except through the colliders.
Therefore: \ryan{maybe revisit this proof?}
\begin{align*}
&P^\spi(\pa(Y^{\sp_i}),\pa(Y^{\sr_{i,{1:\jmax_i}}}),\Pa(Y^{m_i})=u'_{\jmax_i} \mid \doo(\sU_{i,1:\jmax_i}=\su'_{1:\jmax_i})) > 0  \\
\therefore &P^\spi(Y^{m_i}=0 \mid \doo(\sU_{i,1:\jmax_i}=\su'_{1:\jmax_i})) > 0  \\
& \hspace{60mm} (\pa(Y^\sp_i),\pa(Y^{\sr_{i,1:\jmax_i}}) \text{ not compatible with }u'_{\jmax_i})  \\
\therefore & P^\spi(Y^{m_i}=0 \mid \sU_{i,1:\jmax_i}=\su'_{1:\jmax_i}) > 0  \\ 
\intertext{\hfill($\sU_{i,1:\jmax_i}$ are unconfounded, so
$P^\spi(\sV \!\mid\! \doo(\sU_{i,1:\jmax_i}=\su'_{1:\jmax_i}))\!=\!P^\spi(\sV \!\mid\! \sU_{i,1:\jmax_i}\!=\!\su'_{1:\jmax_i})$}
\therefore & P^\spi(Y^{m_i}=0)>0 \hspace{88mm}\text{($P^\spi(\su_{i,1:\jmax_i})>0$)}. \\
\end{align*}
\end{proof}

If $m_i$ is not a directed path, 
then this requirement extends to the values $\pa(Y^{\sr_{i,1:\jmax_i}})$ passed down the auxiliary paths, 
not just the value $\pa(Y^{\sp_i})$ from the control path.
Specifically, $\pa(Y^{\sp_i}),\pa(Y^{\sr_{i,1:\jmax_i}})$ 
must be consistent with
$\pa(Y^{\sp_i}),\su_{i,{1:\jmax_i}}$, where $\su_{i,{1:\jmax_i}}$ denotes the values of forks on the info path.

\begin{restatable}[Collider path requirement]{lemma}{colliderpathrequirements} \label{le:info-transmission-2}
%\begin{lemma}[auxiliary path requirements] \label{le:info-transmission-2}
If the materiality SCM has an info path $m_i$ that is not directed, 
and under the policy $\spi$ there are assignments
$\Pa(Y^{\sp_i,\sr_{i,1:\jmax_i}})\!=\!\pa(Y^{\sp_i,\sr_{i,1:\jmax_i}})$
to parents of the outcome, 
and $\sU_{i,1:\jmax_i}^{m_i} \!=\! \su_{i,1:\jmax_i}^{m_i}$ to the forks of $m_i$,
with 
$P^\spi(\pa(Y^{\sp_i,\sr_{i,1:\jmax_i}}), \su_{i,1:\jmax_i}^{m_i})>0$
and where 
$\pa(Y^{\sp_i,\sr_{i,1:\jmax_i}})$ is inconsistent with $\pa(Y^{\sp_i}), \su_{i,1:\jmax_i}^{m_i}$,
then $P^\spi(Y^{m_i}<1)>0$.
\end{restatable}

The idea of the proof, similar to  \Cref{le:info-transmission-3}, is that whenever the bits transmitted along the auxiliary paths deviate 
from the values $\sw_{i,{1:\jmax_i}}$ of colliders in $m_i$, there exists 
an assignment $\su'_{i,{1:\jmax_i}}$ to forks in $m_i$ that will render the colliders, 
and hence the decision $x_i$ unchanged, while making $x_i$ incompatible with $u_{\jmax_i}$, 
and thereby producing $Y^{m_i}<0$.
A detailed proof is in \Cref{app:collider-path-reqs}.

In order to prove that the context $Z_0$ is needed, we will also need to establish that it is not deterministic, 
even if it is a decision.
In the case where $Z_0$ is a decision, the idea is that random information is generated at $A$, 
which each of the decisions are required to pass along the control path.
We are able to prove this as a corollary of \Cref{le:info-transmission-3}.

\begin{lemma}[Initial truncated info path requirements] \label{le:info-transmission-2b}
If $\spi$ in the materiality SCM does not satisfy:
%\text{the first submitted value almost surely equals the source node, i.e.\ }
%\Pa(Y^{\sp_\imin}) \aseq A^d.
$P^\spi(\Pa(Y^d) = A^d)<1$.
%\tag{iii} \label{eq:2b} \end{equation}
then the MEU is not achieved.
\end{lemma}
\begin{proof}
From \Cref{le:nro-paths}, the control path $d$ begins with a chance node.
So, the first decision $X_\imin$ in $d$ must have a chance node $Z_\imin$ as its parent along $d$.
Furthermore, the intersection node $T_\imin$ must be an ancestor of $Z_\imin$ along $d$, so it is also 
a chance node. 
So it follows from \Cref{le:info-transmission-3}, that any policy $\spi$ must satsify
$P^\spi(T_\imin^{\sp_\imin} = \Pa(Y^{\sp_\imin}))=1$ if it attains the MEU.
As $T_\imin$ is in the control path, we have $d \in \sp_\imin$ (\Cref{le:nro-paths})
so $T_\imin^d \aseq \Pa(Y^d)$ is also required.
Moreover,
all of vertices in the segment $A \pathto T_\imin$ of $d$ are chance nodes, because $X_\imin$ was defined as the first decision in $d$, 
and $T_{\imin}$ precedes it.
And, each chance variable $V^d$ on the control path equals its parent $\Pa(V^d)$ (\Cref{def:nro-model}),
so $A^d = T_\imin^d$, and thus $A^d \aseq \Pa(Y^d)$ is required to attain the MEU.
\end{proof}

We can now combine our previous results to prove that it is impossible to achieve the MEU, if $Z_0$ is not a context of $X_0$, 
in the case where $T_1$ does not exist, or is a non-decision.

\begin{restatable}[Required properties unachievable if child is a non-decision]{lemma}{impossiblenondecision} \label{le:impossible-non-decision}
Let $\calM$ be a materiality SCM %satisfying conditions (A-C) 
%and for some scoped graph $\calG_\calS$, 
where the child of $X_0$ along $d$ is a \emph{non-decision}.
Then, the MEU for the scope $\calS$ cannot be achieved 
by a deterministic policy 
in the scope $\calS_{Z_0 \not \to X_0}$
(equal to $\calS$, except that $Z_0$ is removed from $\sC_{X_0}$).
\end{restatable}

The logic is that if child of $X_0$ in the control path is a non-decision,
then the value of $X_0$ is copied all the way to $\Pa(Y^d)$ (\Cref{le:info-transmission-2b}).
Furthermore, $Z_0^d \aseq \Pa(Y^d)$ is necessary to achieve the MEU (\Cref{le:info-transmission-3}).
But the materiality SCM has been constructed so that the non-$Z_0$ parents of $X_0$ do not contain enough bits to transmit all of the information about $Z_0^d$, 
so the MEU cannot be achieved.
The proof is detailed in \Cref{app:proof-impossible-no-decision}.

\subsubsection{Case 2: child of $X_0$ along $d$ is a decision.} \label{sec:decision-next}
If the child of $X_0$ along $d$ is a decision, as shown in \Cref{fig:T1-decision}, 
we will prove that the decision $X_0$ must depend on $Z_0$ in order to achieve $\EE[Y_1]=1$.
This will be because without $Z_0$, $X_0$ will be limited in its ability to distinguish all of the possible 
values of the first fork node $U_{i,1}$ of $m_1$.
To establish this, we will need to conceive of a possible intervention to the fork nodes in $m_i$, 
that $X_i$ would have to respond to, and so 
we begin by proving that relatively few variables will be causally affected by certain interventions.

\begin{restatable}[Fork information can pass in few ways]{lemma}{limitedeffect} \label{le:limited-effect-from-forks}
If, in the materiality SCM: 

\begin{itemize}
\item the intersection node $T_i$ is the vertex $X_{i-1}$,
\item $\pi_{T_i}$ is a deterministic decision rule 
where
$\pi_{T_i}(\sc^{\neg m_i}(T_i,u_{i,1})=\pi_{T_i}(\sc^{\neg m_i}(T_i,u'_{i,1}))$
for assignments 
$u_{i,1},u'_{i,1}$ to the first fork variable, 
and $\sc^{\neg m_i}(T_i)$ to the contexts of $T_i$ not on $m_i$, and
\item $\sW_{i,1:\jmax_i}=\sw_{i,1:\jmax_i}$, and $\sU_{i,2:\jmax_i}=\su_{i,2:\jmax_i}$ are assignments
to forks and colliders in $m_i$ where each $u_{i,j}$ consists of just $w_{i,j}$ repeated $\exp_2^j(k+\lvert \sp_i\rvert-1)$ times, then:
\end{itemize}
\ifthesis 
\begin{align*}& P^\spi(\pa(Y^{\sp_i,r_{i,1}}),\sc^{\neg m_i}(T_i),\sw_{i,1:\jmax_i}, \su_{i,2:\jmax_i} \mid \!\doo(u_{i,1})) \\
=& P^\spi(\pa(Y^{\sp_i,r_{i,1}}),\sc^{\neg m_i}(T_i),\sw_{i,1:\jmax_i}, \su_{i,2:\jmax_i} \mid \!\doo(u'_{i,1})).
\end{align*}
\else
$$P^\spi(\pa(Y^{\sp_i,r_{i,1}}),\sc^{\neg m_i}(T_i),\sw_{i,1:\jmax_i}, \su_{i,2:\jmax_i} \mid \!\doo(u_{i,1}))\!=\!
P^\spi(\pa(Y^{\sp_i,r_{i,1}}),\sc^{\neg m_i}(T_i),\sw_{i,1:\jmax_i}, \su_{i,2:\jmax_i} \mid \!\doo(u'_{i,1})).$$
\fi
%where $\Pa(Y^{\sp_i,r_{i,1}}))$ are parents of the outcome.
\end{restatable}

The proof follows from the definition of the materiality SCM, 
and it is detailed in \Cref{app:limited-effect-from-forks}.

We can now prove that if a deterministic policy 
does not appropriately distinguish assignments to $U_{i,1}$, 
then the $i$\textsuperscript{th} component of the utility will be suboptimal $\EE[Y^{m_i}]<1$.

\begin{restatable}[Decision must distinguish fork values]{lemma}{infotransmissionfive} \label{le:info-transmission-5}
If in the materiality SCM:
\begin{equation}
  \begin{aligned}
%&\bullet\hspace{3mm} \text{the infopath $m_i$ is not a directed path,} \\
&\bullet\hspace{3mm} \text{the intersection node $T_i$ is the vertex $X_{i-1}$, and} \\
&\bullet\hspace{3mm} \text{$\spi$ is a deterministic policy 
that for assignments $u_{i,1},u'_{i,1}$ to $U_{i,1}$ where $u_{i,1} \!\neq\! u'_{i,1}$}, \\
&\hspace{6.5mm}
\text{has $\pi_{T_i}(\sc^{\neg m_i}(T_i),u_{i,1})\!=\!\pi_{T_i}(\sc^{\neg m_i}(T_i),u'_{i,1})$ 
for every $\sC_{T_i}^{\neg m_i}(T_i)\!=\!\sc^{\neg m_i}(T_i),$}
\end{aligned}\tag{$\dagger$} \label{eq:5} \end{equation}
then $P^\spi(Y^{m_i}<1)>0$
\end{restatable}

The idea of the proof is that if $u_{i,1}$ and $u'_{i,1}$ differ,
there will be some assignment $\pa(Y^{\sp_i})$
such that $u_{i,1}[\pa(Y^{\sp_i})]$ and $u'_{i,1}[\pa(Y^{\sp_i})]$ differ.
When $\Pa(Y^{\sp_i})=\pa(Y^{\sp_i})$ and $u_{i,1}$, 
then $\Pa(Y^{r_{i,1}})$ to assume one value.
But if we intervene $u'_{i,1},u_{i,2:\jmax_i}$, 
then the value of $\Pa(Y^{r_{i,1}})$ will be incorrect, 
making $\Pa(Y^{\sp_i,\sr_{i,1:\jmax_i}})$ inconsistent with 
$\Pa(Y^{\sp_i},U_{i,1:\jmax_i})$ so the maximum expected utility is impossible to achieve.
The details are deferred to \Cref{app:info-transmission-5}.

This will allow us to prove that when the child of $X_0$ along $d$ is a decision,
the MEU cannot be achieved without $Z_0$ as a context of $X_0$.

\begin{restatable}[Required properties unachievable if child is a decision]{lemma}{unachievabledecision} \label{le:impossible-decision-bc-cannot-distinguish}
Let $\calM$ be the materiality SCM for some scoped graph $\calG_\calS$, 
where $\imax > 0$ and $T_1$ is a decision.
Then, there exists no deterministic policy in the scope $\calS_{Z_0 \not \to X_0}$
that achieves the MEU.
%satisfies Conditions \Cref{eq:2,eq:4,eq:2b,eq:5}.% for all $i$.
\end{restatable}

To prove that no deterministic policy in $\calS_{Z_0 \not \to X_0}$ can achieve the MEU 
(achievable with the scope $\calS$), we will show that if a deterministic policy $\spi$ satisfies 
$P^\spi(\Pa(Y^d) = A^d)=1$, as required by \Cref{le:info-transmission-2b},
 then the domain of $X_0 \times C^{\neg m_1}_{X_0}$ 
is smaller than the domain of $C^{m_1}_{X_0}$, 
so \Cref{eq:5} will be satisfied, 
and thus the MEU cannot be achieved.
A detailed proof is presented in \Cref{app:impossible-decision-bc-cannot-distinguish}.

We now combine the lemmas for the two cases to prove the main result.

\begin{proof}[Proof of \Cref{thm:main}]
Any scoped graph $\calG(\calS)$ that satisfies assumptions (A-C) contains materiality paths for the context $Z_0$ of $X_0$ (\Cref{le:nro-paths}), 
and has a materiality SCM (\Cref{def:nro-model}) compatible with $\calG(\calS)$.
In this decision problem, whether the child of $X_0$ along $d$ is or is not a decision, 
the MEU cannot be achieved by a deterministic policy unless $X_0$ is allowed to depend on $Z_0$
(\Cref{le:impossible-non-decision,le:impossible-decision-bc-cannot-distinguish}).
And stochastic policies can never surpass the best deterministic policy (\citep[Proposition 1]{lee2020characterizing}), 
so no such policy can achieve the MEU, 
and so $Z_0$ is material for $X_0$.
\end{proof}

\section{Toward a more general proof of materiality} \label{sec:further-steps}
So far, via \Cref{thm:main} we have established the necessity of condition (I) of LB-factorizability for immateriality.
We now outline some steps toward evaluating the necessity of conditions (II-III) of LB-factorizability, 
and the further condition in \citep[Thm. 2]{lee2020characterizing}.

To begin with, condition (III) requires that we choose an ordering $\prec$, such that 
the parents of each decision $X$ are somewhere before $X$, while the descendants are somewhere afterwards.
Clearly this condition can be satisfied for any acyclic graph, so it instead 

Conditions (II-III) are individually not very restrictive, but are jointly substantial.
So a natural next step is to try to prove that conditions (II-III) are necessary, 
by defining some info paths and control paths for graphs that violate conditions (II-III),
defining a materiality SCM, and proving materiality in that SCM.
So far, however, we have only been able to carry out the first step --- defining the paths --- 
and difficulties have arisen in using those paths to define an SCM that exhibits materiality.
In this section, we will outline what info paths and control paths can be proven to exist, 
and then outline the difficulties in using them to prove materiality.

\subsection{A lemma for proving the existence of paths}
When the variables $\sZ,\sX',\sC',\sU$ are not factorizable, we can prove the existence of info and control paths.

\begin{restatable}[System Exists General]{lemma}{systemexistsgeneral}
\label{le:system-exists-general}
Let $\calG_\calS$ be a scoped graph that satisfies assumptions (A,B) from \Cref{thm:main}.
If $\sZ=\{Z_0\}$, $\sX' \supseteq \Ch(Z_0)$, $\sC'=C_{\sX'} \setminus (\sX' \cup \sZ)$, $\sU=\emptyset$
are not LB-factorizable,
then there exists a pair of paths to some $C' \in \sC' \cup Y$:
\begin{itemize}
    \item an info path $m:Z_0 \upathto C'$, active given $\lceil \sX' \cup \sC' \rceil$, and
    \item a control path $d:X \pathto C'$ where $X \in \sX'$. %$\cap \Ch(Z_0)$.
\end{itemize}
\end{restatable}

A proof is supplied in \Cref{app:proving-path-existence}.
The intuition of this proof is that each of the conditions (I-III) implies a precedence relation between a pair of variables in $\sV' \cup Y$.
Each of these precedence relations can be used to build an ``ordering graph'' over $\sV' \cup Y$.
If the ordering graph is acyclic, then we can let $\prec$ be any ordering that is topological on the graph, and then $\sZ,\sX', \sC',\sU$ are LB-factorizable.
Otherwise, we can use a cycle in the graph to prove the existence of an info path and a control path.
By iterating through these cycles, we can obtain a series of info paths and control paths that terminate at $Y$.

The resulting paths are in some cases, quite useful for proving materiality.
For instance, we can recover the pair of info and control paths used in \Cref{fig:van-merwijk-scheme}.
To prove that $Z$ is material for $X$, 
we can start by choosing $\sX'=\{X,X'\}$, $\sC'=\{Z'\}$, $\sC'=\{Z',W\}$, and $\sU'=\emptyset$.
Then, \Cref{le:system-exists-general} implies the existence of an active path from $Z$ to some $\Desc_X \cap \sC'$, 
so we see that the first info path is the edge $Z \to Z'$.
With $Z'$ being a descendant of $X$, we also have the first control path, $X \to Z'$.
We must then obtain some paths that exhibit why $Z'$ is itself useful for the decision $X$ to know about, and to influence.
To do this, we can reapply \Cref{le:system-exists-general} using the sets
$\sX'=\{X'\}, \sZ=\{Z'\}, \sC'=\{W\}$, and $\sU'=\emptyset$.
We then obtain the new info path $Z' \to W \gets U \to Y$, and the new control path $Z' \to X' \to Y$.
The SCM in \Cref{fig:van-merwijk-scheme} uses these paths to prove $Z$ material for $X$.

\subsection{A further challenge: non-collider contexts}
In some graphs, it is not clear how to use the info and control paths \Cref{le:system-exists-general}  
to prove materiality, because non-collider nodes on the info path may be contexts.
(In previous work, this possibility was excluded by the solubility assumption \citep[Lemma 28]{van2022complete}.)
We will now highlight one case, in \Cref{fig:obstacle}, where it is relatively clear how this challenge can be overcome, 
and one case, \Cref{fig:superimposed}, where it is unclear how to make progress.

In the graph of \Cref{fig:obstacle}, we would like to prove that $Z_0$ is material for $X_0$.
Using \Cref{le:system-exists-general}, we can obtain the red and blue info paths as shown, 
and the corresponding control paths in darker versions of the same colors.
In the approach of \Cref{def:nro-model}, shown in \Cref{fig:obstacle-1}, 
$X_0$ should need to observe $Z_0$ in order to know which slice from $V$ is presented at its parent $X_1$.
Then, $X_1$ would play two roles, one for the red info path, and one for the dark blue control path.
As a collider on the red info path, its role is to present the $Z_0^\text{th}$ bit from $V$.
As the initial endpoint of the blue control path, so its role is to copy the assignment of $Z_0$.
The problem, however, is that $X_0$ then does not need to observe $Z_0$ in order to reproduce its value, 
because this value is already observed at $X_1$, so $Z_0$ is not material.

To remedy this problem, we can construct an alternative SCM, where the value of $Z_0$ is ``concealed'', 
i.e. it is removed from the other contexts, $C_{Z_0} \setminus Z_0$.
At $X_1$, we directly remove $Z_0$, leaving this decision with a domain of only one bit.
At $C$, we impose some random noise, so that it is not always a perfect copy of $Z_0$.
The result is shown in \Cref{fig:obstacle-2}.
When this model is not intervened, an expected utility of $\EE[Y]=10.99$ is achieved, 
because the red term in $Y$ always equals $10$, while the blue term has an expectation of $0.99$.
(This is the MEU, because there is no way to improve the blue term to have expectation $1$ without decreasing 
the expectation of the red term by at least $0.05$.)
If instead, $Z_0$ is removed as a context for $X_0$, then the expected utility can only be as high as $\EE[Y]=10.95$.
To understand this, restrict our attention to deterministic policies, 
and note that in order for the red term to be better than a coin flip (with an expected value of $5$),
we would either need to have $X_0=\langle C,X_1 \rangle$ --- and the red term will have an expectation of $9.95$, 
or we must have $X_1=V[0]$ and $X_0=\langle 0,X_1 \rangle$ --- and then the blue term will have an expectation of $0.5$.
In either case, performance is worse than $10.99$, so $Z_0$ is material for $X_0$.

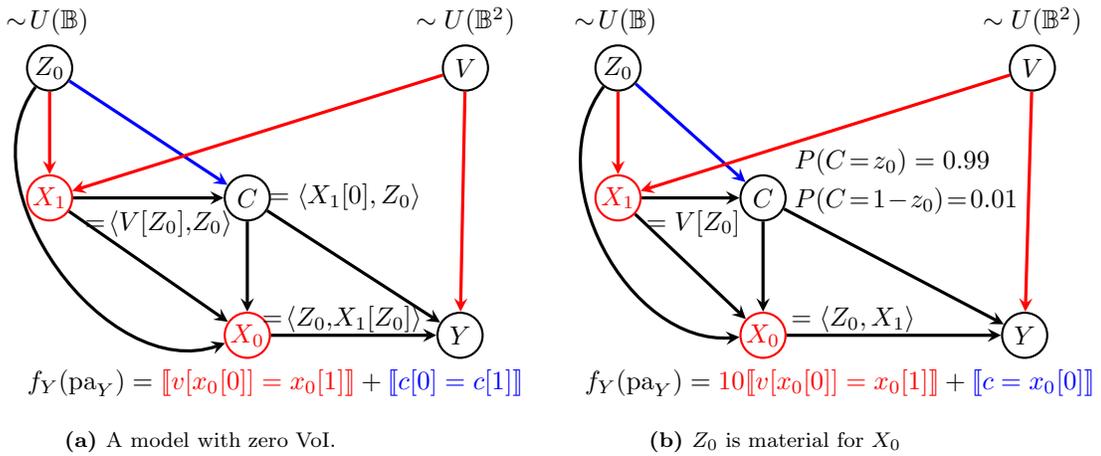
\begin{figure}[ht]
\begin{subfigure}[b]{0.42\textwidth}\centering
\begin{tikzpicture}[rv/.style={circle, draw, thick, minimum size=6mm, inner sep=0.2mm}, node distance=15mm, >=stealth]
\node (Z0) [rv, label=above:{$\!\sim\! U(\bool)$}]  {$Z_0$};
\node (X1) [rv,red, below=11mm of Z0, label={[xshift=14.5mm,yshift=-9.5mm]$=\!\langle V[Z_0],\! Z_0 \rangle$}]              {$X_1$};
\node (C) [rv, right=20mm of X1, label={[xshift=13mm,yshift=-6mm]$=\langle X_1[0],Z_0 \rangle$}]              {$C$};
\node (X0) [rv,red, below=12mm of C, label={[xshift=12.6mm,yshift=-4mm]$=\!\langle Z_0,\!X_1[Z_0]\rangle$}]              {$X_0$};
\node (V) [rv, right=49mm of Z0, label=above:{$\sim U(\bool^2)$}]              {$V$};
\node (Y) [rv, right=22mm of X0, label={[xshift=-24.6mm,yshift=-12.5mm,align=left]$f_Y(\pa_Y)={\color{red}\llbracket v[x_0[0]]=x_0[1] \rrbracket} + {\color{blue}\llbracket c[0]=c[1]\rrbracket}$}]              {$Y$};
%\node (Y) [rv, right=16mm of X1, label={[xshift=21.1mm,yshift=-4mm,align=left]$=\delta (x^1=x^2)$ \\ + $\delta (w=v[x^1])$}] {$Y$};

  \draw[->, very thick, red] (Z0) -- (X1);
  \draw[->, very thick, blue] (Z0) -- (C);
  \draw[->, very thick] (X1) -- (C);
  \draw[->, very thick] (X1) -- (X0);
  \draw[->, very thick] (C) -- (X0);
  \draw[->, very thick] (C) -- (Y);
  \draw[->, very thick, red] (V) -- (X1);
  \draw[->, very thick, red] (V) -- (Y);
  \draw[->, very thick] (X0) -- (Y);
  \draw[->, very thick] (Z0) to[bend right=73] (X0);

\end{tikzpicture}
\caption{A model with zero VoI.} \label{fig:obstacle-1}
\end{subfigure}
\hspace{10mm}
\begin{subfigure}[b]{0.42\textwidth}\centering
\begin{tikzpicture}[rv/.style={circle, draw, thick, minimum size=6mm, inner sep=0.2mm}, node distance=15mm, >=stealth]
\node (Z0) [rv, label=above:{$\!\sim\! U(\bool)$}]  {$Z_0$};
\node (X1) [rv,red, below=11mm of Z0, label={[xshift=10mm,yshift=-9.5mm]$=V[Z_0]$}]              {$X_1$};
\node (C) [rv, right=13mm of X1, label={[xshift=19mm,yshift=-7.5mm]$
\begin{aligned}
    &P(C \!=\! z_0)=0.99 \\[-1.5ex]
    &P(C\!=\!1 \!-\! z_0)\!=\!0.01 \\
\end{aligned}
$}] {$C$};
\node (X0) [rv,red, below=12mm of C, label={[xshift=12mm,yshift=-4mm]$=\langle Z_0, X_1\rangle$}]              {$X_0$};
\node (V) [rv, right=48.8mm of Z0, label=above:{$\sim U(\bool^2)$}]              {$V$};
\node (Y) [rv, right=28.5mm of X0, label={[xshift=-24.6mm,yshift=-12.5mm,align=left]$f_Y(\pa_Y)={\color{red}10\llbracket v[x_0[0]]=x_0[1] \rrbracket} + {\color{blue}\llbracket c=x_0[0]\rrbracket}$}]              {$Y$};
%\node (Y) [rv, right=16mm of X1, label={[xshift=21.1mm,yshift=-4mm,align=left]$=\delta (x^1=x^2)$ \\ + $\delta (w=v[x^1])$}] {$Y$};

  \draw[->, very thick, red] (Z0) -- (X1);
  \draw[->, very thick, blue] (Z0) -- (C);
  \draw[->, very thick] (X1) -- (C);
  \draw[->, very thick] (X1) -- (X0);
  \draw[->, very thick] (C) -- (X0);
  \draw[->, very thick] (C) -- (Y);
  \draw[->, very thick, red] (V) -- (X1);
  \draw[->, very thick, red] (V) -- (Y);
  \draw[->, very thick] (X0) -- (Y);
  \draw[->, very thick] (Z0) to[bend right=66] (X0);

\end{tikzpicture}
\caption{$Z_0$ is material for $X_0$} \label{fig:obstacle-2}
\end{subfigure}
\caption[Randomness may be needed to prove materiality]{Two alternative models that use the same two info paths, red and blue.} \label{fig:obstacle}
\end{figure}

The problem is that concealing the value of $Z_0$ does not work for all graphs.
To see this, let us add two decisions, $X_2$ and $X_3$, to the graph from \Cref{fig:obstacle}, to thereby obtain the graph in \Cref{fig:superimposed}.
Let us retain the materiality SCM from \Cref{fig:obstacle-2}, except that $X_2$ and $X_3$ copy the value from $C$ along to $Y$.
One might expect that $Z_0$ should still be material, but it is not.
Now, there is a policy that achieves the new MEU of $11$ by superimposing the value of $Z_0$ on the assignments of decisions $X_2$ and $X_3$.
In this policy $\spi$,
$x_1 = v[z_0]$,
$x_2 = z_0 \oplus z_0$, 
$x_3 = x_2 \oplus z_0$, 
and $x_0 = x_2 \oplus x_3 = z_0$
where $\oplus$ represents the XOR function.
Under $\spi$, the red term equals $10$ always, while the blue term always equals $1$, i.e.\ the MEU is achieved,
and $\spi$ is a valid policy even if $Z_0$ is not a context of $X_0$, meaning that $Z_0$ is not material for $X_0$.

In summary, 
whenever $\sZ \ni Z_0,\sX' \ni X_0,\sC',\sU$ are not LB-factorizable, 
then we can find some info and control paths for $Z_0$ and $X_0$, 
but then $X_0$ can recover the value of $Z_0$, making it possible to achieve the MEU even when $Z_0$ is removed as a context of $X_0$.
In some graphs, we can devise an alternative SCM that conceals the value of $Z_0$.
But in others, a policy can superimpose the information from $Z_0$ on other decisions, such as $X_2$ and $X_3$ in \Cref{fig:superimposed}, 
so that $X_0$ can recover the value of $Z_0$, making $Z_0$ immaterial for $X_0$ once again.

It seems that new insights are needed to solve this superimposition problem, 
and that therefore that we will need new insights to establish a complete criterion for materiality in insoluble decision problems.

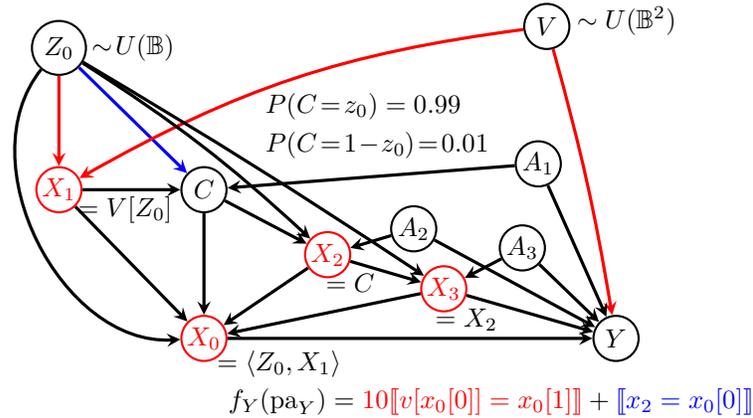
\begin{figure}[ht]\centering
\begin{tikzpicture}[rv/.style={circle, draw, thick, minimum size=6mm, inner sep=0.8mm}, node distance=15mm, >=stealth]
\node (Z0) [rv, label=0:{$\!\sim\! U(\bool)$}]  {$Z_0$};
\node (X1) [rv,red, below=12mm of Z0, inner sep=0mm, label={[xshift=8.8mm,yshift=-8.8mm]$=V[Z_0]$}]              {$X_1$};
\node (C) [rv, right=13mm of X1, label={[xshift=23mm,yshift=-1.2mm]$
\begin{aligned}
    &P(C \!=\! z_0)=0.99 \\[-1.5ex]
    &P(C\!=\!1 \!-\! z_0)\!=\!0.01 \\
\end{aligned}
$}] {$C$};
\node (A1) [rv, above right=-1mm and 40mm of C, inner sep=0mm]              {$A_1$};
\node (X2) [rv,red, below right =4.2mm and 12mm of C, inner sep=0mm, label={[xshift=3mm,yshift=-9.3mm]$=C$}]              {$X_2$};
\node (A2) [rv, above right=-1mm and 7mm of X2, inner sep=0mm]              {$A_2$};
\node (X3) [rv,red, below right =0.1mm and 11mm of X2, inner sep=0mm, label={[xshift=3mm,yshift=-10mm]$=X_2$}]              {$X_3$};
\node (A3) [rv, above right=1mm and 6mm of X3, inner sep=0mm]              {$A_3$};
\node (X0) [rv,red, below=13.5mm of C, inner sep=0mm, label={[xshift=10mm,yshift=-9.5mm]$=\langle Z_0, X_1\rangle$}]              {$X_0$};
\node (V) [rv, above right=-2mm and 60mm of Z0, label={[xshift=10.6mm,yshift=-5.5mm]$\sim U(\bool^2)$}]              {$V$};
\node (Y) [rv, right=48.5mm of X0, label={[xshift=-16.6mm,yshift=-14.5mm,align=left]$f_Y(\pa_Y)={\color{red}10\llbracket v[x_0[0]]=x_0[1] \rrbracket} + {\color{blue}\llbracket x_2=x_0[0]\rrbracket}$}]              {$Y$};
%\node (Y) [rv, right=16mm of X1, label={[xshift=21.1mm,yshift=-4mm,align=left]$=\delta (x^1=x^2)$ \\ + $\delta (w=v[x^1])$}] {$Y$};

  \draw[->, very thick, red] (Z0) -- (X1);
  \draw[->, very thick, blue] (Z0) -- (C);
  \draw[->, very thick] (X1) -- (C);
  \draw[->, very thick] (X1) -- (X0);
  \draw[->, very thick] (C) -- (X0);
  \draw[->, very thick] (C) -- (X2);
  \draw[->, very thick] (X2) -- (X3);
  \draw[->, very thick] (X2) -- (X0);
  \draw[->, very thick] (X3) -- (Y);
  \draw[->, very thick] (X3) -- (X0);
  \draw[->, very thick, red] (V) to[bend right=11] (X1);
  \draw[->, very thick, red] (V) to[bend left=3] (Y);
  \draw[->, very thick] (X0) -- (Y);
  \draw[->, very thick] (Z0) to[bend right=66] (X0);
  
  \draw[->, very thick] (Z0) to[bend left=5] (X2);
  \draw[->, very thick] (A1) -- (C);
  \draw[->, very thick] (A1) -- (Y);
  \draw[->, very thick] (Z0) to[bend left=1.4] (X3);
  \draw[->, very thick] (A2) -- (X2);
  \draw[->, very thick] (A2) -- (Y);
  \draw[->, very thick] (A3) -- (X3);
  \draw[->, very thick] (A3) -- (Y);

\end{tikzpicture}
\caption[Another materiality counterexample: superimposition]{A model with zero VoI} \label{fig:superimposed}
\end{figure}

\section{Conclusion} \label{sec:conclusion}
We have found that in a graph whose contexts cannot satisfy condition (I) of LB-factorizability, 
any context can be material.
We encountered some new problems for materiality proofs, and devised appropriate solutions:
\begin{itemize}
\item if the variable $Z_i$ whose materiality we are trying to establish is a decision, whose value can be determined by other available contexts,
---
then we must choose a different info path so that non-observed variables would be needed to determine the value of $Z_i$
\item if the info path begins with a context of multiple decisions, 
---
then we must construct the SCM differently along the info path
\item if the control path contains consecutive decisions,
then we require more bits to be copied along the control path, so that not all of these bits can be copied along alternative paths.
\end{itemize}

As a next step towards establishing a complete criterion for materiality, we then considered the 
more general setting where 
no context can jointly satisfy conditions (I-III) of LB-factorizability.
In this setting, it is possible to identify info paths and control paths for a target context $Z_0$ and decision $X_0$, 
and to apply our SCM construction to these paths.
However, there may exist policies that transmit the assignment of $Z_0$ through alternative paths, 
and that achieve the MEU even when $Z_0$ is removed as a context of $X_0$.
Although there exist ways of concealing the information about $Z_0$ from a descendant decision $X_{i'},i<i'$, 
there can also be other ways that information about $Z_0$ may be transmitted, such as transmitting this information 
in other decisions, undermining materiality once again.
Thus, the challenge of proving a complete criterion of materiality for insoluble graphs currently remains open.

\section{Acknowledgements}
Thanks to Minwoo Park and Tom Everitt for comments on draft versions of this manuscript.

\addtocontents{toc}{\protect\setcounter{tocdepth}{1}}
\appendix
\clearpage
\bibliography{main-old}\newpage
\section{Recap of Lee and Bareinboim (2020)} \label{app:recap}
Our result \Cref{thm:main} is an initial step in a larger potential project of proving 
that \citep[Thm.\ 2]{lee2020characterizing} is a complete criterion for materiality.

\citep[Thm.\ 2]{lee2020characterizing} result begins with the following factorization \citep{lee2020characterizing}, of which we are only focused on 
cases where the first condition is violated.
The result uses the definition of ``redundancy'' (which is a looser condition than immateriality):
if a scoped graph $\calG(\calS)$ has $X \not \subseteq \Anc_Y$ or $(C \not \dsep Y \mid X \cup \Pa_X \setminus \{ C\})$
then it is ``redundant'',
%Moreover, we know from Thm 2 of \citep{lee2020characterizing} that certain graphs are redundant under optimality.
The result from \citet{lee2020characterizing} is reproduced verbatim:
\begin{customlem}{LB-1}
Given an MPS $\calS$, which satisfies non-redundancy, 
	let $\sX'\subseteq \sX(\calS)$, actions of interest, $\sC'\subsetneq \sC_{\sX'}\setminus \sX'$. non-action contexts of interest.
	If there exists a subset of exogenous variables $\sU'$ in $\calG_\calS$, a subset of endogenous variables $\sZ$ in $\calG_\calS$ that is disjoint with $\sC'\dotcup\sX'$ and subsumes $\sC_{\sX'}\setminus(\sC'\dotcup\sX')$,  
	and 
	an order $\prec$ over $\sV'\doteq \sC'\dotcup\sX'\dotcup\sZ$ such that
	\begin{enumerate}
		\item $\parens{Y \perp \spi_{\sX'} \mid \lceil{\sX'\dotcup\sC'}\rceil}_{\calG_\calS}$,
		\item $\parens{C \perp \spi_{\sX'_{\prec C}},\sZ_{\prec C},\sU' \mid \lceil{{(\sX'\dotcup\sC')}_{\prec C}}\rceil}_{\calG_\calS}$ for every $C\in\sC'$, and
		\item $\sV'_{\prec X}$ is disjoint with $de(X)_{\calG_{\calS}}$ and subsumes $pa(X)_{\calG_{\calS}}$ for every $X\in\sX'$,
	\end{enumerate}
        where, the policy node $\pi_X$ is a new parent added to $X$, %\footnote{In \citep{lee2020characterizing}, $\pi_X$.}
	then the expected reward for $\spi$, a deterministic policy optimal with respect to $\calS$, can be written as
	\begin{align}\label{eq:intermediate-nr-opt}
        \mu_{\spi}=
            \sum_{y,\sc',\sx'} y{Q'_{\sx'}(y,\sc')}
            \sum_{\su', \sz}Q(\su')
            \prod_{Z\in\sZ} Q(z|\sv'_{\prec z},\su')
            \prod_{X\in\sX'} \pi(x|\sc_x).
    \end{align}
\end{customlem}

Lemma LB-1 provides conditions for asserting \Cref{eq:intermediate-nr-opt} given $(\calS, \sX',\sC')$, whether $(\sU', \sZ, \prec)$ exist satisfying three conditions.
It is then used to prove redundancy under optimality using the following theorem.

\begin{customthmalt}{LB-2}
Let $\sU',\sZ$ and $\prec$ satisfy Lemma LB-1. For $Z \in \sZ$, let $\sV_Z$ be a minimal subset of $\sV'_{\prec Z} \cup \sU'$ such that 
$Z \dsep \sU' \mid \sV_Z$.
%$Q(Z \mid \sV_Z)=Q(Z \mid \sV'_{\prec Z},\sU')$. 
We define $\text{fix}(\sT)$ with respect to 
$\{\langle Z,\sV_Z \rangle \}_{Z \in \sZ}$, that is with $\hat{\sT}:=\lceil \sT \rceil \cup \{Z \in \sZ \mid \sV_Z \setminus \sU' \subseteq \lceil \sT \rceil \}$, and 
$\text{fix}(\sT)$ is $\sT$ if $\sT=\hat{\sT}$, and $\text{fix}(\hat{\sT})$ otherwise.
If $\text{fix}(\sC_X \setminus \sZ) \supseteq \sC_X$ for $X \in \sX'$, then 
$\calS' := (\calS \setminus \sX') \cup \{\langle X, \sC_X \setminus \sZ \rangle \}_{X \in \sX'}$
satisfies $\mu^*_{\calS'}=\mu^*_\calS$.
\end{customthmalt}

Let us apply Theorem LB-2 to the graph \Cref{fig:triangle-no-voi}, which we discussed in \Cref{sec:setup}.
We have noted that using $\sZ=\{Z\},\sX'=\{X\}$, and the ordering $\prec = \langle Z,X \rangle$, 
$\sZ$ and $\sX'$ are LB-factorizable.
To apply the theorem, we must confirm that $\text{fix}(\sC_X \setminus \sZ) \supseteq \sC_X$ is true.
The right hand side is simply equal to $Z$.
To evaluate the left hand side, note that $\sC_X \setminus \sZ = \emptyset$.
Furthermore, $\hat{\emptyset}$ includes $\lceil \sT \rceil$, which includes $Z$.
So $\text{fix}(\emptyset)$ also includes $Z$, meaning that 
the left hand side, $\text{fix}(\emptyset)$ is a superset of the right hand side, $\sC_X$, 
and thus $Z$ is immaterial for $X$.

An interested reader may refer to \citet{lee2020characterizing} for
further examples 
where LB-2 is used to establish immateriality.

\section{Supplementary proofs regarding the main result (Theorem \ref{sec:main-result})} \label{app:proofs}

\subsection{Proof of \Cref{le:decision-backdoor-revised}}
We begin by restating the lemma.

\pathstoz*

%\begin{lemma} \label{le:chance-parent-of-Z} 
%If a scoped graph $\calG(\calS)$ meets the conditions for \Cref{thm:main}, then for every edge 
%$Z \to X$ between decisions $Z,X  \in \sX(\calS)$, 
%there exists a chance node $N \in \Pa_Z \setminus \condset{\set{Z}}$ %\label{le:chance-parent-of-Z} 
%\end{lemma}
%\begin{proof}
%\slee{why should $Z$ have a parent? what happens if $Z$ does not have a parent. I can't see where this is covered.}
%\ryan{discussed on call, and hopefully clarified below?}
%Assumption (3) of \Cref{thm:main} implies the existence of a path $\Pi_X \upathto Y$ active given $\condset{Z}$.
%The conditioning set includes $X$, so such a path must begin as $\Pi_X \to X \gets Z$.
%Since this path is active at $Z$ given $\condset{Z}$, we have $Z \not \in \condset{Z}$. 
%Recall that $\lceil \sV \rceil$ is defined 
%so that if $\forall N \in \Pa_Z:J \in \sV$, and $Z$ is a decision then $Z \in \lceil \sV \rceil$.
%Since the path is active, there must exist a parent $N$ of $Z$ such that $N \not \in  \condset{Z}$.
%Since $N \neq Z_0$, we have $N \not \in \sX(\calS)$, thereby proving the result.\bigskip
%\end{proof}

We now prove \Cref{le:decision-backdoor-revised}.

\begin{proof}[Proof of \Cref{le:decision-backdoor-revised}]
%\emph{Proof of \Cref{le:d0}.} \slee{?I can't see how $C$ is the only parent?}\ryan{should be fixed now.}
%If $Z$ is a non-decision, the path $Z \to X$ satisfies the condition \slee{because ...?}\ryan{Should be fixed now.}, 
%and otherwise, using \Cref{le:chance-parent-of-Z-revised}, we have $N \in \Pa_Z \setminus \condset{Z}$, 
%so the path $N \to Z \to X$ satisfies the condition.\bigskip

%\emph{Proof of \Cref{le:decision-backdoor-revised}.}
Since $Z$ is assumed to be a decision, we have from \Cref{le:chance-parent-of-Z-revised}, that there exists 
$N \in \Pa_Z \setminus \condset{Z}$, which therefore is also a chance node.
Assumption (C) of \Cref{thm:main} for $N \to Z$ implies the existence of a path
$p:\Pi_Z \to Z \gets N \upathto Y$ active given $\condset{N}$, which can be truncated as $p':Z \gets N \overset{p}{\upathto} Y$.
We will consider the cases where every collider in $p'$ is in $\condset{Z}$, or there exists one that is not.

Case 1. Every collider in $p'$ is in $\condset{Z}$.
Clearly $p'$ begins as $Z \gets \cdot$ and terminates at $Y$
and is active at colliders, given $\condset{Z}$.
We will now prove that $p'$ is also active given $\condset{Z}$ at non-colliders.
Note that 
$\condset{N} = \lceil (\sX(\calS)\cup C_{\sX(\calS)}) \setminus N \rceil \supseteq \lceil (\sX(\calS) \cup C_{\sX(\calS) \setminus Z}) \setminus (Z \cup N) \rceil = \condset{Z}$, 
 where the
first equality follows from $N$ being a chance node, 
and the latter follows from that and $N \not \in \sC_{\sX(\calS) \setminus Z}$, which jointly imply that $N \in \Pa_Z \setminus \condset{Z}$.
So $p'$ is active given $\condset{Z}$ at non-colliders, and the result is proved for this case.

Case 2. There exists a collider in $p'$ that is not in $\condset{Z}$.
Let $M$ be the collider in $p'$ that is not in $\condset{Z}$, nearest to $Z$ along $p'$.
Since $p'$ is active given $\condset{N}$, we have $M \in \Anc_{\condset{N}}$, 
which implies $M \in \sX(\calS) \cup (C_{\sX(\calS)}$ (because $M \in \lceil \sW \rceil \setminus \sW \implies M \in \sX(\calS)$),
so $M$ is an ancestor of some decision $X'$.
By assumption (B) of \Cref{thm:main},
$X'$ is an ancestor of $Y$,
so we can construct $p'':Z \overset{p'}{\upathto} M \pathto X' \pathto Y$, 
and prove that it satisfies the required conditions.
Clearly $p''$ begins at $Z$, terminates at $Y$. 
The first segment $Z \overset{p'}{\upathto} M$ is active 
at non-colliders given $\Anc_{\condset{Z}}$ by the same argument as in Case 1, 
and at colliders by the definition of $M$.
From $M \not \in \Anc_{\condset{Z}}$, it follows that
$M \pathto X' \pathto Y$ of $p''$ is active given $\Anc_{\condset{Z}}$,
proving the result.

\end{proof}

\subsection{Proof of \Cref{le:nro-paths}} \label{app:materiality-paths}

We begin by restating the lemma.

\materialitypaths*

The proof was described in \Cref{sec:materiality-paths-general}, and is as follows.

\begin{proof}
We prove the existence of each path in turn.
%We know from \Cref{le:paths-to-z} \Cref{le:d0} 

From \Cref{le:chance-parent-of-Z-revised}, 
there exists a control path $A \pathto Z_0$ that contains no parents of $X_0$ other than $Z_0$
(if $Z_0$ is a decision, choose $A=N$, and otherwise choose $A=Z_0$.)
Moreover, from \Cref{thm:main} assumption (A), there exists a path $X_0 \pathto Y$, so
we can concatenate these to obtain $d:A \pathto Z_0 \to X_0 \pathto Y$. \ryan{must the latter segment be a shortest path?}

From assumption (C) of \Cref{thm:main}, there exists an info path $m'_i:Z_i \upathto Y$, active given $\condset{Z_i}$, 
and if $Z_i$ is a decision, one that begins as $Z_i \gets \cdot$, by \Cref{le:decision-backdoor-revised}.
The existence of a truncated info path is immediate from this.

Each collider $W_{i,j}$ is an ancestor of $\condset{Z_i}$ by activeness, hence an ancestor of $\sX(\calS) \cup C_{\sX(\calS)\setminus Z_i}$ by the definition of the 
closure property $\lceil \cdot \rceil$, so $W_{i,j}$ is an ancestor of some $X \in \sX(\calS)$; in addition, from assumption (A) of \Cref{thm:main} we have
$X \in \Anc(Y)$. Hence, there exists a auxiliary path $r_{i,j}:W_{i,j} \pathto Y$.
\end{proof}

\subsection{Proof of \Cref{le:encryption2}}
We begin by proving an intermediate result. \ryan{todo: give better description.}

\begin{lemma} \label{le:encryption1}
Let $\sw = \langle w_0,\ldots,w_J \rangle$, $\bm{\bar{w}} = \langle \bar{w}_0,\ldots,\bar{w}_J \rangle$, %, J \geq 1
and $\bm{u}_{0:J'} = \langle u_0,\ldots,u_{J'} \rangle, J'<J$ 
where $w_0,\bar{w}_0 \in \bool^k$, $w_j,\bar{w}_j \in \bool$ for $n\geq 1$, and $u_j\in \bool^{\exp^n_2(k)}$. %\slee{if J=0, M=?} \ryan{Changed to $J>0$}
If $\sw_{0:J'}$ is consistent with $\su_{0:J'}$ but $\bar{\sw}_{0:J'}$ is not compatible with $u_{J'}$,
then there exists $\su:=\langle u_0,\ldots,u_J',u_{J'+1},\ldots,u_J \rangle$ where $u_j\in \bool^{\exp^n_2(k)}$, 
such that $\sw$ is consistent with $\su$, 
but $\bm{\bar{w}}$ is not compatible with $u_J$.
\end{lemma}
\begin{proof}
We will prove by induction. The base case $j=J'\geq 0$ is given by the condition.\medskip

Induction step: for $j>J'$, if $\bm{w}_{0:j-1} \sim \bm{u}_{0:j-1}$ and $\bm{\bar{w}}_{0:j-1} \not \sim u_{j-1}$, then there exists $\bm{u}_{0:j}$ such that 
(a) $\bm{w}_{0,j} \sim \bm{u}_{0:j}$ and 
(b) $\bar{w}_{0,j} \not \sim u_{j}$.\smallskip

Let us construct $u_j$ such that $u_j[u_{j-1}]\gets w_j$ and $u_j[i]\gets 1-\bar{w}_j$ for every $i\in \bool^{\exp^{j-1}_2(k)} \setminus \{u_{j-1}\}$. \smallskip

(a) First, by the construction $u_j[u_{j-1}]= w_j$ and given condition $\bm{w}_{0:j-1} \sim \bm{u}_{0:j-1}$, we can induce $\bm{w}_{0,j} \sim \bm{u}_{0,j}$. \smallskip

(b) Next, we show that $\bm{\bar{w}}_{0:j} \not \sim u_j$. 
For the sake of contradiction, assume that $\bm{\bar{w}}_{0:j} \sim u_j$.
Then, there exists $\bm{u}'_{0:j} = \langle u'_0, \dotsc,u'_{j-1},u_j \rangle$ satisfying $\bm{\bar{w}}_{0:j} \sim \bm{u}'_{0:j}$. Since $\bar{w}_{0,j-1} \not\sim u_{j-1}$, we can observe $u'_{j-1}\neq u_{j-1}$.
Now, by construction, $u_j[u'_{j-1}]=1-\bar{w}_j\neq \bar{w}_j$. Thus, $\bm{\bar{w}}_{0:j-1} \not\sim u_j$. Contradiction.\smallskip

By induction, $\bm{\bar{w}}$ is not compatible with $u_J$.
\end{proof}

\encryptiontwo*

\begin{proof}
If $u_0,\ldots,u_n$ is incompatible with $\sw$, then the result follows from \Cref{le:encryption1}.
Otherwise, let $u_{n+1}$ be $w_n$ repeated $\exp^{n+1}_2(k)$ times.
Then $u_0,\ldots,u_{n+1}$ is compatible with $w_0,\ldots,w_{n+1}$ but $u_{n+1}$ is incompatible with $\bm{b}$.
We can then apply \Cref{le:encryption1} to obtain the result.
\end{proof}

\subsection{Proof of \Cref{le:nro-yes-edge}} \label{app:nro-yes-edge}
We now prove the expected utility in the non-intervened model 
(which we will later establish is the MEU).

\nroyesedge*

\begin{proof}
Since $Y = \sum_{\imin \leq i \leq \imax} Y^{m_i}$, 
it will suffice to prove that $Y^{m_i}=1$ for every $i$.
We will consider the cases where $m_i$ is, or is not, a directed path.

If the info path $m_i$ contains no collider, then 
every chain node $V$ in $d$ from $T_i$ to $Y$ has $V^d = \Pa^d_V$, 
so 
$\pa(Y^{\sp_i})=T_i^{\sp_i}$.
The same is true for the chain nodes in $m_i$, 
so $\Pa^*(Y) = T_i^{\sp_i}$, and so $Y^{m_i} = 1$, surely.

If $m_i$ contains a collider,
%, write $m_i$ as $T_i \pathto W_{i,1}\! \pathfrom \!U_{i,1} \! \pathto \! W_{i,2}, \ldots, Y$, where possibly $W_{i,1}=T_i$. 
%\hl{dots?}\ryan{I was just indicating that there can be more colliders $W_{i,j}$, but I'm not very attached to the notation, so feel free to change.}
%the values from $X_i$ and each $W_{i,j}$ are passed down to $Y$ along $d$ and $r_{i,j}$ respectively,
% --- formally, 
each chain in $m_i$ and $r_{i,j}$ copies the value of its parent, so
$\Pa(Y^{\sp_i,r_{i,1},\ldots,r_{i,\jmax^i}}) = \langle T_i^{\sp_i}, W_{i,1}^{m_i},\ldots,W_{i,\jmax^i}^{m_i}\rangle$,
and 
%the value of $U_{\jmax^i}$ is likewise conveyed to $\Pa(Y^{m_i})$, i.e.\ 
$\Pa^*(Y) \!=\! U_{\jmax^i}$.
By construction, $\langle T_i^{\sp_i}, W_{i,1}^{m_i},\ldots,W_{i,\jmax^i}^{m_i}\rangle$ 
is consistent with $\langle U_1,\ldots, U_{\jmax^i} \rangle$, 
so by definition it is compatible with $U_{\jmax^i}$, 
%so $\Pa(Y^{\sp_i,r_{i,1},\ldots,r_{i,\jmax^i}}) \sim \unif_{\jmax^i}$, 
so $Y^{m_i}=1$, surely. 
%So for all $i$, $Y^{m_i}=1$, surely.
\end{proof}

\subsection{Proof of the requirements of an optimal policy} \label{app:collider-path-reqs}

\colliderpathrequirements*
\begin{proof}[Proof of \Cref{le:info-transmission-2}]
Let us index the forks and colliders of $m_i$ as $T_i \upathto V_{i,1} \pathfrom U_{i,1} \pathto W_{i,1} \pathfrom, \ldots, W_{i,\jmax^i} \pathfrom U_{i,\jmax^i} \pathto Y$.
Then, by assumption, there exists a set of assignments 
$\sw := \pa(Y^{\sp_i}), \sw_{i,1:\jmax_i}$ \,
$\bar{\sw} := \pa(Y^{\sp_i}), \pa(Y^{\sr_{i,1:\jmax_i}})$ \, and
$\su := \pa(Y^{\sp_i}), \su_{i,1:\jmax_i}^{m_i}$,
where $\sw \sim \su$ and $\bar{\sw} \not \sim \su$
and $P^\spi(\sw,\bar{\sw},\su)>0$.
Let $J'$ be the smallest index such that $\bar{\sw}_{1:J'} \not \sim \su_{1:J'}$, and clearly we will have $J'\geq 1$.
Then, from \Cref{le:encryption2}, 
there exists $\bar{\su} = \pa(Y^{\sp_i}),\su_{1:J'},u_{i,J'+1}^{m_i},\ldots,u_{i,\jmax^i}^{m_i}$
such that $\sw \sim \bar{\su}$ and $\bar{\sw} \not\sim \bar{\su}_{\jmax^i}$.
Consider the intervention $\doo(U_{i,J'+1}^{m_i},\ldots,U_{i,\jmax^i}^{m_i} = \bar{\su}_{J'+1:\jmax^i}$.
By the definition \Cref{def:nro-model}, the intervention to forks on the info path can only affect variables outside of the info path 
via the intersection node $T_i$ and the colliders $W_{i,j},1 \leq j \leq \jmax^i$.
But $\bar{\su}_{1:J'}=\su_{1:J'}$, so $T_i$ and the colliders $W_{i,j},1 \leq j \leq J'$ are unchanged (note that this is true even if $T_i$ is a decision, which it can be).
Furthermore, $\bar{\sw} \sim \su$ so the colliders $W_{i,j},J' < j \leq \jmax^i$ are similarly unaffected by the intervention.
We also have $\bar{\sw} \not\sim \bar{\su_{\jmax^i}}$.
Then, by the same arguments as in the proof of \Cref{le:info-transmission-3}, we have that 
$P^\spi(Y^{m_i}=0 \mid \doo(\bar{\su}))>1$ and then $P^\spi(Y^{m_i}=0)>0$.
\end{proof}

\subsection{Proof of \Cref{le:impossible-non-decision}} \label{app:proof-impossible-no-decision}

We begin by restating the lemma.

\impossiblenondecision*

The proof was described in \Cref{sec:non-decision-next} and it is detailed as follows.

\begin{proof}
Consider the scope $\sX(\calS)_{\setminus Z_0}$, 
equal to $\sX(\calS)$ except that $\sC_{X_0}$ is replaced with $\sC_{X_0} \setminus \{Z_0\}$, 
and assume that a deterministic policy $\spi$ in this scope achieves the MEU,
then we will prove a contradiction.
Specifically, we will establish two consequences that are clearly contradictory given a deterministic policy:
(a) the support of $P^\spi(X^{\sp_0}_0)$ contains at least $2^k$ assignments, 
(b) the domain of $\sC_{X_0} \setminus \{Z_0\}$ contains fewer than $2^k$ assignments.

(Proof of a.) 
We know that 
$A$ assigns a strictly positive probability to $2^k$ assignments (\Cref{def:nro-model})
and so if $\spi$ achieves the MEU, then $\Pa(Y^d) \aseq A$ (\Cref{le:info-transmission-2b}).
So $\Pa(Y^d)$ has at least $2^k$ assignments in its support.
Let us now consider the cases where $X_0$ is, or is not, the decision nearest $Y$ along $d$.

If $X_0$ is the decision nearest $Y$ along $d$,
then by the model definition, $\Pa(Y^d)=X_0^d$ surely, so $X_0$ must have at least $2^k$ assignments in its support, and so (a) follows.

If $X_0$ is not the decision nearest $Y$ along $d$, 
then note that by assumption, there are one or more chance nodes in $d$ separating $X_0$ from $X_1$.
Furthermore, $T_1$ must be one of these nodes (because $T_1$ is defined by a 
segment $T_1 \pathto Z_1$, shared by $d$ and $m'_i$, and active given $\condset{Z_1}$, 
and such a path cannot be active if it includes $X_0$.)
The materiality SCM is constructed to pass values along $d$, and since the segment $T_1 \pathto Z_1$ has no decisions, we have $T^d_1 = X^d_0$, surely.
Since $T_1$ is a chance node, if $\spi$ achieves the MEU, we also have by \Cref{le:cannot-lapse} and \Cref{le:info-transmission-3} that
$\Pa(Y^{\sp_1}) \aseq T^{\sp_1}_1$ and, since $d \in \sp_1$, that $\Pa(Y^{d}) \aseq T^{d}_1$.
So $X^d_0 \aseq \Pa(Y^d)$.
Since $\Pa(Y^d)$ places strictly positive probability on at least $2^k$ assignments, so does $X^d_0$.

(Proof of b.) 
The domain of $\sC_{X_0} \setminus \{Z_0\}$ is a Cartesian product of variables $V^p$ for $V \in \sC_{X_0} \setminus \{Z_0\}$
where $p$ is either $d$, some $m_i$ or some $r_{i,j}$ \Cref{def:nro-model}.

The control path $d$ does not intersect $\sC_{X_0} \setminus \{Z_0\}$ as
it is defined not to include parents of $X_0$ other than $Z_0$ (\Cref{le:nro-paths}).
Each info path $m_i$ is active given $\condset{Z_0}$ (\Cref{le:nro-paths}),
so can only intersect $\sC_{X_0} \setminus \{Z_0\}$ at the colliders, which have domain $\bool$.
Finally, any variable in a path $r_{i,j}$ would also have domain $\bool$.
So the domain of $\sC_{X_0} \setminus Z_0$ is not larger than $2^{c \cdot \lvert \sC_{X_0} \rvert}$, 
where $c$ is the maximum number of materiality paths passing through any vertex in the graph, and
$\lvert \sC_{X_0} \rvert$ is the number of variables in $\sC_{X_0}$.
By construction, $k > c \cdot \max_{X\in\sX(\calS)} \lvert C_X \rvert$, 
so the domain of $\sC_{X_0} \setminus Z_0$ is less than $2^k$, proving (b).

A deterministic policy cannot map fewer than $2^k$ assignments to greater than $2^k$ assignments, and so (a-b) imply a contradiction.
\end{proof}

\subsection{Proof of \Cref{le:limited-effect-from-forks}} \label{app:limited-effect-from-forks}
We firstly restate the lemma.

\limitedeffect*

The proof is as follows.

\begin{proof}
An intervention $\doo(u'_{i,1})$ could, in the materiality SCM (\Cref{def:nro-model}) only affect 
the variables $\Pa(Y^{\sp_i,r_{i,1}}),\sC^{\neg m_i}(T_i),\sW_{i,1:\jmax_i}, \sU_{i,2:\jmax_i}$
in four ways:
\begin{enumerate}
\item via the intersection node $T_i$, 
\item via the collider $W_{i,2}$ of $m_i$, 
%\item via the colliders $\sW_{i,1:\jmax_i}$ of $m_i$, 
%\item via contexts lying on $m_i$, apart from $T_i$ or $\sW_{i,1:\jmax_i}$, or
\item via contexts lying in the segment $m_i:T_i \pathfrom U_{i,1} \pathto W_{i,2}$,
\item if $\Pa^{\sp_i}_Y,\sC^{\neg m_i}(T_i)$ or $\sU_{i,2:\jmax_i}$ were distinct from $T_i,W_{i,2}$ and lay on $m_i:T_i \pathfrom U_{i,1} \pathto W_{i,2}$
\end{enumerate}

The deterministic decision rule has $\spi_{T_i}(u_{i,1},\sc^{\neg m_i}(T_i))=\spi_{T_i}(u'_{i,1},\sc^{\neg m_i}(T_i))$, so (1) is false.
Also, $u_{i,2}$ equals $w_{i,2}$ repeated, so $u_{i,2}[x] = w_{i,2}$ for all $x$, and thus (2) is false also.
Moreover, $m_i:T_i \pathfrom U_{i,1} \pathto W_{i,2}$ is active given $\condset{T_i}$ and so contexts can only lie 
at the endpoints $T_i$ and the collider $W_{i,2}$, meaning that (3) is false.
Finally, $\Pa^{\sp_i}_Y$ is a descendant of $T_i$ by the definition of the control path, so can only lie on $m_i:T_i \pathfrom U_{i,1} \pathto W_{i,2}$
if it is the vertex $T_i$, which we have already proved is not influenced by $u_{i,1}$;
meanwhile, $\sC^{\neg m_i}(T_i)$ does not intersect $m_i$ by definition, 
and $\sU_{i,2:\jmax_i}$ are fork variables, which cannot lie on $m_i:T_i \pathfrom U_{i,1} \pathto W_{i,2}$, 
so (4) is false, and the result follows.
\end{proof}

\subsection{Proof of \Cref{le:info-transmission-5}} \label{app:info-transmission-5}
We begin by restating the lemma.

\infotransmissionfive*

The proof has been described already, and it proceeds as follows.

\begin{proof}
Let us assume \Cref{eq:5}, and that the MEU is achieved, and we will prove a contradiction.
Given \Cref{eq:5}, there is an index at which $u_{i,1}$ and $u'_{i,1}$ differ.
We write this index as an assignment 
$\pa(Y^{d,r_{i,j}})$, belonging to $\Pa(Y^{\sp_i})$.
Define each $u_{i,j},2\leq j \leq J_i$ as equal to $\pa(Y^{r_{i,j}})$, repeated 
$\exp_2^j(k+\lvert \sp_i\rvert-1)$ times.
Then, we have:
\begin{align*}
0&<P^\spi(A^d=\pa(Y^d),\sU_{i,1:\jmax_i}=\su_{i,1:\jmax_i}) & \\
\intertext{because $A$ and $\sU_{i,1:\jmax_i}$ are independent random variables with full support.
Then, %let $t_i:=\pi_{T_i}(\sc^{\neg m_i}(T_i),u_{i,1})$, and 
let $\sc^{\neg m_i}(T_i)$ and $\sw_{1,1:\jmax_i}$
be any assignments to the parents of $T_i$ not on $m_i$, 
and to the colliders on $m_i$ such that:}
0&<P^\spi(A^d=\pa(Y^d),\sc^{\neg m_i}(T_i),\sw_{1,1:\jmax_i},\su_{i,1:\jmax_i}). & \\
\intertext{Given these assignments, in order to achieve $P^\spi(Y^{m_i}=1)=1$, 
we must have $\Pa(Y^d)\aseq A^d$ (\Cref{le:info-transmission-2b})
and $\pa(Y^{\sp_i})$ must be consistent with $\su_{i,1:\jmax_i}$ (\Cref{le:info-transmission-2}).
We must also therefore have $\Pa(Y^{\sp_i,\sr_{i,1:\jmax_i}})=\pa^{\sp_i,\sr_{i,1:\jmax_i}})$, so
marginalising over $A^d$, we must have:}
0&<P^\spi(\Pa(Y^{\sp_i,\sr_{i,1:\jmax_i}})\!=\!\pa(Y^{\sp_i,\sr_{i,1:\jmax_i}}\!),\sc^{\neg m_i\!}(T_i),\sw_{1,1:\jmax_i}\!,\!\su_{i,1:\jmax_i}\!) & \\
\therefore 0&<P^\spi(\pa(Y^{\sp_i,\sr_{i,1:\jmax_i}}),\sc^{\neg m_i}(T_i),\sw_{1,1:\jmax_i}, \su_{i,2:\jmax_i}) \mid \doo(u_{i,1})) & \!\!\!\!\hspace{-3mm}(U_{i,1:\jmax_i}\text{ \!unconfounded}) \\
&=P^\spi(\pa(Y^{\sp_i,\sr_{i,1:\jmax_i}}),\sc^{\neg m_i}(T_i),\sw_{1,1:\jmax_i}, \su_{i,2:\jmax_i}\mid \doo(u'_{i,1})) & \text{(by }\Cref{le:limited-effect-from-forks}) \\
 &=P^\spi(\pa(Y^{\sp_i,\sr_{i,1:\jmax_i}}),\sc^{\neg m_i}(T_i),\sw_{1,1:\jmax_i},\su_{i,2:\jmax_i}\mid u'_{i,1})& (P^\spi(u'_{i,1})>0.) \\
 \therefore 0&<P^\spi(\pa(Y^{\sp_i,\sr_{i,1:\jmax_i}}),u'_{i,1})& (P^\spi(u'_{i,1})>0.)
\end{align*}

However, $u'_{i,1}[\pa(Y^{\sp_i})] \neq u_{i,1}[\pa(Y^{\sp_i})]$
and $u_{i,1}[\pa(Y^{\sp_i})]=\pa(Y^{r_{i,1}})$,
so 
$\pa(Y^{\sp_i})$, $\pa(Y^{\sr_{i,1:\jmax_i}})$ is inconsistent with 
$\pa(Y^{\sp_i}), u'_{i,1},\su_{i,2:\jmax_i}$.
So $0<P^\spi(\pa(Y^{\sp_i,\sr_{i,1:\jmax_i}}),u'_{i,1})$ implies that $P^\spi(Y_1=1)<1$ (by \Cref{le:info-transmission-2}),
and the MEU is not achieved.
\end{proof}

\section{Proof of Lemma~\ref{le:impossible-decision-bc-cannot-distinguish}} \label{app:impossible-decision-bc-cannot-distinguish}

We first restate the lemma.

\unachievabledecision*

The proof was explained in section \Cref{sec:decision-next}, and is detailed as follows.

\begin{proof}
To begin with, by assumption, the child of $X_0$ along $d$ is a decision, 
so $X_0$ is the same node as $Z_1$,
and since the segment $T_1 \pathto X_1$ must be active given $\condset{Z_1}$, 
$X_0$ is also $T_0$.
We will now bound the domains of $X_0$ and $\sC^{\neg m_1}(X_0)$.
%Then, the info path $m_1$ must be of the form $Z_1 \gets C^{m_1}_{X_0} \upathto Y$ %where $C^{m_1}_{X_0} \in \sC_{Z_0} \setminus \condset{Z_1}$ 
%(\Cref{le:nro-paths}), meaning that requirement \Cref{eq:5} is applicable to $i=1$.

\emph{The domain of $X_0$.}
Given that $X_0$ is a decision,
while each truncated info path $m_{i'}$ is active given $\condset{Z_i}$, 
it follows that $X_0$ cannot overlap with info paths, except for 
colliders of $m_{i'},i' \neq i$, and the endpoint of $m_1$.
As such, the domain of $T_0$ is at most $\lvert \dom{X_0} \rvert \leq 2^{k+c}$, 
due to $k$ bits from $d$ (\Cref{def:nro-model}), 
and at most $c$ bits from the info paths and auxiliary paths 
(where $c$ is the maximum number of materiality paths passing through any vertex in the graph).

\emph{The domain of $\sC^{\neg m_1}(X_0)$.}
Given that each info path $m_i$ is active given $\condset{Z_i}$, 
the contexts $\sC^{\neg m_1}(X_0)$ cannot intersect any $m_i$, except at colliders in $m_i$.
Moreover, by the definition of the control path, the only parent of $X_0$ that it contains is $Z_0$.
So, $\sC^{\neg m_1}(X_0)$ can only intersect portions of the materiality paths with domain $\bool$, 
and so the size of the domain of $\sC^{\neg m_1}(X_0)$ cannot exceed
$\lvert \dom{\sC^{\neg m_1}(X_0) \setminus Z_0}\rvert \leq 2^{bc}$, 
where $b$ is the maximum number of variables belonging to any context $C_X$, 
and $c$ is the largest number of materiality paths passing through any vertex.
\ryan{upto here.}

\emph{Proof of \Cref{eq:5}}
As the domain of $X_0$ has $\dom{X_0} \rvert \leq 2^{k+c}$,
for any particular $\sC^{\neg m_1}({X_{0}})=\sc^{\neg m_1}({X_{0})}$, there are at most $2^{k+c}$ assignments $\dom{U'_{1,1}}\subseteq \dom{U_{1,1}}$  
such that for all $u_{1,1},u'_{1,1} \in \dom{U'_{1,1}}$, $\pi_{X_1}(\sc^{\neg m_1}({X_{0})},u_{1,1}) \neq \pi_{X_1}(\sc^{\neg m_1}({X_{0})},u'_{1,1})$.
Furthermore, as $\lvert \dom{\sC^{\neg m_1}(X_0) \setminus Z_0}\rvert \leq 2^{bc}$,
by the union property, there are at most $2^{bc(k+c)}$ assignments $\mathfrak{X}'_{U_{1,1}}$ such that there exists
$\sc^{\neg m_1}({X_{0})}$ such that for all $u_{1,1},u'_{1,1} \in \mathfrak{X}'_{U_{1,1}}$, $u_{1,1},u'_{1,1} \in \dom{U'_{1,1}}$, $\pi_{X_1}(\sc^{\neg m_1}({X_{0})},u_{1,1})=\pi_{X_1}(\sc^{\neg m_1}({X_{0})},u'_{1,1})$.
However, the domain of $U_i$  is $\bool^{\exp_2^{1}(k+\lvert \sp_0 \rvert - 1)} \supseteq \bool^{2^k}$ (as $\sp_0$ contains at least $d$), so:
$$\lvert \dom{\Pa(X_0^{m_i})} \rvert \geq 2^{2^{k}} > 2^{(k+c)bc}  \geq \lvert \dom{\sC^{\neg m_1}(X_0)} \rvert \lvert \dom{X_0} \rvert,$$
where the strict inequality is from the definition of $k$ in \Cref{def:nro-model}.
So, there must exist a pair of assignments $u_{1,1},u'_{1,1}$ in the domain of $U_{1,1}$ such that for all 
$\sc^{\neg m_1}({X_{0})} \in \dom{\sC^{\neg m_1}(X_0)}$, $\pi_{X_1}(\sc^{\neg m_1}({X_{0})},u_{1,1}) = \pi_{X_1}(\sc^{\neg m_1}({X_{0})},u'_{1,1})$.
This satisfies \Cref{eq:5}, which by \Cref{le:info-transmission-5} proves the result.
\end{proof}
\section{Supplementary proofs for Section \ref{sec:further-steps} (Proof of Lemma~\ref{le:system-exists-general})}
\subsection{Proving the existence of paths} \label{app:proving-path-existence}
In this section, we will prove that when LB-factorizability is not satisfied, 
then there exist info paths and control paths, 
a potential intermediate step toward establishing completeness of Theorem LB-2 from \citet{lee2020characterizing}.

\systemexistsgeneral*

Since we will have to establish activeness given a set of implied variables, the following lemma 
will be useful.

\begin{lemma} \label{le:activeness-given-closure-set}
Let $p$ be a path.
If 
(i) $p$ contains no non-collider in $\sN$, 
(ii) every fork variable in $p$ is not in $\lceil \sN \rceil$, and
(iii) every endpoint of $p$ that has a child along $p$ is not in $\lceil \sN \rceil$,
then $p$ contains no non-collider in $\lceil \sN \rceil$.
\end{lemma}
\begin{proof}
Write $p$ as $W_1 \pathfrom U_1 \pathto W_2 \pathfrom U_2 \ldots U_J \pathto W_{J+1}$, where possibly $W_1$ is $U_1$, and possibly $U_J$ is $W_{J+1}$.
%Every collider has a descendant in $\sN$ by (i), and $\sN \subseteq \lceil \sN \rceil$, hence is active given $\lceil \sN \rceil$.
Every $U_j$ is not in $\lceil \sN \rceil$ by (ii-iii).
Each non-collider child $V$ of any $U_j$ has a parent that is not in $\lceil \sN \rceil$, and $V \not \in \sN$ by (i), so $V \not \in \lceil \sN \rceil$.
The same is then true for the non-collider child of $V$, and so on.
Since every non-collider $V'$ in $p$ has a segment $U_j \pathto V'$ of $p$ consisting of only non-colliders, every $V' \not \in \lceil \sN \rceil$, 
and $V'$ contains no non-collider in $\lceil \sN \rceil$, proving the result.
\end{proof}

Conditions II-III of LB-factorizability require that there must exist an ordering over variables, that where certain variables are placed 
before others (i.e.\ that satisfies certain precedence relationships).
Our approach will be to encode the precedence relationships from condition III in a graph, as follows.

\begin{definition} \label{def:ordering-graph}
Let the ``ordering graph'' $\calH$ be a graph on vertices $\sZ \cup \sX' \cup \sC'$, 
with an edge $A \to B$ from each parent $A \in \Pa(B)$ of a decision $B \in \sX'$,
and an edge $B \to C$ from each decision $B \in \sX'$ to a descendant $C \in \Desc(B)$.
\end{definition}

A useful property of the ordering graph is that if a variable $V$ is downstream of a context $C$ in the ordering graph, 
then there exists a decision, that has $C$ as a context, and can influence $V$.

\begin{lemma} \label{le:descendant-in-h}
If vertex $V$ is a descendant in $\calH$ of a context $Z \in C_{\calS(\sX)}$, 
then $\calG_\calS$ contains a path $Z \to X \pathto V$, where $X \in \sX'$.
\end{lemma}
\begin{proof}
Assume that $V \in \Desc^\calH(Z)$. 
The path in $\calH$ from $Z$ begins with an edge $Z \to X$ where $X \in \sX'$, which implies that $\calG_\calS$ has an edge $Z \to X$.
The path in $\calH$ must continue from $X$ to $Z$, and since each edge $A \to B$ in $\calH$ has $B \in \Desc^{\calG_\calS}(A)$, 
it follows that $V \in \Desc^{\calG_\calS}(X)$, proving the result.
\end{proof}

It is also useful to note that the expression $\spi_{\sX'_{\prec C}}$ is unnecessary in condition II.
\begin{lemma}[Unnecessary separation in condition II] \label{le:unnecessary-separation}
Let $\sX'$ be a set of decisions, $\sZ$ be a set of variables disjoint with $\sX'$, and $\sC'$ be the set of contexts not in $\sC'$ or $\sZ$, and $\prec$ be an ordering over $\sC' \cup \sX' \cup \sZ$.
If $\spi_{\sX'_{\prec C}} \not \dsep C \mid \lceil (\sX' \cup \sC')_{\prec C} \rceil$ for some $C \in \sC'$
then $\sZ_{\prec C} \not \dsep C \mid \lceil (\sX' \cup \sC')_{\prec C}$ 
\end{lemma}
\begin{proof}
By assumption, there is a path $p$ from $\spi_{X}$ to $C$, active given $\lceil (\sX' \cup \sC')_{\prec C} \rceil$, 
for some $X \in \sX'_{\prec C}$.
The only neighbour of $\pi_{X}$ is $X$, so $p$ must terminate as $X \gets \pi_{X}$.
As $X$ is in $\sX'$, activeness given $\lceil (\sX' \cup \sC')_\prec C \rceil$ implies that $p$ terminates as $C \to X \gets \pi_X$.
Every parent of $X$ is in $\sX' \cup \sC'$ except $\sZ$.
So by truncating $p$ at $\sZ$, we have that there is a path from $\sZ_{\prec C}$ to $C$, active given $\lceil (\sX' \cup \sC')_{\prec C} \rceil$.
\end{proof}

We are now equipped to prove \Cref{le:system-exists-general}.
Recall that for $\sZ,\sX'$ to be LB-factorizable, there only needs to be one ordering $\prec$
that satisfies the precedence relationships from conditions II-III.
So the approach in our proof will be to define one such $\prec$ that satisfies 
conditions III.
Since $\sZ,\sX'$ are not LB-factorizable, that must mean that condition I or II is violated, 
which will imply the existence of paths $m,d$ in each case.
(We will use the notation $\Desc^\calH(Z_0)$ to denote the set of vertices that are descendants of $Z_0$ in 
the ordering graph $\calH$.)

\begin{proof}[Proof of \Cref{le:system-exists-general}]
Let $\prec$ be any ordering $\langle V_0,\cdots V_m, Z_0, V_{m+2}, \cdots V_M \rangle$, 
over $\sZ \cup \sX' \cup \sC$ that is topological in $\calH$
and where $V_{m+2},\cdots,V_M$ are in $\Desc^\calH(Z_0)$ whereas
$V_0 \cdots V_m$ are not.
Since $\prec$ is topological in $\calH$, Condition III is satisfied, 
and since LB factorizability is not satisfied, Condition I or II must be be violated;
we consider these cases in turn.

\textbf{Case 1: Condition I is violated.}

If Condition I is violated, there is a path $m':V_1,V_2,\cdots,V_n$ where $V_1=\spi_{\sX'}$ and $V_n=Y$, active given $\lceil \sX' \cup \sC' \rceil$.
From the definition of $\spi_X$, this path must begin as $\spi_{\sX'} \to X$ for $X \in \sX'$.
As $X$ is in the conditioning set, it must be a collider, i.e.\ $m'$ begins as $\Pi_X \to X \gets V_3$.
The only parent of $X$ that is not in the conditioning set is $Z_0$, so we have $\Pi_X \to X \gets Z_0 \upathto Y$.
We truncate $m'$ as $m:Z_0 \upathto Y$.
Since $Z_0 \to X$ satisfies condition (A) of \Cref{thm:main}, there exists some $d:X \pathto Y$, proving the result in this case.

\textbf{Case 2: Condition II is violated. Step 2.1}

The violation of condition II implies that there is an active path from some $C \in \sC'$ to $\pi_{\sX'_{\prec C}},\sZ_{\prec C}$, or $\sU'$. 
This path cannot go to $\sU'$, which was chosen to be empty.
Moreover, if there is an active path to $\pi_{\sX'_{\prec C}}$, then there is a similarly active path to $\sZ_{\prec C}$ (\Cref{le:unnecessary-separation}.
So let $m':Z_0 \upathto C'$ (where $Z_0 \prec C'$) be the path to $Z_0$, active given $\lceil (\sX' \cup \sC')_{\prec C} \rceil$.
Replace this path with a walk $w'$ with an added segment $V \pathto S \pathfrom V$ from each collider $Z$ to a 
variable $S$ in the conditioning set.
Truncate $w'$ as $Z_0 \upathto C$, where $C$ is the node in $\sC'_{\succ Z_0}$ nearest $Z_0$ along $w'$.
Then let $m$ be the path obtained from $w$ by removing all retracing segments.
Clearly $m$ is active given $\lceil (\sX' \cup \sC')_{\prec C} \rceil$
%Truncate this path as $m:Z_0 \upathto C$, where $C$ is the node in $\sC'_{\succ Z_0}$ nearest $Z_0$ along $m'$.
From $Z_0 \prec C$, it follows that $C \in \Desc^\calH(Z_0)$, so there exists a path $d:Z_0 \to X \pathto C$ for $X \in \sX'$ (\Cref{le:descendant-in-h}).

\textbf{Case 2: Condition II is violated. Step 2.2}

We will now establish that $m$ is active given $\lceil \sX' \cup \sC' \rceil$. 
Since $m$ is active given $\lceil (\sX' \cup \sC')_{\prec C} \rceil$, 
and 
$\lceil (\sX' \cup \sC') \rceil \supseteq \lceil (\sX' \cup \sC')_\prec C \rceil$, 
$m$ is active given $\lceil \sX' \cup \sC' \rceil$ at each collider.
We now prove that $m$ also contains no non-collider in $\lceil (\sX' \cup \sC')_{\prec C} \rceil$ using \Cref{le:activeness-given-closure-set}, 
by proving that the 
non-colliders are not in $(\sX' \cup \sC')$ while the
endpoints and forks are not in $\lceil (\sX' \cup \sC') \rceil$.

\emph{Step 2.2.1: no non-collider in $w$ is in $(\sX' \cup \sC')$.}

We consider three sub-cases: a non-collider in 
2.2.1.1: $(\sC' \cup \sX')_{\prec C}$, 
2.2.1.2: $\sC'_{\succ C}$, or 
2.2.1.3: $\sX'_{\succ C}$.
\emph{Sub-case 2.2.1.1: a non-collider in $(\sC' \cup \sX')_{\prec C}$.}
As $w$ is active given $\lceil (\sX' \cup \sC')_\prec C \rceil$, $w$ does not contain a non-collider in $(\sC' \cup \sX')_{\prec C}$.
\emph{Sub-case 2.2.1.2: a non-collider in $\sC'_{\succ C}$.} Moreover, the definition of $C$ implies that $m$ cannot contain a non-collider in $\sC'_{\succ C}$.
\emph{Sub-case 2.2.1.3: a non-collider in $\sX'_{\succ C}$.} Finally, $w$ cannot contain any non-collider $X \in \sX'_{\succ C}$, because 
being a vertex being a non-collider in any path implies that it is an ancestor of a collider or an endpoint of that path, but
being an ancestor of a collider or an endpoint of $w$ implies $X \prec C$, which is a contradiction.
If $X$ is an ancestor of the endpoint $C$, then by the definition of $\calH$, $X \prec C$, which contradicts $X \in \sX'_{\succ C}$.
If $X$ is an ancestor of the other endpoint $Z_0$, then $X \prec Z_0$ by the definition of $\calH$, and so $X \prec C$, implying a contradiction once again.
If $X$ is an ancestor of a collider $V$, then by activeness, the collider must have a descendant $V'$ in $\lceil (\sX' \cup \sC')_{\prec C}$, 
and so $X$ is an ancestor of $V'$. 
By the definition of $\calH$, it follows that $X \prec V'$, and since $V' \prec C$, we have $X \prec C$.
Since no non-collider in $w$ is in $(\sX' \cup \sC')$, it also follows that no non-collider in $m$ is in $(\sX' \cup \sC')$.

\emph{Step 2.2.2: no endpoint of $m$ is in $\lceil (\sX' \cup \sC') \rceil$.}

The endpoint $Z_0$ cannot be in $\lceil (\sX' \cup \sC')_{\prec C} \rceil$ because $Z_0 \in \sZ$, and $\sZ$ is disjoint from $\sX'$ and $\sC'$.
The endpoint $C$ cannot be in $\lceil (\sX' \cup \sC')_{\prec C} \rceil$ because we cannot have $C \prec C$.

\emph{Step 2.2.3: If no non-collider in $(\sX' \cup \sC')$ then no fork in $\lceil \sX' \cup \sC' \rceil \setminus \sX' \cup \sC'$.}

Assume that a fork $V$ in $\lceil \sX' \cup \sC' \rceil \setminus \sX' \cup \sC'$ is in $m$, and we will prove a contradiction.
The vertex $V$ must not be in $\lceil (\sX' \cup \sC')_{\prec C} \rceil$, since $m'$ is active given $\lceil (\sX' \cup \sC')_{\prec C} \rceil$.
As $V$ is in $\lceil \sX' \cup \sC' \rceil \setminus \lceil (\sX' \cup \sC')_{\prec C} \rceil$, 
$V$ must in $\calG_\calS$ have an ancestor $A \in (\sX' \cup \sC')_{\succ C}$.
Since $Z_0 \prec C$, $V$ this ancestor $A$ also has $Z_0 \prec A$.
So, $A \in \Desc^\calH(Z_0)$ by the definition of $\prec$, 
and $A \in \Desc^\calG(Z_0)$ by the definition of $\calH$, 
and $V \in \Desc^\calG(Z_0)$, since $A$ is an ancestor of $V$.

Any fork in a path must either be an ancestor of the initial endpoint (in this case $Z$), 
or an ancestor of a collider in the path.
Since $V \in \Desc^\calG(Z_0)$ and $V$ is a fork, not an endpoint, $V$ cannot be an ancestor of the initial endpoint.
So $V$ must be an ancestor of a collider in the walk $w$.
As $w$ is active given $\lceil (\sX' \cup \sC')_\prec C \rceil$, the collider $D$ must be in $\lceil (\sX' \cup \sC')_\prec C \rceil$.
We consider three sub-cases:
2.2.3.1: $D$ is in $\lceil (\sX' \cup \sC')_\prec C \rceil \setminus (\sX' \cup \sC')$,
2.2.3.2: $D$ is in $\sX'_{\prec C}$,
2.2.3.3: $D$ is in $\sC'_{\prec C}$, 
and will prove a contradiction in each case.
\emph{Sub-case 2.2.3.1: $D$ is in $\lceil (\sX' \cup \sC')_\prec C \rceil \setminus (\sX' \cup \sC')$.}
Then all the parents of $\lceil (\sX' \cup \sC')_\prec C$ must also be in $\lceil (\sX' \cup \sC')_\prec C$ by the definition of implied variables, 
and these parents would be non-colliders, which would make $w$ blocked given $\lceil (\sX' \cup \sC')_\prec C$, giving a contradiction.
\emph{Sub-case 2.2.3.2: $D$ is in $\sX'_{\prec C}$.} 
Then at least one parent of $D$ must be a non-collider in $\sC'_{\prec C}$, which contradicts 
the statement that $w$ contains no non-collider in $(\sX' \cup \sC')$.
\emph{Sub-case 2.2.3.3: $D$ is in $\sC'_{\prec C}$.} 
Then $D \in \Desc^\calG(Z_0)$ (since $D \in \Desc^\calG(V)$ and $V \in \Desc^\calG(Z_0)$).
It follows that $Z_0 \prec D$, but this contradicts the definition of $C$ as the nearest variable along $w$ to $Z_0$ that is in $\sC'_{\succ Z_0}$.

From \Cref{le:activeness-given-closure-set} the result follows.
\end{proof}

\end{document}